\documentclass[11pt,letterpaper]{article}
\usepackage[T1]{fontenc}
\usepackage{microtype}
\usepackage[margin=1in]{geometry}
\usepackage{amsmath,amssymb,amsthm,mathtools,mathrsfs}
\usepackage{libertinus}
\ifPDFTeX
  \usepackage{libertinust1math}
\fi
\usepackage{enumitem,booktabs,longtable,array}
\usepackage[dvipsnames]{xcolor}
\usepackage[unicode,colorlinks=true,linkcolor=Rhodamine,citecolor=Rhodamine,filecolor=Rhodamine,urlcolor=blue,bookmarksnumbered=true,hypertexnames=false]{hyperref}
\newtheorem{theorem}{Theorem}[section]
\newtheorem{lemma}[theorem]{Lemma}
\newtheorem{proposition}[theorem]{Proposition}
\newtheorem{corollary}[theorem]{Corollary}

\theoremstyle{remark}
\newtheorem{remark}[theorem]{Remark}
\newcommand{\ket}[1]{\lvert#1\rangle}
\newcommand{\bra}[1]{\langle#1\rvert}
\newcommand{\norm}[1]{\left\lVert#1\right\rVert}
\newcommand{\E}{\mathbb E}
\newcommand{\poly}{\operatorname{poly}}
\newcommand{\negl}{\operatorname{negl}}
\newcommand{\QMA}{\mathsf{QMA}}
\newcommand{\PH}{\mathsf{PH}}
\newcommand{\Haar}{\operatorname{Haar}}
\DeclareMathOperator{\Tr}{Tr}
\allowdisplaybreaks[2]
\setlist[itemize]{leftmargin=*,itemsep=3pt,topsep=5pt}
\setlist[enumerate]{leftmargin=*,itemsep=3pt,topsep=5pt}
\theoremstyle{definition}
\newtheorem{definition}[theorem]{Definition}
\newtheorem{construction}[theorem]{Construction}
\numberwithin{equation}{section}
\newcommand{\kabir}[1]{\textcolor{red}{K: #1}}

\title{How to Build Pseudorandom Unitaries in Microcrypt}
\author{Aditya Gulati\\{\small UCSB}\\{\small\texttt{adityagulati@ucsb.edu}}\and Dakshita Khurana\\{\small UIUC and NTT Research}\\{\small\texttt{dakshita@illinois.edu}}\and Kabir Tomer\\{\small UIUC}\\{\small\texttt{ktomer2@illinois.edu}}}
\date{September 27, 2026}
\hypersetup{pdftitle={How to Build Pseudorandom Unitaries in Microcrypt},pdfauthor={Aditya Gulati, Dakshita Khurana, Kabir Tomer}}
\begin{document}
\maketitle
\begin{abstract}
We provide the first evidence relative to a classical oracle that pseudorandom unitaries (PRUs) can exist without quantum-computable one-way functions.

To obtain this result, we first prove that five independent random diagonal phase layers interleaved with Hadamard transforms
\[
U=F_5HF_4HF_3HF_2HF_1
\]
form a strong PRU with security under adaptive, controlled access to the unitary, its inverse, transpose, and complex conjugate.

Our main result is that, when the phase layers are implemented using a classical random oracle \(O\), the construction remains secure against uniform $\mathsf{BQP}$ adversaries with coherent access to suitable Boolean completions of classical-input \(\mathsf{QMA}^{\mathsf{PH}^{O}}\) decision problems.
Such an adversary can invert (even quantum-computable) one-way functions and invert efficiently verifiable one-way puzzles defined relative to \(O\). Our result therefore provides classical-oracle evidence that PRUs do not imply quantum-computable one-way functions, in the form of a black-box separation.

We substantially strengthen the classical-oracle result of Kretschmer, Qian, and Tal (STOC 2025) by obtaining PRUs secure against these \(\mathsf{QMA}\)-aided adversaries. Our construction also consists of a simple, efficient circuit that queries a random function, suggesting a route to concrete implementation via the random-oracle heuristic.

\end{abstract}
\clearpage
\begingroup\footnotesize\setlength{\parskip}{0pt}\tableofcontents\endgroup
\clearpage
\section{Introduction}\label{sec:introduction}

Pseudorandom unitaries (PRUs) are keyed families of efficient quantum circuits that no polynomial-time quantum algorithm can distinguish from Haar-random unitaries, even under adaptive queries~\cite{JLS}. They are a powerful primitive in quantum cryptography: they imply pseudorandom states, one-way state generators, and one-way puzzles, which underpin constructions of commitments, secure computation, and private-key quantum money~\cite{JLS,AQY,MY,KT24}. They also form a useful working model for how Haar-random behavior emerges in physical dynamics~\cite{MH, SHH24, SMLBH}.

Remarkably, they may exist even if one-way functions (OWFs) do not! Relative to a quantum oracle, PRUs have been shown to exist even when $\mathsf{BQP} = \mathsf{QMA}$~\cite{Kretschmer21}, and relative to a classical oracle, even when $\mathsf{P} = \mathsf{NP}$~\cite{KQT25}. 
These separations establish that PRUs, and hence other types of quantum cryptography, can be based on assumptions {\em weaker than} one-way functions. A central challenge in the area is to turn this possibility into {\em concrete} constructions.

Can these oracle separations guide us toward concrete PRUs in the plain model whose security does not require even quantum-computable one-way functions?
Prior works fall short of providing such guidance.
In Kretschmer's \cite{Kretschmer21} construction, the PRU is supplied directly by a quantum Haar random oracle, leaving  it unclear how to realize it by an explicit circuit.
Replacing this oracle with efficiently sampled random circuits suggests a natural candidate, but the oracle separation does not establish whether such a replacement preserves security in the absence of one-way functions. 
The PRU of Kretschmer, Qian, and Tal~\cite{KQT25} is obtained from a quantum-computable OWF that exists relative to their oracle, so it also essentially relies on one-wayness. Moreover, as the authors themselves note, it is unclear how to heuristically instantiate their construction in the plain model.

Moreover, all existing provably secure constructions in the plain model~\cite{MPSY24, MH, SHH24, SMLBH} are built from a pseudorandom function or permutation, and their security proofs therefore rely on one-way functions. One may wonder whether these same plain model constructions plausibly remain PRUs even if $\mathsf{P} = \mathsf{NP}$, but
recent no-go results~\cite{GKT26} rule this out for several prominent examples, including CPFC~\cite{MPSY24,MH} and base strong LRFC~\cite{SMLBH}. While these attacks do not cover every known architecture, the absence of an attack alone 
provides limited evidence of security without one-way functions. To make meaningful progress, we must move beyond the absence of known attacks, and address both construction and security simultaneously:


\begin{center}
    {\em Can we identify a concrete PRU candidate that plausibly does not require OWFs?\\ Can we provide evidence for its security even without quantum-computable OWFs?}
\end{center}

\subsection{Our Results}

We give an affirmative answer by analyzing a simple circuit:
five diagonal phase layers interleaved with Hadamard transforms
($\mathsf{F}_5\mathsf{H}\mathsf{F}_4\mathsf{H}\mathsf{F}_3\mathsf{H}\mathsf{F}_2\mathsf{H}\mathsf{F}_1$), where each $\mathsf{F}_i$ denotes a keyed random function family. In their influential work introducing pseudorandom states and pseudorandom unitaries,~\cite{JLS} conjectured that a constant number of independently keyed diagonal phase layers interleaved with Hadamard transforms should yield a PRU. 
However, despite major progress in analyzing new PRU constructions over the past couple of years, security of this particular candidate remained open\footnote{An earlier version of \emph{Efficient Quantum Pseudorandomness from Hamiltonian Phase States}~\cite{BHHP25} included a security analysis of a related six-phase construction. The authors subsequently removed that analysis after identifying a bug.}. Unlike other recent candidates, this particular construction uses only (constant rounds of) diagonal randomness and fixed Hadamard transforms.

We resolve their conjecture, and prove the following stronger theorem. 

\begin{theorem} (Informal)
$\mathsf{F}_5\mathsf{H}\mathsf{F}_4\mathsf{H}\mathsf{F}_3\mathsf{H}\mathsf{F}_2\mathsf{H}\mathsf{F}_1$ is a strong PRU with security under adaptive,
controlled queries to the unitary, its inverse, transpose, and complex
conjugate.
\end{theorem}

\paragraph{Security beyond quantum-computable one-wayness.} Our main result shows that this circuit remains a strong PRU even with access to powerful $\mathsf{QMA}$ and $\mathsf{PH}$ decision oracles that depend on \(F = F_1 F_2 F_3 F_4 F_5 \). In particular, we define an oracle $O_F$ such that $\mathsf{P}^{O_F} = \mathsf{NP}^{O_F}$. Then, we prove that the PRU remains secure even against $\mathsf{QMA}^{O_F}$, and going further, also against $\mathsf{QMA}(2)^{O_F}$! We state our main theorem, informally, below.

\begin{theorem} (Informal)
$\mathsf{F}_5\mathsf{H}\mathsf{F}_4\mathsf{H}\mathsf{F}_3\mathsf{H}\mathsf{F}_2\mathsf{H}\mathsf{F}_1$  remains a strong PRU against uniform polynomial-time quantum adversaries with coherent access to \(O_F\), (some Boolean completions of) any fixed finite family of \(\mathsf{QMA}^{O_F}\) and \(\mathsf{QMA}(2)^{O_F}\) promise problems, and any fixed finite family of \(\mathsf{PH}^{F}\) languages.
\end{theorem}

\paragraph{Separating Quantum Computable One-way Functions from Pseudorandom Unitaries.}
The decision oracles we introduce can invert quantum-computable OWFs and solve efficiently verifiable one-way puzzles defined relative to \(O_F\). Yet, they do not suffice to distinguish our construction from Haar-random. This yields the first fully-black-box separation (using a classical oracle) between PRU and quantum-computable one-way functions.

Fully black-box separations are a standard tool for establishing
limitations of generic cryptographic constructions
\cite{IR89,GMR01,RTV04,HR04,BBF13}.
They rule out constructions that access the underlying primitive
only through its oracle interface and whose security reductions
also treat the attacking algorithm as a black box.
In our setting, the construction and reduction access the PRU
through its specified oracle interface, and the reduction accesses
the inverter only as an oracle.
This is weaker than a full oracle separation, which would rule
out every relativizing reduction.
The distinction between these notions is discussed in
\cite{GMR01,RTV04,HR04,BBF13}.

\paragraph{Concrete Instantiation.} Our construction suggests a direct plain-model candidate through the random-oracle heuristic: replace every \(F_i\) with a public hash function, with the secret key contributing the first few bits of the input to the hash. This yields an explicit circuit with oracle-model evidence for security beyond quantum-computable one-way functions. We also hope that this construction will admit instantiations from concretely chosen explicit functions $F$ that are {\em not} one-way, and nevertheless suffice for the pseudorandomness guarantee.

\subsection{AI methodology}
\label{sec:ai-methodology}

This project grew out of the authors' discussions that we recently formalized into a paper, \emph{How Not to Build Microcrypt}~\cite{GKT26}. That work gave \(\mathsf{NP}\) attacks on several prominent PRU candidates. We used those attacks to investigate which constructions could remain secure in Microcrypt and what could account for their security. Our attempts to extend the attacks repeatedly failed for circuits that alternate many diagonal operators with Hadamard or Clifford operators. This led us to revisit the alternating diagonal-Hadamard construction proposed by Ji, Liu, and Song~\cite{JLS}, as a simple candidate that had resisted our attacks and whose security remained open.

The proof that the five phase-Hadamard construction is a strong PRU was human generated; the authors brainstormed with AI about proof strategies but did not find them useful. We subsequently ``taught'' AI the proof of strong PRU security, and then used AI systems to investigate security against \(\mathsf{PH}\)- and \(\mathsf{QMA}\)-aided adversaries. At this stage, AI was instrumental in suggesting proof approaches building on Kretschmer, Qian, and Tal~\cite{KQT25} and in developing those approaches into detailed proofs. The authors spent considerable time simplifying and cleaning up the suggested approaches.

We also used AI to draft and revise substantial portions of the manuscript, with input and direction from the authors. The authors take full responsibility for the paper, including all mathematical claims, proofs, and references.

\section{Technical Overview}

\subsection{The Construction}

Let \(N=2^n\), and let \(H\) be the \(n\)-qubit Hadamard transform. Sample five independent functions \(h_j:\mathbb F_2^n\to\mathbb Z_N\), each with independent uniform values, and define
\[
 \xi_j(x)=e^{2\pi i h_j(x)/N},\qquad
 F_j\ket{x}=\xi_j(x)\ket{x}.
\]
The construction is
\[
 U=F_5HF_4HF_3HF_2HF_1.
\]

For the analysis, we split the circuit into a recording \emph{core} $C$ and two boundary \emph{scramblers} $G$ and $D$:
\[
 C=F_4HF_3HF_2,\qquad
 G=HF_1,\qquad D=F_5H,\qquad U=DCG.
\]
The core enables path-recording, while the scramblers provide a flattening property: for any polynomial-size collection of input states fixed independently of the boundary phases, with overwhelming probability the outputs of \(G\) and \(D^\dagger\) all have computational-basis measurement probabilities bounded by \(O(n/2^n)\).

\subsection{Strong PRU Security}\label{sec:overview-strong-pru}

Let random variable $S$ denote all five function tables \(S=(h_1,\ldots,h_5)\), and let $\Pr[S]$ be the probability of sampling any particular table set. 

In the proof, we add a register that the adversary cannot access. This register has one basis state $|S\rangle$ for each choice of table. The initial state of the hidden register is
\[
\sum_S\sqrt{\Pr[S]}\,|S\rangle \propto \sum_S \ket{S}
\]

Let $\mathsf{A}$ be the adversary's query register, and let $U_S$ be the unitary defined by the five tables $S$. A forward query acts as
\[
 |\psi\rangle_{\mathsf{A}}|S\rangle
 \longmapsto
 \bigl(U_S|\psi\rangle_{\mathsf{A}}\bigr)|S\rangle
\]
for every state $|\psi\rangle_{\mathsf{A}}$ and table choice $S$.
If we ignore the table register, the experiment is equivalent to sampling the tables once and reusing them for all queries.

\paragraph{The Idealized Experiment.} To compare our circuit with Haar, we consider an idealized experiment based on Ma and Huang's path-recording framework~\cite{MH}. We will run the same adversary with the same initial state and the same operations between queries. 

But as a helpful approximation to interacting with a Haar oracle, each oracle call will now apply a fixed linear operator, called the {\em recording operator} to the query register and two private record registers. 
The operator depends on the query type and current records, and is independent of the tables $S$.

Two private registers $\mathsf{L},\mathsf{R}$ hold collections of input--output pairs, denoted $L$ and $R$. Together, these collections form a \emph{record}. They distinguish phase factors from their complex conjugates, as explained below. Let $N=2^n$, and let $x,y$ denote $n$-bit strings. The record starts empty. On its first forward query, the recording operator acts as follows:
\[
 |x\rangle_{\mathsf{A}}|\emptyset,\emptyset\rangle_{\mathsf{LR}}
 \longmapsto
 \frac{1}{\sqrt N}\sum_y
 |y\rangle_{\mathsf{A}}|\{(x,y)\},\emptyset\rangle_{\mathsf{LR}}.
\]
Each term pairs an output $|y\rangle$ with the record $L=\{(x,y)\}$, $R=\emptyset$. 
A superposition of inputs gives the corresponding superposition of these output--record states.

For later queries, the recording operator includes both creation and
deletion terms. Forward creation adds $(x,y)$ to $L$, in a uniform
superposition over unused outputs $y$. Inverse creation instead adds
$(x,y)$ to $R$: here the query input is $y$ and the output is $x$.

Deletion allows the recorder to represent the cancellation $UU^\dagger=I$.
Deletion reverses creation in the opposite direction. In particular,
a forward query on $x$ can delete a pair $(x,y)$ from $R$ and output
$y$, undoing a recorded inverse transition from $y$ to $x$. This
deletion term vanishes if $R$ contains no pair with first coordinate
$x$. Similarly, an inverse query on $y$ can delete a pair $(x,y)$
from $L$ and output $x$.

As an example, consider an inverse query on $y$, starting with empty records. It
produces a superposition over $x$. In each term, the query register
holds $x$, and $R$ stores $(x,y)$, recording where that term came from.
Now apply a forward query to this entire output state. Its deletion
term returns the query register from $x$ to $y$ and removes the
accompanying pair $(x,y)$ from $R$. The contributions from all $x$
add coherently, restoring the original state with empty records:
\[
 |y\rangle|\emptyset,\emptyset\rangle
 \;\xmapsto{\mathrm{inverse}}\;
 \frac{1}{\sqrt N}\sum_x
 |x\rangle|\emptyset,\{(x,y)\}\rangle
 \;\xmapsto{\mathrm{forward}}\;
 |y\rangle|\emptyset,\emptyset\rangle.
\]
In the second step, deleting each recorded pair contributes amplitude
$1/\sqrt N$; the contributions from all $x$ add coherently to recover
the original state. 

This sequence of recording operations approximates the adversary's state when interacting with a Haar oracle~\cite{MH}.

\paragraph{The Core of the Construction.}
To prove that our construction is a strong PRU, we show that the view of the adversary interacting with our circuit is close to the same recording computation. 

We begin by examining the three central layers, which form the \emph{core} $C$:
\[
 U=(F_5H)\,C\,(HF_1),\qquad C=F_4HF_3HF_2.
\]
Its matrix entries satisfy
\begin{equation}
 C_{y,x}=\langle y | C| x\rangle = \frac{1}{\sqrt N}\xi_2(x)\xi_4(y)\widehat\xi_3(x\oplus y),
 \qquad
 \widehat\xi_3(d)=\frac{1}{\sqrt N}\sum_z(-1)^{d\cdot z}\xi_3(z).
 \label{eq:core}
\end{equation}
Here, $\widehat\xi_3$ denotes the normalized Walsh--Hadamard transform, $\oplus$ is bitwise XOR, and $d\cdot z$ is the inner product mod 2.

We first examine one phase gate $F_j$. We will imagine that the entire truth table $h_j$ is stored in a hidden quantum register. This register has one cell for each address $u\in\mathbb F_2^n$, holding the value $h_j(u)\in\mathbb Z_N$. As in the purification above, each cell starts in a uniform superposition of all possible values. Write $\omega=e^{2\pi i/N}$. The gate $F_j$ multiplies $|x\rangle$ by $\omega^{h_j(x)}$: if the cell at address $x$ holds $a$, this contributes the phase $\omega^a$ to the joint state. Thus, at a fixed query address $x$, the state of this cell changes as follows:
\[
\begin{aligned}
 \text{Initially:}\quad &\frac{1}{\sqrt N}\sum_{a\in\mathbb Z_N}|a\rangle,\\
 \text{After one application:}\quad &\frac{1}{\sqrt N}\sum_{a\in\mathbb Z_N}\omega^a|a\rangle,\\
 \text{After two applications:}\quad &\frac{1}{\sqrt N}\sum_{a\in\mathbb Z_N}\omega^{2a}|a\rangle.
\end{aligned}
\]
Each application adds one to the coefficient of $a$ in the exponent. We introduce $k\in\mathbb Z_N$ to denote this coefficient and write
\[
 |\widetilde{k}\rangle:=\frac{1}{\sqrt N}\sum_{a\in\mathbb Z_N}\omega^{ka}|a\rangle.
\]
Here $a$ labels the stored function value, whereas $k$ labels the accumulated phase pattern. A further application of $F_j$ at $x$ sends $|\widetilde{k}\rangle$ to $|\widetilde{k+1}\rangle$ in that cell. An application of $F_j^\dagger$ multiplies by $\omega^{-a}$ and sends it to $|\widetilde{k-1}\rangle$. We can therefore interpret $k$ as a signed counter, modulo $N$, of these applications. The states $\{|\widetilde{k}\rangle:k\in\mathbb Z_N\}$ form the Fourier basis of a cell. For a superposition of query addresses, the corresponding cells are updated coherently. 

For each $j$, let $\mathbf k_j$ be the vector of these Fourier labels, one for each cell of the register holding $h_j$. We write $|\mathbf k_j\rangle_j$ for the tensor product of the corresponding cell states. In this notation, $|\mathbf 0\rangle_j$ is the initial state, and $|\mathbf e_u\rangle_j$ has label $1$ at address $u$ and label $0$ everywhere else. Distinct label vectors give orthogonal states.

Now consider the first forward query to the core $C=F_4HF_3HF_2$, on input $|x\rangle_{\mathsf A}$, with all three hidden registers initially in $|\mathbf 0\rangle$. The first phase gate $F_2$ changes the second register to $|\mathbf e_x\rangle_2$. The first Hadamard creates a superposition of addresses $z$, with amplitude $(-1)^{x\cdot z}/\sqrt N$; on each such branch, $F_3$ changes the third register to $|\mathbf e_z\rangle_3$. The second Hadamard gives output $y$ with amplitude $(-1)^{z\cdot y}/\sqrt N$, and $F_4$ changes the fourth register to $|\mathbf e_y\rangle_4$. Since $(-1)^{x\cdot z}(-1)^{z\cdot y}=(-1)^{(x\oplus y)\cdot z}$, the joint output is
\[
 \frac{1}{N}\sum_{y,z}(-1)^{(x\oplus y)\cdot z}
 |y\rangle_{\mathsf A}|\mathbf e_x\rangle_2|\mathbf e_z\rangle_3|\mathbf e_y\rangle_4.
\]
To identify what the middle register contains, define
\[
 |\chi_d\rangle_3:=\frac{1}{\sqrt N}\sum_z(-1)^{d\cdot z}|\mathbf e_z\rangle_3.
\]
These states are orthonormal, because
\[
 \langle\chi_d|\chi_{d'}\rangle
 =\frac{1}{N}\sum_z(-1)^{(d\oplus d')\cdot z}
 =\delta_{d,d'}.
\]
Thus the output can be written as
\[
 \frac{1}{\sqrt N}\sum_y |y\rangle_{\mathsf A}
 |\mathbf e_x\rangle_2|\chi_{x\oplus y}\rangle_3|\mathbf e_y\rangle_4.
\]
The second hidden register identifies $x$, the fourth identifies $y$, and the third identifies $x\oplus y$ through the orthogonal states $|\chi_d\rangle$.

Thus we will say that the three registers \emph{record} the input, output, and XOR. In the middle register, the XOR is encoded by the Hadamard-transformed basis of the $\ket{\mathbb{e}}$ states.

Compare this expression with the first query of the idealized experiment. They agree under the encoding
\[
 |\{(x,y)\},\emptyset\rangle_{\mathsf{LR}}
 \longmapsto
 |\mathbf e_x\rangle_2|\chi_{x\oplus y}\rangle_3|\mathbf e_y\rangle_4,
\]
which sends distinct pairs to orthogonal hidden states.

While this mapping works well for a single query, when there are multiple queries, these same registers accumulate further updates, including cancellations from inverse queries, and we need more care. 

For example if we naively consider two forward queries, the recording registers $2$ and $4$ contain
\[
|\mathbf e_{x_1}+\mathbf e_{x_2}\rangle_2
\otimes
|\mathbf e_{y_1}+\mathbf e_{y_2}\rangle_4.
\]
This tells us which inputs and outputs occurred, but loses information about which input was paired with which output.

On the other hand, the ideal records {\em do} carry information about which input gets mapped to which output: the records
\[ L = \{(x_1, y_1), (x_2, y_2)\},~~~~L' = \{(x_1, y_2), (x_2, y_1)\}\]
are different.

This means we need to identify, in our recording registers for $h_2, h_3, h_4$, information about {\em which} $x$ value got mapped to {\em which} $y$ value after multiple queries.
This missing information is supplied by the middle register, which records the XORs of pairs of $(x,y)$ queries. 
If all candidate pairs between the recorded inputs and outputs have distinct XORs, then each recorded XOR identifies a unique pair. Under this condition, \((x,y)\) belongs to the record exactly when its XOR label \(x\oplus y\) belongs to the recorded XOR collection.

This explains how the middle register can preserve the matching information lost by the second and fourth registers. We must now specify conditions that make this reconstruction unambiguous.
We must also account for inverse queries, which decrement the counters. These decrements can cancel increments at individual addresses even when the corresponding input–output pairs differ, potentially obscuring which pairs remain in the record.
Furthermore, recovering the pairs is only one requirement for comparison with the idealized experiment.

We must also encode superpositions of records faithfully, not just identify individual records. Different ideal records are orthogonal, and mapping them to overlapping hidden states can change interference between their amplitudes. We therefore need to show that the encoding approximately preserves norms and inner products on the relevant space of record superpositions.

We must also show that the encoding remains valid after each query. Starting from an encoded ideal state, applying the core should give approximately the same result as applying the ideal recording operator first and then encoding its output. This comparison includes both the creation and deletion terms, with their amplitudes.
We next state the required conditions on records and describe the encoding used to establish these properties.

\paragraph{The input--output pairs.}
We first specify which records our encoding will represent. Recall that $L$ contains pairs corresponding to forward queries (positive pairs) and $R$ contains pairs corresponding to inverse queries (negative pairs). These signs describe their contributions to the phase counters. A positive pair $(x,y)$ contributes $+1$ at input address $x$ in the second register and output address $y$ in the fourth register. A negative pair contributes $-1$ at those addresses.

Opposite contributions can conceal information. For example, a positive pair $(x,y)$ and a negative pair $(x,y')$ cancel in the input counter at $x$, even when $y\ne y'$. The input counters alone would then fail to show that either pair was present. To avoid this and the matching ambiguity described above, we restrict attention to records satisfying three conditions:
\begin{enumerate}
\item \textbf{All input labels are distinct across $L$ and $R$, and all output labels are distinct across $L$ and $R$.} Each occupied position in a record on registers 2 and 4 therefore identifies both an address and the sign of its contribution.
\item \textbf{No XOR collisions within either sign.} Among all candidate pairs between the positive input and output sets, distinct pairs must have distinct XORs. Impose the same condition separately on the negative input and output sets.
\item \textbf{No XOR collisions across signs among recorded pairs.} For every $(x,y)\in L$ and $(x',y')\in R$, require $x\oplus y\ne x'\oplus y'$.
\end{enumerate}
These conditions explain how to reconstruct the record from its signed input, output, and XOR collections. The endpoint collections identify the available inputs and outputs for each sign. The XOR collection then determines which of those inputs and outputs belong together. We will call records satisfying these conditions \emph{good records}. Restricting to good records is part of the proof; we must later show that this restriction discards little of the computation.

We now extend the single-query encoding to several recorded pairs. Recall that each positive XOR entry with label $d$ is represented by a superposition of counter locations through the state $|\chi_d\rangle_3$; negative entries use the corresponding counter decrements. All entries contribute to the same hidden counter register, so their updates add. Consequently, a term may assign two entries to the same location, even when their XOR labels are distinct. We simplify the encoding by retaining only terms with distinct locations, including locations arising from opposite signs, and must bound the weight this discards.

This is a projection in the definition of the encoding. It introduces neither a measurement nor an additional operation into the actual circuit. We leave the projected vector unnormalized so that its lost squared norm measures how much the restriction discarded.

The proof bounds this loss by examining collisions between pairs of entries. Crucially, the estimate applies to superpositions of good records, rather than only to each record individually. For a unit vector supported on good records with at most $t$ pairs, the discarded squared norm is at most
\[
 \frac{2t(t-1)}{N}.
\]
The bound above is small when $t^2/N\ll1$. In that regime, restricting the middle locations discards little squared norm. This bound concerns the encoding of good records. We still have to show why restricting the computation to good records causes little error.

\paragraph{The role of the outer layers.}
Our encoding applies only to good records. We must also bound the error caused by discarding terms that violate the endpoint and XOR conditions. The outer transformations $G=HF_1$ and $D=F_5H$ provide the randomness used for these bounds.

The difficulty is that the outer tables are sampled once and reused. Earlier queries can produce states that depend on these tables, and the adversary uses those states to prepare later queries. Thus, we cannot simply apply a bound that assumes the input is independent of the random transformation acting on it.

We handle this by first running the adversary against the ideal recorder. This computation never uses the hidden phase tables, so the joint state of the adversary and the records is independent of them. We call this the \emph{reference computation}. The adversary can still choose queries using earlier answers; those answers now come from the ideal recorder.

We then incorporate the outer tables into the representation of these records. Using $G$ and $D$, we transform the stored input and output registers, replacing basis labels by superpositions of labels. The transformed records depend on the tables, as they should. But the reference state before this transformation is independent of them. This lets us apply the averaged collision bounds to a fixed reference state. The same tables are used at every step of the comparison.

The reference records also have a useful restriction. Regard each pair $(x,y)$ as an edge between an input vertex and an output vertex. Each creation adds an edge with at least one new endpoint, so it cannot close a cycle or duplicate an existing edge. Each deletion removes an edge. Starting from empty records, these graphs therefore have no cycles: they are \emph{forests}. In particular, the same pair cannot occur with both signs. The proof uses this property of the original records to bound unwanted collisions after transforming their labels. The bounds apply to the full superposition of records and give a small total loss for polynomially many queries.

Some matching labels must nevertheless be retained. For example, a forward query on $x$ can delete a stored negative pair $(x,y)$ and return $y$. Matching the query's input with the stored input is necessary for this operation. The outer transformations preserve the matching operation needed for deletion, including its amplitudes. We therefore retain this part while bounding the unwanted matches.

\paragraph{The complete computation.}
We must combine these comparisons for every query the adversary is allowed to make. The adversary may coherently choose among $I,U,U^\dagger,U^T,U^*$, where $U^*$ denotes entrywise complex conjugation. For example, its control register may be in a superposition of requesting a forward query and an inverse query.

It is therefore insufficient to compare the query types separately while allowing different phase conventions or incompatible encodings. A phase difference between two control branches changes their interference. We use one common encoding and compare all query types coherently, keeping the same hidden tables throughout.

We also consider a recording computation that discards disallowed components after each query. We bound its total loss of squared norm, showing that these restrictions remove little weight. Combining this bound with the core comparison controls the difference between the full states, including the hidden registers.

Combining our comparison with the ideal recorder with the recorder's approximation to Haar gives distinguishing advantage $O(t^2N^{-1/8})$ for $t$ queries when $3(2t)^{16}\le N$. The proof builds on the recorder analysis of Schuster et al.~\cite{SMLBH}, with restricted records as in Ananth et al.~\cite{ABGL}. We can replace the random tables with quantum-query-secure pseudorandom functions to get a plain-model construction.


\subsection {\texorpdfstring{Security against (suitable Boolean completions of) \(\mathsf{QMA}^{\mathsf{PH}}\) oracles}{Security against (suitable Boolean completions of) QMA with PH oracles}}

So far the random bits used to define $U$ remained completely hidden. Next, we will define a keyed version of the construction using a public random oracle $R$ that is used to define the five tables, for each key in an exponential keyspace. In other words each key
\(k\in\{0,1\}^{\kappa}\) selects a row of $R$ containing the five
phase tables that define \(U\). We call the unitary thus obtained $U_k$. The construction samples a secret random key $k$ and outputs $U_k$. Now even though the phase tables are no longer completely hidden, they can only be accessed by queries to $R$ containing the hidden key $k$, and therefore the resulting PRU can be shown to be secure by an elementary argument.

This keyed construction however allows us to greatly increase the power of the adversary without compromising security. In particular we allow the adversary to query an additional oracle $Q$ whose answers depend on $R$ and which allows the adversary to decide languages in $\mathsf{QMA}^{\mathsf{PH}^{R}}$\footnote{In fact with some tweaks even $\mathsf{QMA}(2)^{\mathsf{PH}^{R}}$}. We construct this oracle $Q$ in steps. First we define 
\(O_R=(R,B_R)\), where \(B_R\) is an auxiliary truth table
determined by \(R\). Following~\cite{KQST, KQT25}, this table essentially contains answers to all $\mathsf{NP}^{O_R}$ problems and it has been shown that 
\(\mathsf{PH}^R\subseteq\mathsf P^{O_R}\).


Next we define the oracle $Q$. Fix a universal \(\mathsf{QMA}\) verifier \(\mathcal V^{(\cdot)}\) equipped with oracle gates.
Given input a description of an oracle-aided quantum circuit $z$ and a quantum state $\rho$ the verifier executes the circuit on $\rho$ and returns its measured bit, answering the circuit's oracle queries using its own oracle. For an oracle $O$ define the \textit{optimal acceptance probability} $g_z$ as
\[
 \begin{aligned}
 g_z(O)&=\max_\rho
   \Pr[\mathcal V^O(z,\rho)\text{ accepts}]
 \end{aligned}
\] The oracle $Q$ on input $z$ outputs $1$ if $g_z(O_R) \geq 2/3$ and outputs $0$ if $g_z(O_R) \leq 1/3$. Otherwise $Q$ independently samples a random bit from a Bernoulli distribution where the accept probability depends on $g_z(O_R)$, linearly interpolating between $1$ when $g_z(O_R) = 2/3$ and $0$ when $g_z(O_R) = 1/3$. We refer to this interpolated value as the \textit{clipped score} for $z$ and define it as
\[
r_z(O)=\min\{1,\max\{0,3g_z(O)-1\}\}.
\]
The oracles $O_R$ and $Q$ are sampled at the beginning of the experiment and fixed, and the adversary is given access to $Q$. Note that $Q$ access can be used to simulate access to $O_R$.

Our goal is to show that our PRU remains secure against such adversaries. For the randomly chosen $k$, we call the row of $k$ in $R$ the \textit{challenge row}, since it is used to define the challenge oracle $U_k$ in the real experiment for PRU security. It suffices to show that the adversary cannot distinguish the following two experiments with nonnegligible advantage. In both experiments, the adversary receives the same challenge \(U_k\). In the first experiment, we define \(Q\) from \(R\). In the second experiment, \(U_k\) still uses the original challenge row, however when defining $Q$ we replace the challenge row of \(R\) with an independent uniform row \(S\)  (calling the resultant oracle $R^{k \leftarrow S}$), and then define $Q$ using the new oracle $R^{k \leftarrow S}$.
Indistinguishability suffices because the replacement makes the phase tables used by $U_k$ independent of the public oracles, and the
independent-phase theorem from Part~\ref{part:construction} then shows that $U_k$ is indistinguishable from Haar. 


Section~\ref{sec:classical-stability} first shows that replacing the challenge row cannot change the answers given by $Q$ by much. Specifically it shows that for polynomially bounded queries \(z\) and all
sufficiently large \(\kappa\),  except with probability
at most \(2^{-\kappa^2}\) over \(R\),
\[
 \mathbb E_{k,S}
 \left|r_z(O_R)-r_z(O_{R^{k\leftarrow S}})\right|
 \le 2^{-\Omega(\kappa)}
 \quad\text{simultaneously for every }z.
\]
At one address \(z\), the
absolute score difference is exactly the total variation distance
between the two Bernoulli answer distributions. 

A natural way to understand how replacement affects the answers of \(Q\)
is to use one switch for each row. Fix the bits of \(R\) outside the
rows at key length \(\kappa\). For each key \(k\), fix two candidate
rows \(S_k^{(0)},S_k^{(1)}\). Define a Boolean vector
\(\mathbf s\in\{0,1\}^K\), where \(K=2^\kappa\), that selects row
\(S_k^{(s_k)}\) at each key, resulting in the oracle $R_{\mathbf s}$. We can think of the success probability \(g_z\) as a function
of the switch bits, with \(z\), the candidate rows, and all outside
oracle bits fixed:
\[
 f_z(\mathbf s):=g_z(O_{R_{\mathbf s}}).
\]
Flipping one switch replaces one row. Our goal is to show that, for
most switch settings, flipping a uniformly random switch changes
\(f_z\) by very little.

For bounded low-degree polynomials, there exist strong bounds on their sensitivity, i.e. the number of single-bit flips that change their value by at least a
given amount can be bounded in terms of the degree, independently of
the number of variables. However, our $f_z$'s dependence on the switch bits is not straightforward and additionally involves maximizing over the set of witnesses. $f_z$ is unlikely to be low degree.

Nevertheless, we are able to use low degree polynomials to analyze $f_z$.
First, we use the recursive definition of the auxiliary oracle
\(B_R\). Its recursive queries have at most the square root of the
current input length. Thus, for polynomial-length verifier queries,
only constantly many recursive levels can depend on the rows at key
length \(\kappa\); shorter queries depend only on the bits we fixed.
Expanding these computations into ORs of ANDs shows that the bits of $B_R$ can be computed by constant-depth
circuits of size \(\exp(\operatorname{poly}(\kappa))\) on the bits of $R$, as shown in~\cite[Lemma~38]{KQST}. For our setup this gives us that
\(O_{R_{\mathbf s}}\) has circuits with the same bounds in the
switches.

Next we apply the
random-restriction theorem of Aaronson, Ingram, and
Kretschmer~\cite{AIK} (with some fairly straightforward adaptation from decision trees to polynomials) to show that, except with very small probability, randomly fixing a large fraction of bits of the switch vector causes the
truth-table bits of \(O_{R_{\mathbf s}}\) to be expressible as low degree polynomials in the remaining bits of $s$. Since a query-bounded quantum algorithm's success probability is a low degree polynomial in the oracle bits, this lets us show that the success probability is also low degree in the switch bit (up to random restriction).

If the answers of $Q$ depended on $O_R$ simply via oracle queries this would be sufficient for our bound. However, they also involve a maximum over witnesses. We now show how to deal with this issue. For each fixed witness \(\omega\), a verifier that makes at
most \(T\) quantum queries has acceptance probability of degree at
most \(2T\) in the oracle-answer bits, by the polynomial
method~\cite{BBCMW98}. Substituting the restricted answer polynomials
therefore gives a polynomial \(P_\omega\) of degree at most \(2Td\)
in the free switches. It remains an acceptance probability, so its
values lie in \([0,1]\) on every assignment of the free switches.
The same successful restriction works for every witness.

The optimal acceptance probability need not itself be a polynomial:
it maximizes over witnesses, and the best witness can change with the
switch setting. Fix a successful restriction, and write \(f(y)\) for
\(f_z\) as a function of the remaining free switches \(y\). Then
\[
 f(y)=\max_\omega P_\omega(y).
\]
At each setting \(y\), choose an optimal witness and denote its
polynomial by \(P_y\). This polynomial agrees with \(f\) at \(y\)
and lies below \(f\) everywhere. Thus, for every other setting
\(y'\),
\[
 f(y)-f(y')\le P_y(y)-P_y(y').
\]
Whenever \(f\) decreases from \(y\) to \(y'\), the fixed-witness
polynomial decreases by at least as much. In particular, from the
higher endpoint of any pair, this polynomial controls the decrease
in optimal acceptance even if a different witness becomes optimal
at the lower endpoint. Appendix~\ref{app:scalar-stability-tools}
shows that having such a bound at every input allows us to still show powerful bounds on the sensitivity of $g_z$. The main claim of Section~\ref{sec:classical-stability} follows from these bounds.

\subsection{Accounting for queries to the construction.}
Consider the state just before a query to $Q$, after at most $t$
queries to $U(S)$. Replacing the queries to $U(S)$ by path-recording gives a reference state $|v\rangle$, whose records contain at most $t$ pairs.
The encoding from Section~\ref{sec:overview-strong-pru} converts these records into a
purification of the five tables, giving a state of the form
\[
  \sum_S \sqrt{\Pr[S]}\,
  |q_v(S)\rangle\otimes|S\rangle.
\]
This approximates the actual query state, including its hidden
tables. The row in this register is the same $S$ that we insert
into $R^{k\leftarrow S}$ when comparing queries to $Q$.

In the experiment without $U(S)$ each $\ket{q_v(S)}$ is the same, hence we get equal weights on all $S$. With access to $U(S)$,
different rows can have different weights
$W_v(S)=\|q_v(S)\|^2$. These weights describe the contribution of each
row to the query and they need not equal one.

\paragraph{A weighted score change.}
Fix the baseline oracle $R$, the key $k$, and the queried verifier input $z$.
Write $r(S) = r_z(O_{R^{k\leftarrow S}})$ for its clipped score after inserting $S$, and $c = r_z(O_{R})$
for its score before replacement. Instead of just the average
score change $|r(S)-c|$, we must now bound
\[
  \mathbb{E}_S\bigl[|r(S)-c|W_v(S)\bigr].
\]
The earlier unweighted estimate does not immediately suffice:
the query could place more weight on precisely those rows where
the score changes most.

We separate the average weight $W_v(S)$ from the average score change $|r(S)-c|$ by breaking their correlation. To do this, we first fix the boundary tables $F_1,F_5$.
This fixes the boundary transformations,
while leaving the three core tables $F_2,F_3,F_4$ independent and
uniform. We average all three core tables together because they
jointly implement the core encoding.

Write $\mathbb{E}_{\mathrm{core}}$ for this conditional average.
Although the core tables are independent of one another, the
score change and recording weight both depend on the same core
sample. We break their dependence as follows
\[
\begin{aligned}
  \mathbb{E}_{\mathrm{core}}[|r-c|W_v]
  ={}&
  \mathbb{E}_{\mathrm{core}}[|r-c|]\,
  \mathbb{E}_{\mathrm{core}}[W_v]
  \\
  &+
  \operatorname{Cov}_{\mathrm{core}}(|r-c|,W_v).
\end{aligned}
\]
The first term is what we would obtain if the score change and
weight were evaluated on independent core samples. The covariance measures the correlation
caused by using the same sample for both. In particular, it captures
the possibility that larger score changes occur on rows receiving
larger query weights. We bound these two terms separately, then
average over the boundary tables to recover the original weighted
expectation.

The first term is easy. Averaging the weight over all three core
tables gives the squared norm of a subnormalised state.
The boundary rotation preserves norms, and the good-record
projection and core encoding cannot increase them. Thus
$\mathbb{E}_{\mathrm{core}}W_v\le\|v\|^2$ for every fixed boundary
sample. Averaging the first term over the boundaries therefore
recovers the unweighted score change $\mathbb{E}_S\bigl[|r(S)-c|\bigr]$, multiplied by $\|v\|^2$.

The covariance is harder because the score and query weight use
the same tables. Section~\ref{sec:weights} bounds it using the expected largest
change in the weight under a single core-bit flip, together with
a bound on the score's total decrease over such flips $|r(\mathrm{core})-r(\mathrm{core}^{\oplus k})|$. The latter
can be shown from techniques in Section~\ref{sec:classical-stability}. The new ingredient is
therefore to show that changing any one core bit has little effect
on the encoded state.

\paragraph{Changing a middle-table bit.}
Recall that a recorded pair $(x,y)$ contributes a middle factor
of the form
\[
  \sum_z H_{x\oplus y,z}\,\xi_3(z)^\sigma,
  \qquad \sigma\in\{+1,-1\}.
\]
Flipping a bit of the table word at $z_0$ changes only the
$z_0$ sum. Its change has magnitude at most
$2|H_{x\oplus y,z_0}|=2/\sqrt N$. Thus, although the phase itself
may change substantially, it appears with a small Hadamard
coefficient.

For several recorded pairs, the collision-free encoding retains
only terms assigning distinct middle locations to different pairs.
Hence, each affected term has exactly one pair responsible
for the change at $z_0$. We can group terms by that pair and remove
its middle factor. What remains is an encoding of the other pairs,
with the additional requirement that they avoid $z_0$. Each record
has at most $t$ choices of the affected pair, and every such
contribution exposes the same $O(1/\sqrt N)$ factor.

The remaining issue is interference between different records.
Here the good-record conditions are useful
again: deleting a marked pair leaves a good record, and the
remaining record together with that pair identifies the original.
Appendices~\ref{app:core-tools} and~\ref{app:derivatives} show that this costs only polynomial factors
in $t$ and $n$. 

\paragraph{Changing an endpoint-table bit.}
For $F_2$, changing a bit at $x$ directly changes an endpoint phase,
so the $F_3$ argument does not apply on its own.
Instead, the Hadamard between $F_1$ and $F_2$ supplies the small
factor. The input labels in our good records came from applying
the boundary scrambler to a reference state.
Tracing an affected input $x$ back to a source label $p$ exposes
the coefficient $H_{x,p}\xi_1(p)^\sigma$, again of magnitude
$1/\sqrt N$.

Hence, we can do a similar analysis as $F_3$ here. Similarly, we can bound the bit flips for $F_4$ by analysing the Hadamard between $F_4$ and $F_5$.

\paragraph{Completing the weighted estimate.}
Together, these arguments bound the largest change in the encoded
vector over all core-bit flips, in expected squared norm, by
$\operatorname{poly}(t,n)\|v\|^2/N$. The expectation includes all
five tables. Passing from vectors to their squared norms gives an
expected largest weight change of at most
$\operatorname{poly}(t,n)\|v\|^2/\sqrt N$.

Section~\ref{sec:weights} combines this with the score-variation bound, controlling
the rare exceptional rows separately, to make the covariance small.
Hence we get 
\[
  \mathbb{E}_S\bigl[|r(S)-c|W_v(S)\bigr]
  \le
  \left(
    \mathbb{E}_S|r(S)-c|+\lambda_t
  \right)\|v\|^2,
  \qquad
  \lambda_t\le
  \frac{\operatorname{poly}(\kappa)}{\sqrt N}.
\]
Averaging over the key $k$, this is reduced to the earlier unweighted stability estimate.
Section~\ref{sec:erasure} converts this bound into a per query error between querying the oracle $O_{R^{k\leftarrow S}}$ and $O_{R}$.

\subsection{Completing the security proof}

We now use the preceding estimates in a hybrid argument.
Sample a random oracle $R$, a uniform key $k$, and an independent
uniform row $S$, and write $R'=R^{k\leftarrow S}$.
In both experiments, the challenge is $U(S)$.
In the first, the public completion is sampled according to
$\nu_{R'}$; in the second, it is sampled according to $\nu_R$.
The first experiment has the real distribution, since the
challenge uses row $k$ of the public random oracle.
In the second, the public oracle is independent of the
challenge tables. Once we establish indistinguishability of
these experiments, we can use the proof developed in the first part to prove indistinguishability from Haar.

\paragraph{Changing one public query at a time.}
We replace the public queries one at a time, always using the
independent oracle before the changed query and the replaced
oracle afterwards. The preceding estimates apply to precisely
such an independent-oracle prefix, including the dependence
of its query amplitudes on $S$ through earlier challenge queries.

One additional step is needed: our estimates concern completion
scores, whereas the adversary queries a sampled Boolean table.
For the analysis, we represent this table in a private register
retained throughout the computation. This representation
preserves the correlations between repeated coherent queries
to the same table and lets us bound a query's change using the
weighted score estimate. Each replacement causes negligible
expected error. Summing over the polynomially many public
queries bounds the difference between the two experiments,
averaged over the public oracles.

\paragraph{Fixing one public oracle.}
The hybrid initially bounds the signed average distinguishing
gap. This alone is insufficient, since positive and negative
gaps for different public oracles could cancel.
We therefore apply the hybrid bound to a two-run test whose
runs share the same public oracle. This gives an exponentially
small bound on the expected squared gap.
Markov's inequality and Borel--Cantelli then imply eventual
negligible advantage for almost every fixed oracle.
Taking a countable intersection gives this guarantee
simultaneously for all uniform polynomial-time adversaries.
Section~\ref{sec:erasure} gives the formal argument.

\subsection{Extensions and the black-box consequence}

The argument described above also extends to \(\mathsf{QMA}(2)\), where the verifier receives two unentangled witnesses. The same product witness is valid before and after a change in the base oracle. Thus, the proof applies without further assumptions on the witness. The universal completion also simulates base queries, and $\mathsf{PH}^R\subseteq\mathsf{P}^{O_R}$ gives coherent \(\mathsf{PH}\) queries. These extensions apply to each fixed finite family of verifiers. Here the restriction concerns the verifiers defining the added \(\mathsf{QMA}\) and \(\mathsf{QMA}(2)\) oracles: they have access to \(O_R\), but not to the added completion oracles or the challenge unitary. The adversary itself may query all of these oracles.

To prove that PRUs do not black-box imply quantum-computable OWFs, we use our oracle to find classical inverses for any quantum computable one-way function. Since this corresponds to finding \(\mathsf{QCMA}\) witnesses, and \(\mathsf{QCMA}\) admits a search-to-decision reduction, we can use our (completion of the) \(\mathsf{QMA}\) oracle to obtain such witnesses. We need some care with the thresholds we pick for our promise gaps; we provide details in Appendix~\ref{sec:oracle-search}. 


This search inverts quantum-computable functions with classical outputs and solves efficiently verifiable classical puzzles over $O_R$. Suppose a fully black-box construction could build such a one-way function from a strong PRU. With our keyed PRU, its evaluator uses the raw oracle $R$. The completion supplies an inverter. The proposed reduction would turn this inverter into a distinguisher against the PRU, contrary to our security theorem.

Thus, in the model of Appendix~\ref{sec:oracle-search}, no uniform fully black-box construction of quantum-computable classical-output one-way functions from strong PRUs exists. This consequence does not cover functions or puzzles whose evaluation or verification uses the added completion, or challenges that are quantum states.

\part{Security of the PRU}\label{part:construction}
\section{Our construction}\label{sec:construction}

We describe the construction below.

Let \(N=2^n\).
Sample five independent functions
\(h_j:\mathbb F_2^n\to\mathbb Z_N\), each with independent uniform
values, and set
\[
 \xi_j(x)=e^{2\pi i h_j(x)/N},\]
 \[F_j\ket{x}=\xi_j(x)\ket{x},\]
 \[
 U=F_5HF_4HF_3HF_2HF_1.
\]

We replace the random functions $F_1, \ldots F_5$ by quantum-query-secure pseudorandom functions (PRFs) to obtain a plain-model construction, described in Section~\ref{sec:plain-model}.

We denote by \(S=(h_1,\ldots,h_5)\) the set of five phase tables and by \(U(S)\), the resulting unitary. 


\subsection{Interfaces and query parameters}
For a Boolean oracle \(T\), quantum access to the oracle is equivalent to accessing the gate
\[U_T: \ket{z,y}\mapsto\ket{z,y\oplus T(z)}\text{, where} y\in\{0,1\}\]

Transpose and conjugation are taken in this fixed computational basis.
Let \(P_0,\ldots,P_4\) be orthogonal projections on a selector register,
with \(\sum_{a=0}^4P_a=I\). Throughout, access to \(U\) means coherent
access to
\[
 \mathsf{Select}(U)=P_0\otimes I+P_1\otimes U+
 P_2\otimes U^\dagger+P_3\otimes U^T+P_4\otimes\overline U.
\]
We write \(U\) for this oracle interface in oracle superscripts and query
recurrences. The selector is not measured, and all five branches use
the same sampled matrix with one fixed global phase.

We use \(t\ge1\) as the query/record-size budget:
\begin{equation}
 \#\{\text{selector queries}\}\le t,\qquad
 \text{record size}\le t.
 \label{eq:selector-cutoff}
\end{equation}

Each query changes the record size by at most one, and inter-query
operations leave the record registers untouched. Thus, after \(j\) queries,
the record size is at most \(j\). For computational applications, \(n=n(\kappa)\le\poly(\kappa)\).

\subsection{Normalized records and comparison recorders}
\label{pubcore:recorders}

\paragraph{Set and relation registers.}
For every finite multiset \(M=\{u_1,\ldots,u_s\}\), write
\(m_M(u)\) for the multiplicity of \(u\) and define
\begin{equation}
 \ket{M}
 =\frac{1}{\sqrt{s!\prod_u m_M(u)!}}
   \sum_{\pi\in\mathfrak S_s}
       \ket{u_{\pi(1)},\ldots,u_{\pi(s)}}.
 \label{pubcore:record-relation-state}
\end{equation}
We use this convention for every set or multiset ket throughout.
For a set, the factor is \(1/\sqrt{s!}\). Distinct multisets give
orthogonal unit vectors, including across lengths; the empty set gives
the vacuum. Write \(\ket{M_1,M_2}=\ket{M_1}\otimes\ket{M_2}\).

A relation is a finite multiset of edges
\((x,y)\in\mathbb F_2^n\times\mathbb F_2^n\), in forward input--output
order. Applying~\eqref{pubcore:record-relation-state} with \(u_i=(x_i,y_i)\)
gives the relation registers of~\cite{MH}.
These states form an
orthonormal basis of
\(\mathcal R_s=\operatorname{Sym}^s(\mathbb C^N\otimes\mathbb C^N)\).
We use two variable-length registers, denoted \(L\) and \(R\), with joint space
\[
 \mathcal H_{LR}=\bigoplus_{\ell+r\le N}\mathcal R_\ell\otimes\mathcal R_r.
\]
Write \(\Pi_{\ell,r}\) for a fixed-length
block and \(\Pi_{\le j}\) for total size at most \(j\). 

The pair \((L,R)\) is \emph{bijective} when its union has distinct domains
and distinct images. Let \(\Pi_{\mathrm{bij}}\) be the
projector on bijective $(L,R)$ and \(\Pi_{\mathrm{bij},\le j}=\Pi_{\mathrm{bij}}\Pi_{\le j}\).

\paragraph{The unrestricted recorder.}
All the isometries below are zero on states outside their
stated domain.
For every \((L,R)\) of total size at most \(N-1\), define
\begin{align}
 V^L\ket{x}\ket{L,R}
 &=\frac{1}{\sqrt{N-\lvert\operatorname{Im}(L\cup R)\rvert}}
   \sum_{y\notin\operatorname{Im}(L\cup R)}
   \ket{y}\ket{L\cup\{(x,y)\},R},
 \label{pubcore:record-VL}\\
 V^R\ket{y}\ket{L,R}
 &=\frac{1}{\sqrt{N-\lvert\operatorname{Dom}(L\cup R)\rvert}}
   \sum_{x\notin\operatorname{Dom}(L\cup R)}
   \ket{x}\ket{L,R\cup\{(x,y)\}}.
 \label{pubcore:record-VR}
\end{align}
We interpret these as adding edges and their
adjoints as removing edges in the relations $(L,R)$. Following
\cite[Definitions~25--26]{MH}, define
\begin{equation}
 V=V^L(I-V^RV^{R\dagger})
    +(I-V^LV^{L\dagger})V^{R\dagger}.
 \label{pubcore:record-V}
\end{equation}
The two terms describe creation and coherent deletion on orthogonal
subspaces. The projections in \eqref{pubcore:record-V} ensure that \(V\)
is a partial isometry.

\paragraph{The bijective partial recorder.}
On a bijective \((L,R)\) of size \(s<N\), define
\begin{align}
 W^L\ket{x}\ket{L,R}
 &=\frac{1}{\sqrt{N-s}}
   \sum_{y\notin\operatorname{Im}(L\cup R)}
   \ket{y}\ket{L\cup\{(x,y)\},R},
 &&x\notin\operatorname{Dom}(L\cup R),
 \label{pubcore:record-WL}\\
 W^R\ket{y}\ket{L,R}
 &=\frac{1}{\sqrt{N-s}}
   \sum_{x\notin\operatorname{Dom}(L\cup R)}
   \ket{x}\ket{L,R\cup\{(x,y)\}},
 &&y\notin\operatorname{Im}(L\cup R),
 \label{pubcore:record-WR}
\end{align}
Set
\begin{equation}
 W=W^L+W^{R\dagger}.
 \label{pubcore:record-W}
\end{equation}
This is the bijective recorder of \cite[Definitions~22 and~24]{MH}.
The domain of \(W^L\) has a fresh query input, i.e. $x$ for $x\notin\operatorname{Dom}(L\cup R)$; the image of \(W^R\) has
a negative edge's domain as query input, i.e. $x$ for $x\in\operatorname{Dom}(R)$. Hence, we have that the domain of \(W^L\) and \(W^{R\dagger}\) are orthogonal. Similarly, the image of \(W^L\)
is orthogonal to the image of \(W^{R\dagger}\). Hence
\[
 (W^L)^\dagger(W^R)^\dagger=0,\qquad W^LW^R=0,
 \qquad
 W^\dagger W=(W^L)^\dagger W^L+W^R(W^R)^\dagger.
\]
The last expression is a sum of orthogonal projectors, so \(W\) is a
partial isometry. Its domain and image lie in the bijective sector.

\paragraph{The barred recorders.}
For conjugate queries, forward creation records a negative edge and
backward creation records a positive edge. On records of total size
at most \(N-1\), define
\begin{align}
 \overline V^L\ket{x}\ket{L,R}
 &=\frac{1}{\sqrt{N-\lvert\operatorname{Im}(L\cup R)\rvert}}
   \sum_{y\notin\operatorname{Im}(L\cup R)}
   \ket{y}\ket{L,R\cup\{(x,y)\}},
 \label{pubcore:record-bar-VL}\\
 \overline V^R\ket{y}\ket{L,R}
 &=\frac{1}{\sqrt{N-\lvert\operatorname{Dom}(L\cup R)\rvert}}
   \sum_{x\notin\operatorname{Dom}(L\cup R)}
   \ket{x}\ket{L\cup\{(x,y)\},R}.
 \label{pubcore:record-bar-VR}
\end{align}
On bijective records of size \(s<N\), the corresponding restricted maps are
\begin{align}
 \overline W^L\ket{x}\ket{L,R}
 &=\frac{1}{\sqrt{N-s}}
   \sum_{y\notin\operatorname{Im}(L\cup R)}
   \ket{y}\ket{L,R\cup\{(x,y)\}},
 &&x\notin\operatorname{Dom}(L\cup R),
 \label{pubcore:record-bar-WL}\\
 \overline W^R\ket{y}\ket{L,R}
 &=\frac{1}{\sqrt{N-s}}
   \sum_{x\notin\operatorname{Dom}(L\cup R)}
   \ket{x}\ket{L\cup\{(x,y)\},R},
 &&y\notin\operatorname{Im}(L\cup R).
 \label{pubcore:record-bar-WR}
\end{align}
Define
\begin{align}
 \overline V
 &=\overline V^L(I-\overline V^R\overline V^{R\dagger})
   +(I-\overline V^L\overline V^{L\dagger})\overline V^{R\dagger},
 \label{pubcore:record-bar-V}\\
 \overline W&=\overline W^L+\overline W^{R\dagger}.
 \label{pubcore:record-bar-W}
\end{align}
The bar names these directly defined recorders; it does not mean
entrywise matrix conjugation.

For a partial isometry \(Z\), write
\(\Pi_{\mathcal D(Z)}=Z^\dagger Z\) and
\(\Pi_{\mathcal I(Z)}=ZZ^\dagger\). The following properties of the ordinary
recorders are proved in~\cite[Claims~11--12, 14--15 and
Corollaries~7.1 and~8.3]{MH}. For the directly defined barred recorders,
see~\cite[Lemmas~15, 19, 21 and~23]{SMLBH}.

\begin{lemma}[Recording restrictions]
\label{pubcore:recording-restrictions}
The operators \(V^L,V^R,V,W\) and their barred counterparts are partial
isometries. They and their adjoints increase total record size by at
most one, and
\begin{equation}
 W^\dagger V=\Pi_{\mathcal D(W)},
 \qquad
 WV^\dagger=\Pi_{\mathcal I(W)}.
 \label{pubcore:record-restriction-identities}
\end{equation}
The barred operators satisfy
\begin{equation}
 \overline W^\dagger\overline V=\Pi_{\mathcal D(\overline W)},
 \qquad
 \overline W\overline V^\dagger=\Pi_{\mathcal I(\overline W)}.
 \label{pubcore:record-bar-restriction-identities}
\end{equation}
The projectors onto the domain and image of \(W\) and \(\overline W\)
commute with total record size. In particular, all four recorders and
their adjoints are contractions.
\end{lemma}

\paragraph{The common coherent recorder.}
In the branch order \(I,U,U^\dagger,U^T,\overline U\), define
\begin{equation}
\begin{aligned}
 (V_0,V_1,V_2,V_3,V_4)
   &=(I,V,V^\dagger,\overline V^\dagger,\overline V),\\
 (W_0,W_1,W_2,W_3,W_4)
   &=(I,W,W^\dagger,\overline W^\dagger,\overline W),\\
 \widehat V&=\sum_{a=0}^4P_a\otimes V_a,
 \qquad \widehat W=\sum_{a=0}^4P_a\otimes W_a.
\end{aligned}
 \label{eq:controlled-recorders}
\end{equation}

Let \(\ket{\psi_0}\) be the
initial state, including an arbitrary reference, and let \(A_1,\ldots,A_t\) be the inter-query isometries. Define
the physical and recorder trajectories by
\begin{equation}
\begin{aligned}
 \ket{\psi_0(U)}&=\ket{\psi_0},&
 \ket{v_0}&=\ket{\psi_0}\otimes\ket{\varnothing,\varnothing},\\
 \ket{\psi_j(U)}&=UA_j\ket{\psi_{j-1}(U)},&
 \ket{v_j}&=\widehat V(A_j\otimes I_{LR})\ket{v_{j-1}},\\
 \rho_j^U&=\ket{\psi_j(U)}\bra{\psi_j(U)}.
\end{aligned}
 \label{eq:selector-recorder-trajectories}
\end{equation}
For $j$ with \(1\le j\le t\). The vector \(\ket{v_t}\) may be subnormalized.
Theorem~\ref{thm:controlled-haar} proves that, for \(t<N\),
\[
 \frac12\left\|
   \E_{U\leftarrow\Haar(U(N))}
        \Bigl[\rho_t^U\Bigr]
       -\operatorname{Tr}_{LR}\bigl(\ket{v_t}\bra{v_t}\bigr)
      \right\|_1
 \le9t(t+1)N^{-1/8}.
\]

\section{Recording and scrambling: independent hidden phases}
\label{sec:independent-phase}

For our analysis it will be convenient to split the PRU into three parts:
\[
\text{The {\em core} } C=F_4HF_3HF_2
\]
\[
\text{Two {\em scramblers} } D=F_5H, \text{ and } G=HF_1
\]
Then, the unitary
\begin{equation}
 U=DCG.
 \label{eq:recording-scrambling-factorization}
\end{equation}

We call \(C\) the recording {\em core}, because this is the part of the construction that enables path-recording. We refer to \(G,D\) as scramblers because they act as ``state-flatteners'', in the sense that for any polynomial-sized set of input states, it holds that with overwhelming probability over a truly random function
$\xi$, the outputs of $G$ and $D^\dagger$ are all almost flat in the computational basis, i.e., have computational-basis measurement probabilities bounded by $O(n/2^n)$.

Define fourier coefficients
\begin{equation}
 \widehat\xi_3(d)=\frac1{\sqrt N}\sum_{z\in\mathbb F_2^n}
                       \xi_3(z)(-1)^{d\cdot z},
 \qquad Z_d=\widehat\xi_3(d),
 \label{pubcore:fourier}
\end{equation}
Then, using $C = F_4 H F_3 H F_2$,
\begin{equation}
 C_{y,x}=
    \langle y | C | x \rangle =
 \frac1{\sqrt N}\xi_2(x)\xi_4(y)  \widehat\xi_3(x\oplus y).
 \label{pubcore:entry}
\end{equation}

We compare oracle access to the construction \(U=DCG\) with a Haar-random unitary. To do this, we first show that for some good space, the central core acts as a recorder. We then bound the norm outside the good space lost under
the boundary transformations \(G=HF_1\) and \(D=F_5H\). 

\begin{theorem}[Independent-hidden-phase security]\label{thm:independent-phase}
Let \(N=2^n\), let the five hidden phase tables be independent uniform
functions to \(\mathbb Z_N\), and let \(A\) make \(t\ge1\) unitary-oracle
queries in the model of Section~\ref{sec:construction}. Assume
\begin{equation}
 3(2t)^{16}\le N.
 \label{eq:construction-smallness}
\end{equation}
Then
\begin{equation}
\begin{split}
 &\left|\E_S\Bigl[\Pr[A^{U(S)}=1]\Bigr]
       -\E_{V\leftarrow\Haar(U(N))}
                    \Bigl[\Pr[A^V=1]\Bigr]\right|\\
 &\hspace{25mm}\le\Delta_t
      =O(t^2N^{-1/8}).
\end{split}
 \label{eq:independent-phase-bound}
\end{equation}
\end{theorem}

\subsection{Good records and the common core encoding}\label{pubcore:central}

Recall the recording core \(C=F_4HF_3HF_2\). On a query basis state,
\[
 C\ket{x}
 =\frac1{\sqrt N}\sum_y
       \xi_2(x)\widehat\xi_3(x\oplus y)\xi_4(y)\ket{y},
\]
where \(\widehat\xi_3\) is normalized as in~\eqref{pubcore:fourier}.
We want to represent these phase records by a path-recording pair
\((L,R)\). The factors \(\xi_2(x)\) and \(\xi_4(y)\) record the input
\(x\) and output \(y\) in the respective phase memories. The middle
factor encodes \(d=x\oplus y\) in a Hadamard superposition of table
locations, since
\[
 \widehat\xi_3(d)=\sum_zH_{d,z}\xi_3(z).
\]
The factor \(\xi_j(u)\) has positive phase exponent, while its
conjugate \(\xi_j(u)^{-1}\) has negative phase exponent.
Thus \(C,C^T\) contribute positive exponents and
\(C^\dagger,\overline C\) contribute negative exponents. We associate
positive edges with \(L\) and negative edges with \(R\), always storing
\((x,y)\) in forward input--output. Hence we start seeing a correspondence between the database of $C$ and path-recording. The corresponding signed
endpoint and XOR sets are
\begin{align*}
 X_+&=\operatorname{Dom}(L), & Y_+&=\operatorname{Im}(L), & D_+&=\{x\oplus y:(x,y)\in L\},\\
 X_-&=\operatorname{Dom}(R), & Y_-&=\operatorname{Im}(R), & D_-&=\{x\oplus y:(x,y)\in R\}.
\end{align*}

The size of \((L,R)\) is \(s=|L|+|R|\). We next give conditions under
which these six sets determine the edges of \((L,R)\).

\paragraph{Good signed relations.}
\label{pubcore:good}
We call \((L,R)\) \emph{good} when:
\begin{enumerate}
 \item \(L\cup R\) is bijective: all domain labels are distinct and all
       image labels are distinct.
 \item The map \((x,y)\mapsto x\oplus y\) is injective on
       \(X_\sigma\times Y_\sigma\) for each sign \(\sigma\in\{+,-\}\).
 \item \(D_+\cap D_-=\varnothing\).
\end{enumerate}
To define the XOR restriction on the full relation space, index edge
\emph{occurrences} of sign \(\sigma\) by \((x_i^\sigma,y_i^\sigma)\).
Let \(P_{\mathrm{ss}}\) impose
\begin{equation}
 x_i^\sigma\oplus x_j^\sigma\ne y_k^\sigma\oplus y_\ell^\sigma
 \qquad\text{for every }i\ne j,\ k\ne\ell,
 \label{pubcore:same-sign-condition}
\end{equation}
with the two pairs of indices allowed to overlap. On the bijective
sector this is condition 2, since
\[
 x_i^\sigma\oplus y_k^\sigma=x_j^\sigma\oplus y_\ell^\sigma
 \quad\Longleftrightarrow\quad
 x_i^\sigma\oplus x_j^\sigma=y_k^\sigma\oplus y_\ell^\sigma.
\]
Let \(P_{\mathrm{dx}}\) impose \(D_+\cap D_-=\varnothing\) on actual
positive and negative edges, also on the full relation space. Set
\begin{equation}
 P=P_{\mathrm{ss}}P_{\mathrm{dx}},\qquad
 P_{\mathrm{good}}=\Pi_{\mathrm{bij}}P,\qquad P_{\mathrm{good},\le t}=P_{\mathrm{good}}\Pi_{\le t}.
 \label{pubcore:good-projector}
\end{equation}
Let \(\mathcal G_{\mathrm{good},s}\) have orthonormal basis \(\ket{L,R}\) indexed by
good signed relations of size \(s\), and let
\(\mathcal G_{\mathrm{good}}=\bigoplus_{s=0}^t\mathcal G_{\mathrm{good},s}\).

\paragraph{Recovering the relations.}
For a good record, Cartesian XOR injectivity gives
\begin{equation}
 \begin{aligned}
 L&=\{(x,y)\in X_+\times Y_+:x\oplus y\in D_+\},\\
 R&=\{(x,y)\in X_-\times Y_-:x\oplus y\in D_-\}.
 \end{aligned}
 \label{eq:good-record-recovery}
\end{equation}
Indeed, each label in \(D_\sigma\) identifies exactly one pair in
\(X_\sigma\times Y_\sigma\). Hence
\((L,R)\mapsto(X_+,X_-,Y_+,Y_-,D_+,D_-)\) is a bijection onto its image.
This is the reconstruction used by \(J_D\) below: it lets us pass
between good path records and their signed endpoint and XOR sets.

\paragraph{Restricted recording.}
On a good source \((L,R)\) of size \(s<t\), define the allowed extensions by
\begin{equation}
 \begin{aligned}
 f_+(x;L,R)&=\{y:(L\cup\{(x,y)\},R)\text{ is good}\},\\
 f_-(x;L,R)&=\{y:(L,R\cup\{(x,y)\})\text{ is good}\},\\
 g_\sigma(y;L,R)&=\{x:y\in f_\sigma(x;L,R)\},
       \qquad \sigma\in\{+,-\}.
 \end{aligned}
 \label{eq:allowed-outputs}
\end{equation}
The creation maps retain the amplitudes of the bijective recorder:
\begin{align}
 \mathsf C_{+,\to}\ket{x}\ket{L,R}
 &=\frac1{\sqrt{N-s}}\sum_{y\in f_+(x;L,R)}
       \ket{y}\ket{L\cup\{(x,y)\},R},\nonumber\\
 \mathsf C_{-,\to}\ket{x}\ket{L,R}
 &=\frac1{\sqrt{N-s}}\sum_{y\in f_-(x;L,R)}
       \ket{y}\ket{L,R\cup\{(x,y)\}},\nonumber\\
 \mathsf C_{+,\leftarrow}\ket{y}\ket{L,R}
 &=\frac1{\sqrt{N-s}}\sum_{x\in g_+(y;L,R)}
       \ket{x}\ket{L\cup\{(x,y)\},R},\nonumber\\
 \mathsf C_{-,\leftarrow}\ket{y}\ket{L,R}
 &=\frac1{\sqrt{N-s}}\sum_{x\in g_-(y;L,R)}
       \ket{x}\ket{L,R\cup\{(x,y)\}}.
 \label{eq:allowed-creations}
\end{align}
Each map is zero unless its source is good and has size less than \(t\). The five restricted branches are
\begin{equation}
\begin{array}{c|c}
 I&I_{\mathcal G_{\mathrm{good}}}\\
 U&\mathsf C_{+,\to}+\mathsf C_{-,\leftarrow}^{\dagger}\\
 U^\dagger&\mathsf C_{-,\leftarrow}+\mathsf C_{+,\to}^{\dagger}\\
 U^T&\mathsf C_{+,\leftarrow}+\mathsf C_{-,\to}^{\dagger}\\
 \overline U&\mathsf C_{-,\to}+\mathsf C_{+,\leftarrow}^{\dagger}
\end{array}
\label{eq:allowed-branches}
\end{equation}
Notice that this is equal to \(P_{\mathrm{good},\le t}W_aP_{\mathrm{good},\le t}\) on \(\mathcal G_{\mathrm{good}}\).
The returned label identifies the inserted edge and hence the source.
Thus distinct columns have disjoint support. For \(s<t\),
\[
 \|\mathsf C_{\sigma,\to}\ket{x}\ket{L,R}\|^2
   =\frac{|f_\sigma(x;L,R)|}{N-s}\le1,
 \qquad
 \|\mathsf C_{\sigma,\leftarrow}\ket{y}\ket{L,R}\|^2
   =\frac{|g_\sigma(y;L,R)|}{N-s}\le1.
\]
Each full branch is a compression of a bijective partial recorder,
so it is also a contraction. As in~\cite[Section~4.2]{ABGL}, we
restrict the allowed extensions. 

\paragraph{Purified phase tables.}
Let \(\mathcal T_N\) be the set of functions
\(\mathbb F_2^n\to\mathbb Z_N\), and put \(\omega=e^{2\pi i/N}\).
Write \(\vec f=(f_2,f_3,f_4)\in\mathcal T_N^3\) and
\(\ket{\vec f}=\ket{f_2,f_3,f_4}\). We purify the three tables as
\begin{equation}
 \ket{\Omega_{234}}
 =\frac1{N^{3N/2}}\sum_{\vec f\in\mathcal T_N^3}
       \ket{\vec f}.
 \label{eq:core-table-purification}
\end{equation}
Each \(f_j\) is a possible value of the random function \(h_j\) from
Section~\ref{sec:construction}. On this basis state, write
\(\xi_j(x)=\omega^{f_j(x)}\). A phase query acts by
\[
 \mathsf F_j\ket{x}\ket{\vec f}
 =\xi_j(x)\ket{x}\ket{\vec f},\qquad j\in\{2,3,4\}.
\]
An inverse phase query multiplies by \(\xi_j(x)^{-1}\).
We also write \(\ket{\mathbf1}=\ket{\Omega_{234}}\).

\paragraph{Collision-free encoding.}
For a good record \((L,R)\), define
\begin{align}
 \chi_{L,R}(\vec f)
 &=\prod_{(x,y)\in L}\xi_2(x)\xi_4(y)
       \prod_{(x,y)\in R}\overline{\xi_2(x)\xi_4(y)},\nonumber\\
 \mu_{L,R}^{\mathrm{cf}}(\vec f)
 &=\sum_{\substack{(z_{x,y})_{(x,y)\in L\cup R}\\
                   \text{all }z_{x,y}\text{ distinct}}}
       \prod_{(x,y)\in L}H_{x\oplus y,z_{x,y}}\xi_3(z_{x,y})
       \prod_{(x,y)\in R}H_{x\oplus y,z_{x,y}}
                                      \overline{\xi_3(z_{x,y})}.
 \label{pubcore:feature}
\end{align}
Here \(z_{x,y}\in\mathbb F_2^n\) is the middle-table location assigned
to \((x,y)\), and
\[
 H_{x\oplus y,z_{x,y}}
   =\frac{(-1)^{(x\oplus y)\cdot z_{x,y}}}{\sqrt N}.
\]
The sum assigns distinct middle locations
across both signs and does not depend on the order of the edges.
Set
\begin{equation}
 \begin{aligned}
 a_{L,R}^{\mathrm{cf}}(\vec f)&=\chi_{L,R}(\vec f)\mu_{L,R}^{\mathrm{cf}}(\vec f),\\
 T_{\mathrm{good}}\ket{L,R}=\ket{a_{L,R}^{\mathrm{cf}}}
   &:=\frac1{N^{3N/2}}\sum_{\vec f\in\mathcal T_N^3}
       a_{L,R}^{\mathrm{cf}}(\vec f)\ket{\vec f}.
 \end{aligned}
 \label{pubcore:record-map}
\end{equation}
The collision-free restriction belongs to this comparison map.
The physical core \(\mathsf C=\mathsf F_4H\mathsf F_3H\mathsf F_2\)
applies, within each table block, the matrix
\[
 C(\vec f)=\operatorname{diag}(\xi_4)H\operatorname{diag}(\xi_3)
                      H\operatorname{diag}(\xi_2).
\]
For a joint state of the query, reference, and table registers, write
\begin{equation}
 \ket{\psi}
 =\frac1{N^{3N/2}}\sum_{\vec f\in\mathcal T_N^3}
       \ket{\psi(\vec f)}\otimes\ket{\vec f}.
 \label{eq:core-sample-coefficients}
\end{equation}
Then
\[
 \begin{aligned}
 \mathsf C\ket{\psi}
 &=\frac1{N^{3N/2}}\sum_{\vec f\in\mathcal T_N^3}
     (C(\vec f)\otimes I_{\rm ref})\ket{\psi(\vec f)}
        \otimes\ket{\vec f},\\
 \|\ket{\psi}\|^2
 &=\E_{\vec f}\Bigl[\|\ket{\psi(\vec f)}\|^2\Bigr].
 \end{aligned}
\]
We retain the original norms of the coefficient vectors
\(\ket{\psi(\vec f)}\). The inverse, transpose, conjugate, and idle branches
apply \(C(\vec f)^\dagger,C(\vec f)^T,\overline{C(\vec f)},I\), respectively,
within these table blocks. Untouched registers are implicit below.

\paragraph{A common contractive encoding.}
For a good record \((L,R)\), store its endpoint sets directly:
\[
 J_D\ket{L,R}
 =\ket{X_+,X_-}\otimes\ket{D_+,D_-}
                   \otimes\ket{Y_+,Y_-}.
\]
The set-state convention~\eqref{pubcore:record-relation-state} and
reconstruction~\eqref{eq:good-record-recovery} give
\(J_D^\dagger J_D=I\).
The operator \(\mathsf H_{\mathrm{mem}}\) applies one Hadamard to each
middle-memory entry and leaves the endpoint sets unchanged. Thus
\[
 \mathsf H_{\mathrm{mem}}
 =\bigoplus_{\ell+r\le t}
   I_X\otimes H^{\otimes\ell}\otimes H^{\otimes r}\otimes I_Y.
\]
Let \(\Pi_{\mathrm{cf}}\) retain distinct middle locations across both
signs. On this subspace, define the map into all three table registers by
\begin{equation}
\begin{aligned}
 &J_{\mathrm{pack}}\bigl(\ket{X_+,X_-}\otimes
       \ket{Z_+,Z_-}\otimes\ket{Y_+,Y_-}\bigr)\\
 &\quad=\frac1{N^{3N/2}}\sum_{\vec f\in\mathcal T_N^3}
       \left(\prod_{x\in X_+}\xi_2(x)
             \prod_{x\in X_-}\overline{\xi_2(x)}\right)\\
 &\qquad\qquad\cdot
       \left(\prod_{z\in Z_+}\xi_3(z)
             \prod_{z\in Z_-}\overline{\xi_3(z)}\right)
       \left(\prod_{y\in Y_+}\xi_4(y)
             \prod_{y\in Y_-}\overline{\xi_4(y)}\right)\ket{\vec f}.
\end{aligned}
 \label{eq:all-table-packing}
\end{equation}
Set \(J_{\mathrm{pack}}=0\) on the orthogonal complement. Distinct signed
set pairs give orthogonal phase characters, so
\(J_{\mathrm{pack}}^\dagger J_{\mathrm{pack}}=\Pi_{\mathrm{cf}}\).
Thus
\[
 T_{\mathrm{good}}
 =J_{\mathrm{pack}}\Pi_{\mathrm{cf}}\mathsf H_{\mathrm{mem}}J_D.
\]
Appendix~\ref{app:core} proves this factorization and bounds its loss
by counting pairs of colliding middle locations:
\begin{equation}
 \varepsilon_{\mathrm{cf},t}:=\frac{2t(t-1)}N,\qquad
 0\preceq I-T_{\mathrm{good}}^\dagger T_{\mathrm{good}}\preceq\varepsilon_{\mathrm{cf},t}I,
 \qquad T_{\mathrm{good}}\ket{\varnothing,\varnothing}=\ket{\mathbf1}.
 \label{pubcore:contractive-encoding}
\end{equation}

\begin{proposition}[Common core encoding]
\label{pubcore:compression}
Assume \(3(2t)^{16}\le N\). Let
\((\mathsf C_0,\ldots,\mathsf C_4)
=(I,\mathsf C,\mathsf C^\dagger,\mathsf C^T,\overline{\mathsf C})\)
and \(W_{\mathrm{good},a}=P_{\mathrm{good},\le t}W_aP_{\mathrm{good},\le t}\) on the good-record space.
The single encoding \eqref{pubcore:contractive-encoding} satisfies
\begin{equation}
 \max_{0\le a\le4}\|T_{\mathrm{good}}^\dagger\mathsf C_aT_{\mathrm{good}}-W_{\mathrm{good},a}\|
 \le\eta_t,\qquad
 \eta_t:=\frac{8t^2}{\sqrt N}.
 \label{pubcore:final-compression}
\end{equation}
The estimate holds on the full query and good-record space, tensored
with any untouched reference, and for the controlled oracle with the
five branches.
\end{proposition}
Appendix~\ref{pubcore:xor-recording} proves the single-table XOR-recording
identity. Appendix~\ref{pubcore:all-branch-proof} applies it to the core
and proves the proposition.

\subsection{Boundary loss and reference support}
\label{sec:pub-boundary-details}
For \(U=DCG\), put
\((C_0,C_1,C_2,C_3,C_4)=(I,C,C^\dagger,C^T,\overline C)\).
For interface $a$, the branch is \(D_aC_aG_a\), with factors given in
Table~\ref{tab:controlled-boundaries}. The boundary tables in
\(G=HF_1\) and \(D=F_5H\) are independent and have order \(q=N>8t\).

\begin{table}[ht]
\centering
\begin{tabular}{c|cc}
Branch & Input \(G_a\) & Output \(D_a\)\\\hline
Idle & \(I\) & \(I\)\\
Forward & \(G\) & \(D\)\\
Inverse & \(D^\dagger\) & \(G^\dagger\)\\
Transpose & \(D^T\) & \(G^T\)\\
Conjugate & \(\overline G\) & \(\overline D\)
\end{tabular}
\caption{Boundary factors for the five branches of a unitary-oracle query.}
\label{tab:controlled-boundaries}
\end{table}
For boundary sample \(b\), the relation registers are rotated by
\begin{equation}
 Q_b=Q[G,D]=(G\otimes D^T)_L^{\otimes *}
          \otimes(\overline G\otimes D^\dagger)_R^{\otimes *}.
 \label{eq:common-boundary-rotation}
\end{equation}

\paragraph{Recorder-domain loss.}
For \(a=1,\ldots,4\), let \(\Pi_{\mathcal D(W_a)}=W_a^\dagger W_a\)
be the domain projector of the untruncated bijective partial recorder.
Its restriction to input lengths at most \(i\) is
\(\Pi_{\mathcal D(W_a)}\Pi_{\le i}\). The defect is the part of the
bijective sector excluded from that domain.

\begin{lemma}[All-branch boundary defects]
\label{lem:pub-boundary-defects}
For \(0\le i\le N/2\) and each nonidle branch \(a\),
\begin{equation}
 \E_b\Bigl[(G_a\otimes Q_b)^\dagger
   (\Pi_{\mathrm{bij},\le i}-\Pi_{\mathcal D(W_a)}\Pi_{\le i})
   (G_a\otimes Q_b)\Bigr]
 \preceq\gamma_i I,\qquad \gamma_i=8i\sqrt{i/N}.
 \label{eq:pub-boundary-defects}
\end{equation}
This is a whole-space operator inequality, with identities on any
untouched reference registers.
\end{lemma}

Appendix~\ref{app:prefix} proves the lemma using the exact
phase--Hadamard identities of Lemma~\ref{lem:pub-boundary-twirls}
and the ordered-register collision inequalities
\eqref{eq:pub-domain-input}--\eqref{eq:pub-image-input}.

\paragraph{Unrotated path-recording support.}
View a signed relation as a bipartite multigraph: positive edges are
oriented from domain to image, and negative edges in the reverse
direction. Let \(P_{\mathcal F}\) project onto records whose underlying
undirected multigraph is a forest: it has no cycles, including two-cycles
formed by parallel edge occurrences. The empty record is included.

\begin{lemma}[Unrotated forest support]
\label{lem:reference-forest}
Every path-recorded state from empty records lies in \(P_{\mathcal F}\),
including after accessible address projections. If the initial data and
inter-query operations are independent of the hidden row and an unused
independent index, the path-recorded states retain those independences.
\end{lemma}
\begin{proof}
Each creation factor in
\(V,V^\dagger,\overline V,\overline V^\dagger\) attaches a globally
fresh endpoint, and each adjoint factor deletes an edge. These
operations preserve forests, as do their products, sums, and coherent
direct sums. Accessible inter-query operations and address projections
leave the record registers untouched. Induction from empty records proves
the support claim. Path recording uses neither the hidden
phase tables nor the unused index, which proves independence.
\end{proof}
We use this invariant in unrotated coordinates and impose goodness
after the boundary rotation.

\paragraph{Good-relation restriction.}
Set \(P_{\mathrm{add}}=P_{\mathrm{ss}}P_{\mathrm{dx}}\) as in
\eqref{pubcore:good-projector}, so that
\(P_{\mathrm{good}}=\Pi_{\mathrm{bij}}P_{\mathrm{add}}\). Let \(P_{\mathrm{nc}}\)
project onto records with no actual edge occurring in both signs.
The next bound controls the squared norm removed by the two XOR
conditions after the boundary rotation.

\begin{lemma}[Additional-goodness mass]
\label{lem:pub-boundary-good-mass}
For every fixed vector \(\ket{v}=P_{\mathrm{nc}}\Pi_{\le i}\ket{v}\),
\[
 \E_b\Bigl[\|(I-P_{\mathrm{add}})Q_b\ket{v}\|^2\Bigr]
 \le b_i\|\ket{v}\|^2,\qquad
 b_i=\frac{9i^4}{N}\quad(i\ge1),
\]
with \(b_0=0\).
\end{lemma}

Appendix~\ref{app:prefix} proves the bound by summing the same-sign
rectangle events and the opposite-sign actual-edge XOR collisions.
Every forest-supported vector lies in \(P_{\mathrm{nc}}\), since
opposite-sign copies form a two-cycle. 

\subsection{Error propagation and Haar comparison}
\label{sec:pub-norm-loss-details}

We combine the previous error bounds to prove that the phase--Hadamard
construction is indistinguishable from Haar. In this subsection,
\(P_{\mathrm{good}}\) abbreviates \(P_{\mathrm{good},\le t}\).

\begin{proposition}[Controlled comparison for the five-phase construction]
\label{prop:controlled-recorder-transfer}
Assume \(3(2t)^{16}\le N\). Fix a unit initial vector
\(\ket{\psi_0}\), including arbitrary references, and inter-query
isometries \(A_1,\ldots,A_t\). For a phase-table sample \(S\), define
the construction and path-recorded states by
\[
 \begin{aligned}
 \ket{\Phi_0(S)}&=\ket{\psi_0},&
 \ket{v_0}&=\ket{\psi_0}\otimes\ket{\varnothing,\varnothing},\\
 \ket{\Phi_j(S)}&=U(S)A_j\ket{\Phi_{j-1}(S)},&
 \ket{v_j}&=\widehat V(A_j\otimes I_{LR})\ket{v_{j-1}}.
 \end{aligned}
\]
Let \(\ket{\Phi_m}=\sum_S\sqrt{\Pr[S]}\ket{\Phi_m(S)}\otimes\ket{S}\)
retain all five tables, let \(\ket{\Omega}\) purify the boundary tables,
and set
\[
 E_t\ket{v}=T_{\mathrm{good}}P_{\mathrm{good}}Q_b(\ket{v}\otimes\ket{\Omega}).
\]
For \(0\le m\le t\),
\begin{equation}
 \|\ket{\Phi_m}-E_t\ket{v_m}\|\le e_t=O(t^2N^{-1/8}),
 \label{eq:construction-path-state}
\end{equation}
The accessible reduced states satisfy
\begin{equation}
 \frac12\left\|\rho_m^{\mathrm{phase}}
       -\operatorname{Tr}_{LR}\bigl(\ket{v_m}\bra{v_m}\bigr)\right\|_1
 \le e_t.
 \label{eq:construction-recorder-marginal}
\end{equation}
Moreover, \(P_{\mathcal F}\ket{v_m}=\ket{v_m}\) and
\(\|\ket{v_m}\|\le1\).
\end{proposition}

\begin{proof}
Appendix~\ref{app:five-phase-transfer} applies
Lemmas~\ref{lem:controlled-overlap} and~\ref{lem:approximate-compression-transfer}
to the boundary bounds and common encoding. Equation~\eqref{eq:controlled-full-environment-transfer}
gives~\eqref{eq:construction-path-state}, and the private marginal correction
\eqref{eq:controlled-marginal-correction} gives~\eqref{eq:construction-recorder-marginal}.
Lemma~\ref{lem:reference-forest} gives forest support;
contractivity gives \(\|\ket{v_m}\|\le1\).
\end{proof}

Combining Theorem~\ref{thm:controlled-haar} with
Proposition~\ref{prop:controlled-recorder-transfer} proves
Theorem~\ref{thm:independent-phase}.

\subsection{Strong pseudorandom unitaries in the plain model}
\label{sec:plain-model}

To instantiate the construction in the plain model, let
\(\{f_k:\{0,1\}^n\to\{0,1\}^n\}\) be a quantum-query-secure
pseudorandom function family~\cite{ZhandryQRF}. Assume that
\(n=n(\kappa)\) is poly in $\kappa$. Sample five independent keys,
write \(k=(k_1,\ldots,k_5)\), and set
\begin{equation}
 F_{k_j}\ket{x}=e^{2\pi i f_{k_j}(x)/N}\ket{x},
 \qquad U_k=F_{k_5}HF_{k_4}HF_{k_3}HF_{k_2}HF_{k_1},
 \qquad N=2^n.
 \label{eq:qprf-unitary}
\end{equation}
Each output word represents an element of \(\mathbb Z_N\).

\begin{theorem}[Plain-model strong PRU]\label{thm:qprf-pru}
If quantum-query-secure pseudorandom functions exist, then
\eqref{eq:qprf-unitary} defines a strong pseudorandom unitary in the
plain model. For every uniform quantum polynomial-time adversary \(A\)
with coherent access to
\(I,U_k,U_k^\dagger,U_k^T,\overline{U_k}\),
\begin{equation}
 \left|\Pr_k[A^{U_k}(1^\kappa)=1]
 -\Pr_{V\leftarrow\Haar(U(N))}
          [A^{V}(1^\kappa)=1]\right|
 \le\operatorname{negl}(\kappa).
 \label{eq:qprf-pru-bound}
\end{equation}
Both experiments use the global-phase convention of Section~\ref{sec:construction}.
\end{theorem}

\part{\texorpdfstring{Security with correlated \(\mathsf{QMA}\) and \(\mathsf{PH}\) oracles}{Security with correlated QMA and PH oracles}}\label{part:correlated}
\section{\texorpdfstring{Security with completed \(\mathsf{QMA}\) oracles}{Security with completed QMA oracles}}\label{sec:model}

Section~\ref{sec:independent-phase} proves security for a unitary defined
by five independent random phase tables. The adversary has oracle access
to the unitary, but no separate access to these tables. We now define a
keyed construction to study security with additional oracle access.


\subsection{The oracles \texorpdfstring{\(R\) and \(B_R\)}{R and B\_R}}
Let \(R:\{0,1\}^*\to\{0,1\}\) be a uniform random oracle. For security
parameter \(\kappa\), write \(n=n(\kappa)\) and \(N=2^n\),
where \(n(\cdot)\) is polynomial-time computable and
\(\kappa\le n(\kappa)\le\poly(\kappa)\).


\begin{definition}[{\cite[Definition~18]{KQST}}]
\label{def:recursive-boolean-oracle}
Define \(B_R:\{0,1\}^*\to\{0,1\}\) inductively by input length.
For \(L\ge2\), an input \(z\in\{0,1\}^L\) represents a pair
\(\langle M,y\rangle\) with the following properties:
\begin{enumerate}
\item The binary description of \(\langle M,y\rangle\) has length
      less than \(L\).
\item \(M\) is an \(\mathsf{NP}\) oracle machine receiving input \(y\).
\item \(M\) is syntactically restricted: it takes fewer than \(L\)
      steps, and its queries to either \(R\) or \(B_R\) have
      length at most \(\lfloor\sqrt L\rfloor\).
\end{enumerate}
Set
\[
 B_R(z):=M^{R,B_R}(y),\qquad O_R:=(R,B_R),
\]
where \(M^{R,B_R}(y)\) is one exactly when some computation branch
accepts. Invalid descriptions, including strings of length at most
one, receive answer zero.
For \(L\ge2\), the bound \(\lfloor\sqrt L\rfloor<L\) ensures that
each query to \(B_R\) uses an answer already defined at a shorter length, avoiding self reference issues.
\end{definition}
By \cite[Proposition~19]{KQST},
\[
 \mathsf P^{O_R}=\mathsf{NP}^{O_R},\qquad
 \PH^R\subseteq\mathsf P^{O_R}.
\]

\begin{construction}[Keyed five-phase unitaries]
\label{constr:keyed-five-phase}
We build a keyed version of the construction in Section~\ref{sec:construction}.
For \(\ell\in[5]\), \(k\in\{0,1\}^\kappa\), and \(x\in\{0,1\}^n\),
define the \(n\)-bit string
\[
 h_\ell(k,x):=\bigl(R(\ell,k,x,1),\ldots,R(\ell,k,x,n)\bigr).
\]
Interpreting \(h_j(k,x)\) as an integer, define, for \(j\in[5]\),
\[
 \begin{aligned}
 F_{j,k}\ket{x}
   &:=e^{2\pi i(h_j(k,x))/N}\ket{x},\\
 U_k&:=F_{5,k}HF_{4,k}HF_{3,k}HF_{2,k}HF_{1,k}.
 \end{aligned}
\]
\end{construction}

\subsection{\texorpdfstring{One \(\mathsf{QMA}\) decision oracle}{One QMA decision oracle}}

Fix a universal polynomial-time \(\mathsf{QMA}\) verifier \(\mathcal V^{(\cdot)}\).
Its classical input \(z\) describes a circuit over a standard
universal gate set augmented with Boolean oracle gates. The description
lists every gate and its operand registers explicitly, and writes each
register size in unary.
Given a witness \(\rho\), the verifier executes the circuit on the
designated witness register and rejects if lengths mismatch or descriptions are malformed.
Descriptions may include padding that the verifier ignores.
Note that for a string $z$ the number of gates (and therefore number of queries), witness size, and query lengths are all at most \(|z|\).

Let \(\mathcal W_z\) be the witness density matrix space for input $z$. Define the \emph{optimal
acceptance probability} \(g_z\) and the \emph{clipped score} \(r_z\)
as functions of the oracle:
\begin{equation}
\begin{split}
 g_z(O)&=\max_{\rho\in\mathcal W_z}
   \Pr[\mathcal V^O(z,\rho)\text{ accepts}],\\
 r_z(O)&=\min\{1,\max\{0,3g_z(O)-1\}\}.
\end{split}\label{eq:boolean-completion-bias}
\end{equation}

Now consider $O_R = (R, B_R)$.  On a NO input,
\(g_z(O_R)\le1/3\) and \(r_z(O_R)=0\). On a YES input,
\(g_z(O_R)\ge2/3\) and \(r_z(O_R)=1\). Between these thresholds,
the \(\mathsf{QMA}\) promise does not specify an answer.
Conditional on \(R\), sample the bits \(Q(z)\) independently over all inputs \(z\):
\begin{equation}
 Q(z)\mid R\sim\operatorname{Bernoulli}(r_z(O_R)),\qquad
 Q\mid R\sim\nu_R:=\bigotimes_z\operatorname{Bernoulli}(r_z(O_R)).
 \label{eq:boolean-completion-law}
\end{equation}
That is, conditional on \(R\), \(Q(z)\) equals \(0\) with probability
\(1-r_z(O_R)\) and \(1\) with probability \(r_z(O_R)\).
$Q$ therefore defines a \emph{Boolean completion} oracle.
It answers every classical input and agrees with the \(\mathsf{QMA}\) promise wherever it applies.

\subsection{Main Theorem}
\begin{theorem}[\(\mathsf{QMA}\) oracle security]\label{thm:main}
For almost every pair \((R,Q)\) sampled as above, every \(\mathsf{QPT}\) adversary \(A\) satisfies 
\[
\begin{split}
\operatorname{Adv}_A^{R, Q}(\kappa):={}&
\left|\Pr_{k\leftarrow\{0,1\}^\kappa}
 [A^{Q;U_k}(1^\kappa)=1]\right.\\[-2pt]
&\left.{}-\Pr_{V\leftarrow\Haar(U(N))}
 [A^{Q;V}(1^\kappa)=1]\right|
\leq \negl(\kappa).
\end{split}
\]
\end{theorem}

Note that
$R$ and \(B_R\) queries can be easily simulated through \(Q\).
Let the \emph{row} for \(k\) consist of the five complete tables
\[
 S_k=\bigl(h_1(k,\cdot),\ldots,h_5(k,\cdot)\bigr).
\]
It contains \(5nN\) bits, indexed by
\((\ell,x,i)\in[5]\times\{0,1\}^n\times[n]\).
There are \(K=2^\kappa\) rows, indexed by the keys \(k\in\{0,1\}^\kappa\).
For any row \(S\), write \(U(S)\) for the unitary constructed from
its five tables, as in Section~\ref{sec:construction}.
Thus \(U(S_k)=U_k\). We first state the following lemma showing that replacing \(U_k\) with \(U(S)\) for an independently sampled row \(S\) is undetectable even using queries to \(Q\).

\begin{lemma}[Independent-row comparison]\label{lem:row-decoupling}
Sample a uniform random oracle \(R\), then sample \(Q\mid R\sim\nu_R\)
according to \eqref{eq:boolean-completion-law}.
For almost every fixed pair \((R,Q)\), every uniform \(\mathsf{QPT}\) adversary \(A\)
satisfies
\[
 \begin{aligned}
 &\left|\Pr_{k\leftarrow\{0,1\}^\kappa}
       [A^{Q;U_k}(1^\kappa)=1]\right.\\[-2pt]
 &\left.\qquad-\Pr_S[A^{Q;U(S)}(1^\kappa)=1]\right|
 \le\negl(\kappa).
 \end{aligned}
\]
where $S$ is a uniformly sampled row.
\end{lemma}

Combining Lemma~\ref{lem:row-decoupling} with
Theorem~\ref{thm:independent-phase} proves Theorem~\ref{thm:main}.

\begin{proof}[Proof of Theorem~\ref{thm:main}]
Fix a pair \((R,Q)\) for which Lemma~\ref{lem:row-decoupling} holds,
and fix a uniform \(\mathsf{QPT}\) adversary \(A\).
Let \(t=\poly(\kappa)\) bound its unitary queries.
With this pair fixed, the public oracle gates are independent of \(S\).
Theorem~\ref{thm:independent-phase} therefore gives
\[
 \begin{aligned}
 &\left|\Pr_S[A^{Q;U(S)}(1^\kappa)=1]\right.\\[-2pt]
 &\left.\qquad-\Pr_{V\leftarrow\Haar(U(N))}
       [A^{Q;V}(1^\kappa)=1]\right|
 \le\Delta_t=O(t^2N^{-1/8}).
 \end{aligned}
\]
Since \(n(\kappa)\ge\kappa\), this bound is negligible and the
construction condition holds eventually. The triangle inequality and
Lemma~\ref{lem:row-decoupling} prove the theorem.
\end{proof}
In fact, Section~\ref{sec:fixed-oracle-conclusion} gives
\(\operatorname{Adv}_A^{R,Q}(\kappa)=O(2^{-\kappa/128})\) eventually
under the same quantifiers, with a cutoff depending on the sampled
oracle, adversary, and its fixed family and resource bounds.
The remaining technical sections are devoted to proving
Lemma~\ref{lem:row-decoupling}.


\section{Stability of verifier scores under row resampling}
\label{sec:classical-stability}

Throughout this section, use the scores \(g_z(O)\) and \(r_z(O)\)
from \eqref{eq:boolean-completion-bias}. For each input \(z\), the
set of allowed witnesses \(\mathcal W_z\) is a nonempty compact set of density
operators on a finite-dimensional register, independent of \(O\).

\subsection{The resampling bound}\label{sec:scalar-stability}

At key length \(\kappa\), the random oracle \(R\) contains
\(K=2^\kappa\) rows, each with \(5nN\) bits. Write
\(R^{k\leftarrow S}\) for the oracle obtained by replacing row
\(k\) with \(S\).

\begin{lemma}[Simultaneous row stability]\label{lem:scalar-row-stability}
Fix a positive integer-valued polynomial \(\ell_{\rm in}\).
For each \(\kappa\), let
\(\mathcal Z_\kappa:=\{0,1\}^{\le\ell_{\rm in}(\kappa)}\).
For all sufficiently large \(\kappa\) and every fixed assignment
to the bits of \(R\) outside the \(K\) rows at key length
\(\kappa\), sample those rows independently and uniformly.
With probability at least \(1-2^{-\kappa^2}\) over these rows,
the sampled \(R\) satisfies
\begin{equation}
 \mathbb E_{k,S}
 \left|r_z(O_R)-r_z(O_{R^{k\leftarrow S}})\right|
 \le\alpha:=6\cdot2^{-\lfloor\kappa/16\rfloor}
 \quad\text{for every }z\in\mathcal Z_\kappa,
 \label{eq:scalar-stability}
\end{equation}
where \(k\) is uniform among the \(K\) keys, and \(S\) is a fresh
uniform row independent of \(R,k\).
\end{lemma}

The encoding in Section~\ref{sec:model} bounds both the query count
and each query length by \(\ell_{\rm in}(\kappa)\), for every
input in \(\mathcal Z_\kappa\), oracle, and allowed witness.
Also, \(\log_2|\mathcal Z_\kappa|\le\ell_{\rm in}(\kappa)+1\).

The proof uses two general bounds stated next. The first controls
sensitivity for a function that is a point-wise maximum of bounded low-degree polynomials.
The second bounds the average change in a low sensitivity function under a bit flip.
To apply them, we construct $R$ by fixing two candidate rows per key and use one
bit to select each row to actually include in the oracle. Existing work shows that random restriction makes every
answer of $O_R$ a low-degree polynomial in these free "switch" bits. The polynomial method composed with the above allows us to represent \(g_z\) as a maximum of bounded low-degree polynomials.
The general bounds then allow us to show that this doesn't change much when flipping a random switch bit which
gives the row stability lemma we need.

\subsection{Two general bounds}\label{sec:sensitivity}

For \(x\in\{0,1\}^m\), write \(x^{\oplus i}\) for \(x\) with
bit \(i\) flipped. For \(f:\{0,1\}^m\to[0,1]\) and
\(0<\delta\le1\), define its \emph{\(\delta\)-sensitivity} by
\[
 S_\delta(f,x)
 :=\#\{i:|f(x)-f(x^{\oplus i})|\ge\delta\}.
\]
Thus \(S_\delta(f,x)\) counts the coordinate flips that change
\(f\) by at least \(\delta\). We call these coordinates
\emph{\(\delta\)-sensitive} at \(x\).

\begin{lemma}[Sensitivity of a maximum of polynomials]
\label{lem:polynomial-max-sensitivity}
Let \(m\ge0\) and \(D\ge1\) be integers. Let \((p_a)_{a\in I}\)
be a family of real-valued polynomials on \(\{0,1\}^m\), each of
degree at most \(D\) and with values in \([0,1]\). Define the pointwise the maximum as
\[
 f(x)=\max_{a\in I}p_a(x)
\]
is for every \(x\). For every \(0<\delta\le1\) and \(t\ge0\),
\[
 \Pr_{X\sim\{0,1\}^m}[S_\delta(f,X)\ge t]
 \le\frac2\delta\,2^{-t\delta/(4\pi^2D^2)}.
\]
\end{lemma}
Appendix~\ref{app:scalar-stability-tools} proves the lemma.

Next we show a simple bound that applies to any function with values in \([0,1]\).
For \(m\ge1\), define its \emph{average change} by
\[
 A_f(x):=\frac1m\sum_{i=1}^m|f(x)-f(x^{\oplus i})|.
\]

\begin{lemma}[From sensitivity to average change]
\label{lem:switch-variation}
Let \(m\ge1\), let \(f:\{0,1\}^m\to[0,1]\), and let
\(0<\delta\le1\). For every input \(x\),
\[
 A_f(x)\le\delta+\frac{S_\delta(f,x)}m.
\]
Consequently, for a uniform input \(X\) and every \(t\ge0\),
\begin{equation}
 \Pr_X[A_f(X)>\delta+t/m]
 \le\Pr_X[S_\delta(f,X)\ge t].
 \label{eq:average-switch-stability}
\end{equation}
\end{lemma}
\begin{proof}
Each counted flip changes \(f\) by at most one; every other flip
changes it by less than \(\delta\). Sum these bounds and divide by
\(m\). The probability bound follows from the pointwise inequality.
\end{proof}

\subsection{Verifier polynomials and oracle answers}\label{sec:witnesses}

Let \(M_{\rm qry}\) be the number of oracle $O_R$ entries which can be reached by verifier queries of
length at most \(\ell_{\rm in}(\kappa)\). Then
\begin{equation}
 M_{\rm qry}=2^{O(\ell_{\rm in}(\kappa))}=2^{\poly(\kappa)}.
 \label{eq:oracle-coordinate-bound}
\end{equation}
The verifier's behaviour on any input cannot depend on any oracle coordinates beyond the first $M_{\rm qry}$ bits. We identify their
values with a vector \(a\in\{0,1\}^{M_{\rm qry}}\) for ease of notation; any verifier dependent functions of the oracle can be thought of as e just functions of the bits in $a$.

For a fixed witness \(\omega\in\mathcal W_z\), for any \(a\in\{0,1\}^{M_{\rm qry}}\) let
\(p_{z,\omega}(a)\) be the verifier's acceptance probability.
Purify the witness and defer measurements. The polynomial
method~\cite[Lemma~4.2]{BBCMW98} gives
\begin{equation}
 \deg p_{z,\omega}\le2\ell_{\rm in}(\kappa),\qquad
 0\le p_{z,\omega}(a)\le1,\qquad
 g_z(a)=\max_{\omega\in\mathcal W_z}p_{z,\omega}(a).
 \label{eq:verifier-polynomial-maximum}
\end{equation}
When \(a\) is the answer vector of \(O\), we also write
\(g_z(O)\) and \(p_{z,\omega}(O)\).

We also need to show that each bit of $O_R$ can be computed by an $\mathsf{AC^0}$ circuit representation of each answer in terms of the
raw row bits. The following is a direct consequence of
\cite[Lemma~38]{KQST}. We retain its restriction to the three core tables for
Section~\ref{sec:weights}.

\begin{lemma}[Circuits for oracle answers]
\label{lem:base-completion-locality}
Each answer \(R(x)\) or \(B_R(x)\) with
\(|x|\le\ell_{\rm in}(\kappa)\) has an \(\mathsf{AC}^0\) circuit in the bits
of the rows at \(\kappa\). The circuits have constant depth and
size \(\exp(\poly(\kappa))\), uniformly over \(x\) and the fixed
outside bits. These bounds also hold if a Boolean variable at each key
selects between two fixed candidate rows, or if any subset of raw row
bits is fixed. In particular, one may leave only a single row's three
core tables \(h_2(k,\cdot),h_3(k,\cdot),h_4(k,\cdot)\) variable.
\end{lemma}
\begin{proof}
The answer \(R(x)\) is fixed or one input bit. The row-bit encoding
in Section~\ref{sec:model} gives every address in the current rows
length at least \(\kappa\). By
Definition~\ref{def:recursive-boolean-oracle}, each recursive query
has length at most the square root of the current input length.
Thus all \(B_R\)-answers at lengths at most \(\kappa\) depend only
on fixed \(R\)-bits below length \(\kappa\), and are fixed.

Choose a constant integer \(d_{\rm rec}\ge1\) such that
\(\ell_{\rm in}(\kappa)\le\kappa^{2^{d_{\rm rec}}}\)
for sufficiently large \(\kappa\).
Apply \cite[Lemma~38]{KQST} with its length parameter
\(\ell=\kappa\) and its depth parameter \(d=d_{\rm rec}\).
Each required \(B_R(x)\) has a circuit of depth at most
\(2d_{\rm rec}\) and size \(O(2^{\kappa^{2^{d_{\rm rec}}}})\).
Its inputs are \(R\)-bits at lengths at most
\(\kappa^{2^{d_{\rm rec}-1}}\) and \(B_R\)-answers at lengths
at most \(\kappa\).
Fix the latter inputs and all \(R\)-bits outside the current rows.
The same bounds hold for every assignment of these inputs.

Under the first substitution, each row bit becomes a constant, a
variable, or its negation. This preserves the depth and size bounds.
Fixing any subset of the raw row bits, including all but one row's
three core tables, cannot increase either bound.
\end{proof}

\subsection{Row switches and random restrictions}
\label{sec:base-locality}

Fix the bits of \(R\) outside the rows at \(\kappa\). For each key
\(k\), fix two candidate rows \(S_k^{(0)},S_k^{(1)}\). A Boolean
vector \(\mathbf s\in\{0,1\}^K\) selects row \(S_k^{(s_k)}\)
at each key; write \(R_{\mathbf s}\) for the resulting oracle.
We call the coordinates of \(\mathbf s\) \emph{row switches}.
Write \(O_{\mathbf s}:=O_{R_{\mathbf s}}\), and define
\[
 \mathsf{ans}_j(\mathbf s)=(O_{\mathbf s})_j,\qquad
 f_z(\mathbf s)=g_z(O_{\mathbf s}).
\]
Flipping switch \(k\) replaces exactly row \(k\) of \(R_{\mathbf s}\).
Only the switches will be restricted; each oracle answer remains a
function of the free switches.

Choose common circuit bounds \(d_{\rm circ}\) and \(s_{\rm circ}\ge2\)
for \(\mathsf{ans}_j\), \(j\le M_{\rm qry}\), with negations at the
inputs and alternating AND/OR layers.
Lemma~\ref{lem:base-completion-locality} gives constant
\(d_{\rm circ}\) and polynomially bounded \(\log_2s_{\rm circ}\).
Choose a constant \(c_{\rm res}\), depending only on \(d_{\rm circ}\),
large enough that
\[
 \Gamma_{\rm res}
 =c_{\rm res}(1+\log_2s_{\rm circ})^{d_{\rm circ}-1}\ge1,
 \qquad p_{\rm free}=\frac1{2\Gamma_{\rm res}}
\]
satisfy the restriction bound below. These choices are uniform in the
candidate rows and outside bits.

A restriction \(\varrho\in\{0,1,*\}^K\) leaves each starred
coordinate free and fixes the others to the specified bits.
Sample its coordinates independently, with probabilities
\(p_{\rm free}\) for \(*\) and \((1-p_{\rm free})/2\) for each bit.
Write \(f|_\varrho\) for the resulting function on the free coordinates.
For any integer \(d_{\rm tree}\ge1\) and \(j\le M_{\rm qry}\),
\cite[Theorem~42]{AIK} gives
\[
 \Pr_\varrho[\mathsf{ans}_j|_\varrho
   \text{ has no decision tree of depth at most }d_{\rm tree}]
 \le(p_{\rm free}\Gamma_{\rm res})^{d_{\rm tree}}
 =2^{-d_{\rm tree}}.
\]
Call a restriction \emph{successful} if all \(M_{\rm qry}\)
restricted answer functions have such trees. A union bound gives
\begin{equation}
 \Pr_\varrho[\varrho\text{ is not successful}]
 \le M_{\rm qry}2^{-d_{\rm tree}}.
 \label{eq:restriction-failure}
\end{equation}

On a successful restriction, each answer is a polynomial of degree
at most \(d_{\rm tree}\) in the free switches. Indeed, a decision
tree computes the sum, over its leaves labelled one, of the products
that specify their paths. Each product has at most \(d_{\rm tree}\)
factors. Substitute these answer polynomials into
\(p_{z,\omega}\) from \eqref{eq:verifier-polynomial-maximum}.
The resulting polynomial \(P^\varrho_{z,\omega}\) satisfies, for every
assignment \(\mathbf s_{\rm free}\) to the free switches,
\begin{equation}
 \begin{aligned}
 \deg P^\varrho_{z,\omega}\le
 D:=2\ell_{\rm in}(\kappa)d_{\rm tree},&\qquad
 0\le P^\varrho_{z,\omega}(\mathbf s_{\rm free})\le1,\\
 (f_z|_\varrho)(\mathbf s_{\rm free})
 &=\max_{\omega\in\mathcal W_z}
 P^\varrho_{z,\omega}(\mathbf s_{\rm free}).
 \end{aligned}
 \label{eq:restricted-polynomial-maximum}
\end{equation}
These properties hold on the entire cube of free switches: each
assignment gives Boolean oracle answers, and the composed polynomial
is their acceptance probability.

The next lemma applies the sensitivity bound on the free switches,
then returns the estimate to all \(K\) switches.

\begin{lemma}[Sensitivity of optimal acceptance]\label{lem:sensitivity}
Fix \(z\in\mathcal Z_\kappa\), the candidate rows, and the outside
bits. For \(0<\delta\le1\) and an integer \(d_{\rm tree}\ge3\), set
\begin{equation}
 D=2\ell_{\rm in}(\kappa)d_{\rm tree},\qquad
 u_{\rm free}=
 \left\lceil\frac{4\pi^2D^2d_{\rm tree}}{\delta}\right\rceil,
 \qquad
 s_*=\left\lceil\frac{2u_{\rm free}}{p_{\rm free}}\right\rceil.
 \label{eq:sensitivity-threshold}
\end{equation}
For a uniform switch vector \(\mathbf s\),
\begin{equation}
 \Pr_{\mathbf s}[S_\delta(f_z,\mathbf s)\ge s_*]
 \le2\left(M_{\rm qry}+\frac2\delta\right)2^{-d_{\rm tree}}.
 \label{eq:sensitivity-tail}
\end{equation}
\end{lemma}
\begin{proof}
We prove two claims. First, after restriction the sensitivity reaches
\(u_{\rm free}\) with probability at most
\((M_{\rm qry}+2/\delta)2^{-d_{\rm tree}}\).
Second, at any input with full sensitivity at least \(s_*\), restricted
sensitivity reaches \(u_{\rm free}\) with probability at least \(1/2\).
Together these claims give \eqref{eq:sensitivity-tail}.

\emph{Claim 1: restricted sensitivity.}
Sample \(\varrho\), then a uniform assignment
\(\mathbf s_{\rm free}\) to its free coordinates. On each successful
restriction, \eqref{eq:restricted-polynomial-maximum} permits
Lemma~\ref{lem:polynomial-max-sensitivity}. It gives
\[
 \Pr_{\mathbf s_{\rm free}}[
   S_\delta(f_z|_\varrho,\mathbf s_{\rm free})\ge u_{\rm free}]
 \le\frac2\delta\,2^{-u_{\rm free}\delta/(4\pi^2D^2)}
 \le\frac2\delta\,2^{-d_{\rm tree}}.
\]
Adding the failure probability in \eqref{eq:restriction-failure}
proves Claim~1.

\emph{Claim 2: retaining sensitive coordinates.}
The same joint distribution is obtained by first sampling uniform
\(\mathbf s\), then retaining each coordinate independently with
probability \(p_{\rm free}\) and fixing the others to their values
in \(\mathbf s\). For a specified free set \(J\) and full vector
\(\mathbf s\), both procedures assign probability
\(p_{\rm free}^{|J|}(1-p_{\rm free})^{K-|J|}2^{-K}\).
We use this description without conditioning on restriction success.

For each \(k\in J\), both \(\mathbf s\) and
\(\mathbf s^{\oplus k}\) agree with the fixed coordinates.
Their function values do not change under restriction. Thus
\[
 S_\delta(f_z|_\varrho,\mathbf s_{\rm free})
 =\#\{k\in J:|f_z(\mathbf s)-f_z(\mathbf s^{\oplus k})|
                   \ge\delta\}.
\]
Conditional on \(\mathbf s\), this count is binomial with mean
\(p_{\rm free}S_\delta(f_z,\mathbf s)\).
If \(S_\delta(f_z,\mathbf s)\ge s_*\), the mean is at least
\(2u_{\rm free}\). A lower-tail Chernoff bound gives probability
at most \(e^{-u_{\rm free}/4}<1/2\) that the count is below
\(u_{\rm free}\). This proves Claim~2.

Claim~2 bounds the probability in Claim~1 from below by
\(\frac12\Pr_{\mathbf s}[S_\delta(f_z,\mathbf s)\ge s_*]\).
Claim~1 therefore proves the lemma.
\end{proof}

\subsection{Independent row resampling}\label{sec:scalar-bound}

The sensitivity bound holds for every fixed set of candidate rows.
If both candidates are independent uniform rows, the selected and
unselected rows are also independent and uniform. This gives
independent row replacement. The next lemma makes this step precise.

\begin{lemma}[From row switches to independent replacements]
\label{lem:switch-to-resampling}
Fix the outside bits and a finite input set \(\mathcal Z_\kappa\).
Suppose that for every choice of two candidate rows per key and every
\(z\in\mathcal Z_\kappa\),
\[
 \Pr_{\mathbf s}\!\left[
   \frac1K\sum_k|r_z(O_{\mathbf s})
                -r_z(O_{\mathbf s^{\oplus k}})|>a\right]
 \le q_{\rm bad},
\]
where \(\mathbf s\) is uniform, \(a\ge0\), and \(q_{\rm bad}\ge0\).
If the rows of \(R\) are independent and uniform, then for every
\(b>0\),
\begin{equation}
 \Pr_R\!\left[\exists z\in\mathcal Z_\kappa:
    \mathbb E_{k,S}|r_z(O_R)-r_z(O_{R^{k\leftarrow S}})|>a+b\right]
 \le\frac{|\mathcal Z_\kappa|q_{\rm bad}}b.
 \label{eq:resampling-conversion}
\end{equation}
\end{lemma}
\begin{proof}
Sample all candidate rows \(S_k^{(0)},S_k^{(1)}\) independently and
uniformly, and sample an independent uniform switch vector \(\mathbf s\).
Form \(R=R_{\mathbf s}\), and let \(T_k=S_k^{(1-s_k)}\) be the
unused row at key \(k\). For each fixed \(\mathbf s\), this only
permutes the two independent samples at each key. Thus the rows of
\(R\) and \(T=(T_k)_k\) are jointly independent and uniform.
Flipping switch \(k\) replaces row \(k\) of \(R\) by \(T_k\).
Averaging the hypothesis over the candidate rows gives
\[
 \Pr_{R,T}[Z_z(R,T)>a]\le q_{\rm bad},\qquad
 Z_z(R,T)=\frac1K\sum_k
       |r_z(O_R)-r_z(O_{R^{k\leftarrow T_k}})|.
\]

For fixed \(R\), define
\[
 \operatorname{stab}_z(R)
    =\mathbb E_{k,S}|r_z(O_R)-r_z(O_{R^{k\leftarrow S}})|,
 \qquad
 q_z(R)=\Pr_T[Z_z(R,T)>a].
\]
The replacement distribution gives
\(\operatorname{stab}_z(R)=\mathbb E_T Z_z(R,T)\).
Since \(0\le Z_z(R,T)\le1\), splitting this expectation according to
\(Z_z(R,T)\le a\) or \(Z_z(R,T)>a\) gives
\(\operatorname{stab}_z(R)\le a+q_z(R)\).
Also, the joint probability bound gives
\(\mathbb E_Rq_z(R)\le q_{\rm bad}\). Therefore Markov's inequality
implies
\[
 \Pr_R[\operatorname{stab}_z(R)>a+b]
 \le\Pr_R[q_z(R)>b]\le q_{\rm bad}/b.
\]
A union bound over the fixed input set proves
\eqref{eq:resampling-conversion}.
\end{proof}

\begin{proof}[Proof of Lemma~\ref{lem:scalar-row-stability}]
Fix \(\kappa\) and the outside bits. We first bound the average
switch change in \(g_z\), then in \(r_z\), and then apply
Lemma~\ref{lem:switch-to-resampling}.

Use \(\varepsilon\) for the bound on average acceptance change and
\(\delta\) for the size of a counted change. Set
\begin{equation}
 \begin{aligned}
 \varepsilon&=2^{-\lfloor\kappa/16\rfloor},
 \qquad \delta=\varepsilon/2,\\
 d_{\rm tree}&=\max\left\{3,
  \left\lceil\kappa^2+
    \log_2\!\frac{2|\mathcal Z_\kappa|(M_{\rm qry}+1)}
                     {\varepsilon^2}\right\rceil\right\}.
 \end{aligned}
 \label{eq:scalar-depth}
\end{equation}
The circuit and query bounds give
\(\Gamma_{\rm res},d_{\rm tree}=\poly(\kappa)\).
These choices are independent of the candidate rows and outside bits.
The cutoff in \eqref{eq:sensitivity-threshold} satisfies
\[
 \frac{s_*}{K}
 =O\!\left(\frac{\Gamma_{\rm res}\ell_{\rm in}(\kappa)^2
                         d_{\rm tree}^3}{K\delta}\right)
 \le\poly(\kappa)K^{-15/16}
 =o(\delta).
\]
Thus \(s_*/K\le\delta\) for all sufficiently large \(\kappa\).

Apply Lemma~\ref{lem:switch-variation} to \(f_z\).
Since \(\delta+s_*/K\le2\delta=\varepsilon\),
Lemma~\ref{lem:sensitivity} gives
\[
 \begin{aligned}
 \Pr_{\mathbf s}[A_{f_z}(\mathbf s)>\varepsilon]
 &\le2\left(M_{\rm qry}+\frac4\varepsilon\right)2^{-d_{\rm tree}}\\
 &\le q_{\rm bad}:=
 \frac{6(M_{\rm qry}+1)}{\varepsilon}2^{-d_{\rm tree}}.
 \end{aligned}
\]
The last inequality uses \(M_{\rm qry}\ge1\) and
\(\varepsilon\le1\).
The map \(u\mapsto\min\{1,\max\{0,3u-1\}\}\) is \(3\)-Lipschitz.
Hence clipping gives
\[
 \Pr_{\mathbf s}\!\left[
   \frac1K\sum_k|r_z(O_{\mathbf s})
                  -r_z(O_{\mathbf s^{\oplus k}})|
        >3\varepsilon\right]\le q_{\rm bad}.
\]
This bound holds for every input and every choice of candidate rows.

Apply Lemma~\ref{lem:switch-to-resampling} with
\(a=b=3\varepsilon\). The probability that any input
has expected resampling change greater than
\(6\varepsilon\) is at most
\begin{equation}
 \varepsilon_{\rm stab}
 :=\frac{2|\mathcal Z_\kappa|(M_{\rm qry}+1)}
          {\varepsilon^2}2^{-d_{\rm tree}}
 \le2^{-\kappa^2}.
 \label{eq:scalar-exception}
\end{equation}
The last inequality follows from \eqref{eq:scalar-depth}.
Outside this event, all inputs satisfy \eqref{eq:scalar-stability}
with \(\alpha=6\varepsilon\).
Every estimate is uniform in the outside bits, so the required lower
bound on \(\kappa\) depends only on \(\ell_{\rm in}\).
\end{proof}

\subsection{An answer-bit certificate for weighted stability}

Section~\ref{sec:weights} also needs a direct bound on score decreases
in terms of changed oracle answers. We record that consequence of
\eqref{eq:verifier-polynomial-maximum} here. A \emph{one-sided
fractional certificate} for \(r_z\) at \(O\) is a nonnegative weight
function \(w_O\) that satisfies the first inequality below. Its
\emph{mass} is the sum of its weights.

\begin{corollary}[Certificates for clipped scores]
\label{cor:clipped-certificate}
For every input \(z\in\mathcal Z_\kappa\) and every Boolean
oracle \(O\), the function \(r_z\) has a one-sided fractional
certificate \(w_O\), zero for \(j>M_{\rm qry}\), such that
\begin{equation}
 \begin{aligned}
 r_z(O)-r_z(O')&\le\sum_{j:O_j\ne O'_j}w_O(j)
                  \quad\text{for every Boolean oracle }O',\\
 \sum_{j\ge1}w_O(j)&\le3\pi^2\ell_{\rm in}(\kappa)^2.
 \end{aligned}
 \label{eq:clipped-certificate}
\end{equation}
\end{corollary}
\begin{proof}
Choose a witness \(\omega\) that attains \(g_z(O)\).
For every \(O'\), \eqref{eq:verifier-polynomial-maximum} gives
\[
 g_z(O)-g_z(O')
 \le p_{z,\omega}(O)-p_{z,\omega}(O').
\]
Lemma~\ref{lem:oracle-polynomial-certificate}, applied on the first
\(M_{\rm qry}\) coordinates with degree bound
\(2\ell_{\rm in}(\kappa)\), bounds the absolute polynomial
difference by nonnegative weights of mass at most
\(\pi^2\ell_{\rm in}(\kappa)^2\).
Extend these weights by zero to the other coordinates.
Clipping is nondecreasing and \(3\)-Lipschitz, so multiplying the
weights by three bounds every positive decrease in \(r_z\).
If \(r_z\) does not decrease, nonnegative weights suffice.
This proves \eqref{eq:clipped-certificate}.
\end{proof}

\section{\texorpdfstring{Weighted stability for \(\mathsf{QMA}\) oracle queries}{Weighted stability for QMA oracle queries}}\label{sec:weights}

We prove that the weights from earlier queries to the challenge unitary
cannot substantially increase the expected change in a clipped \(\mathsf{QMA}\)
score. For a computation with at most \(t\) challenge queries, the
weighted cost is at most the unweighted cost plus a small error
\(\lambda_t\). This is the estimate used in
Section~\ref{sec:weighted-query} to compare public oracle queries.

Section~\ref{sec:independent-phase} maps the path-recording database
to a purification of the five-table row \(S\). We later place this same
row at key \(k\) in the public database \(R^{k\leftarrow S}\).
We start from path-recording records because the analysis in that section uses
projection on good spaces to discard bad-record components of the database. This lets us bound the resulting weights on \(S\)
using the structure of the good records.

Fix the initial data and inter-query operations independently of
\(S\). Let \(\ket{v}\) be the path-recorded reference state
after at most \(t\) challenge queries:
\[
 \ket{v}=\sum_{L,R}\ket{v_{L,R}}_{\mathsf A}\otimes\ket{L,R}_{LR},
\]
where \(\mathsf A\) contains all query, work, and auxiliary registers.
The private databases \((L,R)\) have total size at most \(t\).
By Lemma~\ref{lem:reference-forest}, they are supported on signed
forests, so \(P_{\mathcal F}\ket{v}=\ket{v}\).
The state is independent of \(S\) and has norm at most one.

Recall that the map \(E_t\) from Proposition~\ref{prop:controlled-recorder-transfer}
converts the record purification in the \(\ket{L,R}\) basis to the
function-table purification in the \(\ket{S}\) basis, and acts as the
identity on \(\mathsf A\). For a row \(S=(f_1,\ldots,f_5)\), write
\begin{equation}
 E_t\ket{v}
 =\sum_S\sqrt{\Pr[S]}\,\ket{q_v(S)}_{\mathsf A}\otimes\ket{S}_{\mathsf F}.
 \label{eq:record-feature-evaluator}
\end{equation}
Here \(\Pr[S]\) is the uniform sampling probability of \(S\).
For \(\|\ket{v}\|\le1\), the \(S\)-block
\(\sqrt{\Pr[S]}\ket{q_v(S)}\otimes\ket{S}\) is a subnormalised
state. As in \eqref{eq:core-sample-coefficients}, we factor out
\(\sqrt{\Pr[S]}\) and call \(\ket{q_v(S)}\) the \emph{rescaled state}
at \(S\). Its norm can be greater than one.
Define its \emph{recording weight} by
\[
 W_v(S)=\|\ket{q_v(S)}\|^2.
\]
Thus \(\Pr[S]W_v(S)\) is the squared-norm mass of the \(S\) block.
The relative weight \(W_v(S)\) can exceed one, but we have 
\[
 \mathbb E_S[W_v(S)]=\|E_t\ket{v}\|^2\le\|\ket{v}\|^2.
\]
Neither the good-record projection in \(E_t\) nor any rescaled state
\(\ket{q_v(S)}\) is renormalized.

For the application, fix a baseline public oracle \(R\), a key \(k\),
and a verifier input \(z\), independently of \(S\), and put
\[
 r(S)=r_z(O_{R^{k\leftarrow S}}),\qquad c=r_z(O_R).
\]
Here \(r_z\) is the clipped optimal acceptance score from
\eqref{eq:boolean-completion-bias}. The same row \(S\) defines the
challenge \(U(S)\) and replaces row \(k\) in the score \(r(S)\).
The cost we need to bound is
\[
 \mathbb E_S\bigl[|r(S)-c|W_v(S)\bigr].
\]
Lemma~\ref{lem:scalar-row-stability} in
Section~\ref{sec:classical-stability} bounds the unweighted cost,
with the required average over \(k\), on its stated simultaneous
event over \(R\). The average bound on \(W_v\) alone does not control
the weighted cost: large weights could occur on rows with large score
changes. The next lemma rules out a substantial increase from this
correlation for the absolute score change \(|r(S)-c|\).

The scores specify the completion's answer probabilities. The lemma
bounds their weighted change; Section~\ref{sec:weighted-query} uses it
to compare coherent queries to the Boolean table sampled from these
probabilities.

\begin{lemma}[Weighted stability for absolute score changes]
\label{lem:weighted-stability}
Fix a key \(k\), all bits of \(R\) outside its row, and a classical
verifier input \(z\). Let \(S\) be a uniform row at key length
\(\kappa\), and put \(r(S)=r_z(O_{R^{k\leftarrow S}})\).
Fix a polynomial record bound \(t\) and a positive integer-valued
polynomial \(\ell_{\rm in}\) with \(|z|\le\ell_{\rm in}(\kappa)\).
Both the query count and query lengths are then at most
\(\ell_{\rm in}(\kappa)\), for every oracle and witness.
Let \(\ket{v}\) be a vector of record size at most \(t\), independent
of \(S\), with \(P_{\mathcal F}\ket{v}=\ket{v}\).
It may have arbitrary auxiliary registers.
For every \(c\in[0,1]\) fixed independently of \(S\), and all
sufficiently large \(\kappa\),
\begin{equation}
 \mathbb E_S[|r(S)-c|W_v(S)]
 \le\bigl(\mathbb E_S[|r(S)-c|]+\lambda_t\bigr)\|\ket{v}\|^2.
 \label{eq:weighted-stability}
\end{equation}
Here \(\lambda_t\le\poly(\kappa)N^{-1/2}\) depends on the fixed
bound on description length as well as \(t\). Its exact value is given
in \eqref{eq:weighted-leakage}. The lower bound on \(\kappa\) depends
only on these resource bounds. The estimate is uniform over the fixed
outside data, \(z\), and \(c\); the value of \(c\) may depend on the
fixed outside data.
\end{lemma}

We prove the weighted estimate in four steps.
Section~\ref{sec:conditional-mean} conditions on the boundary tables;
contractivity bounds the conditional mean with no error.
Section~\ref{sec:core-covariance} bounds the remaining covariance
by the largest core-bit change of the weight and the score's downward
variation. Appendices~\ref{app:core-tools}--\ref{app:derivatives}
prove the weight estimate. For the endpoint tables, marked-image
Hadamard transforms and whole-coordinate averages compose with residual encoding to give
contractions, and sparse masks control the maximum over domain labels.
Section~\ref{sec:variation} bounds the score
variation, and Section~\ref{sec:covariance} combines the estimates.

\subsection{Decomposition and the conditional mean}
\label{sec:conditional-mean}\label{sec:estimates}

Write the row as \(S=(b,\vec f)\), where
\[
 b=(f_1,f_5),\qquad \vec f=(f_2,f_3,f_4),\qquad m_{\rm core}=3nN.
\]
We condition only on the two boundary tables. The three core tables
remain independent and uniform. These five tables all belong to the
fresh row; none is part of the already-fixed oracle outside row \(k\).
Identify \(\vec f\) with its \(m_{\rm core}\) Boolean coordinates:
\begin{equation}
 h_j(x)=\bigl((\vec f)_{j,x,1},\ldots,(\vec f)_{j,x,n}\bigr),\qquad
 \xi_j(x)=e^{2\pi i h_j(x)/N},\qquad j\in\{2,3,4\}.
 \label{eq:core-bit-coordinates}
\end{equation}
Here \(h_j(x)\) is the binary word representing \(f_j(x)\), with the
construction's integer interpretation, and
\((\vec f)_{j,x,i}=R(j,k,x,i)\) for a row at key \(k\).

Retain the assumptions of Lemma~\ref{lem:weighted-stability}, and take
\(\kappa\) large enough that \(3(2t)^{16}\le N\). The threshold depends only on \(t\) and \(\ell_{\rm in}\),
uniformly in the fixed outside data. The weight estimates below hold
for every row-independent cutoff vector; forest support is used in
Section~\ref{sec:weighted-query} when applying the comparison from
Section~\ref{sec:independent-phase}.

Put
\[
 M_v=\|\ket{v}\|^2,\qquad
 \mu_v(b)=\mathbb E_{\vec f}[W_v(b,\vec f)],\qquad
 a_c(b)=\mathbb E_{\vec f}[|r(b,\vec f)-c|].
\]
For fixed \(b\), the coefficient vector in
\eqref{eq:record-feature-evaluator} is the sampled coefficient of
\(T_{\mathrm{good}}P_{\mathrm{good}}Q_b\ket{v}\) in the three
core-table registers. The square-root sampling factors therefore give
\begin{equation}
 \mu_v(b)=\|T_{\mathrm{good}}P_{\mathrm{good}}Q_b\ket{v}\|^2
 \le\|P_{\mathrm{good}}Q_b\ket{v}\|^2\le M_v.
 \label{eq:core-conditional-contraction}
\end{equation}
Here \eqref{pubcore:contractive-encoding} bounds \(T_{\mathrm{good}}\),
\(P_{\mathrm{good}}\) is an orthogonal projection, and \(Q_b\) is
unitary. All three maps act as the identity on the auxiliary registers.
There is no renormalization and no error term in this bound.

Write the conditional covariance as
\[
 \mathcal C_v(b)
 =\mathbb E_{\vec f}[|r(b,\vec f)-c|W_v(b,\vec f)]-a_c(b)\mu_v(b).
\]
Then
\begin{equation}
 \mathbb E_{b,\vec f}[|r(b,\vec f)-c|W_v(b,\vec f)]
 =\mathbb E_b[a_c(b)\mu_v(b)]+\mathbb E_b[\mathcal C_v(b)].
 \label{eq:weighted-decomposition}
\end{equation}
Since \(a_c(b)\ge0\), \eqref{eq:core-conditional-contraction} implies
\begin{equation}
 \mathbb E_b[a_c(b)\mu_v(b)]
 \le M_v\mathbb E_{b,\vec f}[|r(b,\vec f)-c|].
 \label{eq:conditional-mean-bound}
\end{equation}
It remains to bound the covariance over the three core tables.

\subsection{Bounding the core-table covariance}
\label{sec:core-covariance}

We bound the covariance using the weight's largest single-bit change
and the score's total decrease under bit flips. For a core-bit
coordinate \(u=(j,x,i)\), let \(\vec f^{\oplus u}\) denote
\(\vec f\) with bit \(u\) flipped. For a scalar- or vector-valued
function \(h\), define its half-difference function by
\[
 h^{\Delta u}(b,\vec f)
 :=\frac{h(b,\vec f)-h(b,\vec f^{\oplus u})}{2}.
\]
Thus \(\vec f^{\oplus u}\) is the modified input, whereas
\(h^{\Delta u}(b,\vec f)\) is the change in the function value, divided
by two. For a real number \(a\), write \(a_+=\max\{a,0\}\).
The score's \emph{downward variation} is the function
\begin{equation}
 \begin{aligned}
 \operatorname{Down}_{\rm core}(r;b,\vec f)
 &:=\sum_u\bigl(r(b,\vec f)-r(b,\vec f^{\oplus u})\bigr)_+\\
 &=2\sum_u\bigl(r^{\Delta u}(b,\vec f)\bigr)_+.
 \end{aligned}
 \label{eq:downward-variation-definition}
\end{equation}
There are \(m_{\rm core}\) coordinates. The next lemma applies to a
cube of any dimension \(m\). It allows the variation bound to fail on
a small set, whose contribution is controlled by a pointwise weight bound.

\begin{lemma}[Covariance for absolute score changes]
\label{lem:cube-covariance}
Let \(m\ge1\), let \(\vec f\) be uniform on \(\{0,1\}^m\), independently of \(b\),
and let \(r(b,\vec f)\in[0,1]\). Suppose
\(0\le W(b,\vec f)\le BM_0\), where \(B,M_0\ge0\).
For a fixed \(c\in[0,1]\), define
\[
 \mathcal C(b)
 :=\operatorname{Cov}_{\vec f}
       \bigl(|r(b,\vec f)-c|,W(b,\vec f)\bigr).
\]
For every \(L\ge0\),
\begin{equation}
 \begin{aligned}
 \mathbb E_b[|\mathcal C(b)|]
 &\le L\mathbb E_{b,\vec f}
       \bigl[\max_u|W^{\Delta u}(b,\vec f)|\bigr]\\
 &\quad+\frac{mBM_0}{2}
       \Pr_{b,\vec f}[\operatorname{Down}_{\rm core}(r;b,\vec f)>L].
 \end{aligned}
 \label{eq:core-covariance}
\end{equation}
\end{lemma}
\begin{proof}
Fix \(b\). Write the average over bit \(u\) as
\[
 h^{\mathrm{av}\,u}(b,\vec f)
 :=\frac{h(b,\vec f)+h(b,\vec f^{\oplus u})}{2}.
\]
Then \(h(b,\vec f)-h^{\mathrm{av}\,u}(b,\vec f)
=h^{\Delta u}(b,\vec f)\).
For \(0\le\theta\le1\), let \((T_\theta h)(b,\vec f)\) average
\(h(b,\vec g)\), where each bit of \(\vec g\) independently retains
the corresponding bit of \(\vec f\) with probability \(\theta\)
and is otherwise resampled uniformly. Define \(T_\theta^{\neg u}\)
in the same way, except that bit \(u\) is always retained.
On coordinate \(u\), resampling sends \(h(b,\vec f)\) to
\(\theta h(b,\vec f)+(1-\theta)h^{\mathrm{av}\,u}(b,\vec f)\).
The coordinate maps commute, so their product rule gives
\[
 \frac{d}{d\theta}(T_\theta h)(b,\vec f)
 =\sum_u(T_\theta^{\neg u}h^{\Delta u})(b,\vec f).
\]

For an independent uniform \(\vec f'\), the covariance formula is
\[
 \mathcal C(b)=\mathbb E_{\vec f}\Bigl[
 |r(b,\vec f)-c|\bigl(W(b,\vec f)
                   -\mathbb E_{\vec f'}[W(b,\vec f')]\bigr)\Bigr].
\]
Since \(T_1=I\) and \(T_0\) averages all bits, the centered weight
equals \(((T_1-T_0)W)(b,\vec f)\). Put
\(K_u^\theta(b,\vec f)=(T_\theta^{\neg u}W^{\Delta u})(b,\vec f)\).
The fundamental theorem of calculus gives
\[
 ((T_1-T_0)W)(b,\vec f)
 =\int_0^1\sum_uK_u^\theta(b,\vec f)\,d\theta.
\]
Substitute this identity in the covariance formula. Each
\(K_u^\theta(b,\vec f)\) changes sign when bit \(u\) is flipped,
so pairing those two inputs gives
\[
 \mathcal C(b)=\int_0^1\sum_u\mathbb E_{\vec f}\left[
       \frac{|r(b,\vec f)-c|-|r(b,\vec f^{\oplus u})-c|}{2}
       K_u^\theta(b,\vec f)\right]d\theta.
\]
The reverse triangle inequality gives
\[
 \frac{\bigl||r(b,\vec f)-c|-|r(b,\vec f^{\oplus u})-c|\bigr|}{2}
 \le|r^{\Delta u}(b,\vec f)|.
\]
Since \(|K_u^\theta(b,\vec f)|\) is unchanged by flipping \(u\),
pairing the two endpoints also gives
\[
 \begin{aligned}
 &\mathbb E_{\vec f}
   \bigl[|r^{\Delta u}(b,\vec f)|\,|K_u^\theta(b,\vec f)|\bigr]\\
 &\quad=\mathbb E_{\vec f}\bigl[
       (r(b,\vec f)-r(b,\vec f^{\oplus u}))_+
       |K_u^\theta(b,\vec f)|\bigr].
 \end{aligned}
\]
Consequently,
\begin{equation}
 |\mathcal C(b)|
 \le\int_0^1\mathbb E_{\vec f}\bigl[
       \operatorname{Down}_{\rm core}(r;b,\vec f)
       \max_u|K_u^\theta(b,\vec f)|\bigr]d\theta.
 \label{eq:covariance-orientation}
\end{equation}
We orient each edge by the score \(r(b,\vec f)\), even when the
absolute cost \(|r(b,\vec f)-c|\) increases in the other direction.

Put \(G_W(b,\vec f)=\max_u|W^{\Delta u}(b,\vec f)|\).
Each resampling map takes an average, so the triangle inequality
bounds the absolute value of its output by the average of the
absolute values. Also, \(|W^{\Delta u}(b,\vec f)|\) is unchanged
by flipping \(u\); resampling that bit therefore leaves this function
unchanged. Finally, \(|W^{\Delta u}(b,\vec f)|\le G_W(b,\vec f)\)
pointwise, and averaging preserves this inequality. These three facts give
\[
 |K_u^\theta(b,\vec f)|
 \le(T_\theta^{\neg u}|W^{\Delta u}|)(b,\vec f)
 =(T_\theta|W^{\Delta u}|)(b,\vec f)
 \le(T_\theta G_W)(b,\vec f).
\]
The right side is common to all \(u\). Taking the maximum and using that \(T_\theta\) preserves the uniform measure gives
\begin{equation}
 \mathbb E_{\vec f}\bigl[\max_u|K_u^\theta(b,\vec f)|\bigr]
 \le\mathbb E_{\vec f}\bigl[\max_u|W^{\Delta u}(b,\vec f)|\bigr].
 \label{eq:noise-derivative}
\end{equation}
Also, \(0\le W(b,\vec f)\le BM_0\) gives
\(\max_u|K_u^\theta(b,\vec f)|\le BM_0/2\), while
\(\operatorname{Down}_{\rm core}(r;b,\vec f)\le m\).
Split according to whether the downward variation is at most \(L\):
\[
 \begin{aligned}
 &\mathbb E_{\vec f}\bigl[
   \operatorname{Down}_{\rm core}(r;b,\vec f)
       \max_u|K_u^\theta(b,\vec f)|\bigr]\\
 &\quad\le L\mathbb E_{\vec f}
       \bigl[\max_u|W^{\Delta u}(b,\vec f)|\bigr]
       +\frac{mBM_0}{2}
          \Pr_{\vec f}[\operatorname{Down}_{\rm core}(r;b,\vec f)>L].
 \end{aligned}
\]
Integrate over \(\theta\) in \eqref{eq:covariance-orientation}, then
average over \(b\).
\end{proof}

\paragraph{The weight estimates.}
For the core bits \(u=(j,x,i)\), Proposition~\ref{prop:maximal-core-bit-change}
gives
\begin{equation}
 \begin{aligned}
 \mathbb E_S\bigl[\max_u\|\ket{q_v^{\Delta u}(S)}\|^2\bigr]
   &\le\mathfrak d_{\rm core}^2M_v,\\
 \mathbb E_S\bigl[\max_u|W_v^{\Delta u}(S)|\bigr]
   &\le\phi_{\rm core}M_v,\qquad
 \phi_{\rm core}=2\mathfrak d_{\rm core}+2\mathfrak d_{\rm core}^2.
 \end{aligned}
 \label{eq:weight-derivatives}
\end{equation}
Here
\[
 p_\star=2\lceil\ln(24nN)\rceil,\quad K_t=16(t+1)^3,\quad
 \mathfrak d_{\rm core}^2
 =\frac{64\mathrm e^7(t+1)^{10}(1+K_t)(p_\star-1)^3}{N}.
\]
Thus \(\phi_{\rm core}\le\poly(t,n)N^{-1/2}\).

The pointwise bound \eqref{eq:pointwise-weight} remains
\[
 0\le W_v(b,\vec f)\le B_tM_v,\qquad
 B_t=2(t+1)(2N^3)^t,\qquad \log B_t=\poly(\kappa).
\]
Both estimates allow arbitrary auxiliary registers and require no
forest support. Apply Lemma~\ref{lem:cube-covariance} with
\(m=m_{\rm core}\), \(W=W_v\), and \(M_0=M_v\).
Its covariance is \(\mathcal C_v(b)\), so for every \(L\ge0\),
\begin{equation}
 \mathbb E_b[|\mathcal C_v(b)|]
 \le\left(L\phi_{\rm core}+\frac{m_{\rm core}B_t}{2}
       \Pr_{b,\vec f}[\operatorname{Down}_{\rm core}(r;b,\vec f)>L]\right)M_v.
 \label{eq:recording-covariance-bound}
\end{equation}
We next bound the exceptional probability.

\subsection{Downward variation of the clipped score}\label{sec:variation}

A certificate bounds the score decrease caused by changing oracle
answers. If each answer changes under at most \(u_0\) core-bit
flips, summing the same certificate over all flips bounds the downward
variation by \(u_0\) times the certificate's mass. We first prove this
implication, then bound the probability that an answer has larger
sensitivity.

Write the base-oracle answers available to the verifier as
\[
 \mathsf{ans}_j(b,\vec f)=(O_{R^{k\leftarrow S}})_j,
 \qquad S=(b,\vec f),\quad j\le M_{\rm qry}.
\]
By \eqref{eq:oracle-coordinate-bound}, the query-length bound
\(\ell_{\rm in}(\kappa)\) gives \(M_{\rm qry}=2^{\poly(\kappa)}\)
possible coordinates. This counts all possible addresses, even though
the verifier makes only polynomially many queries. Define the
sensitivity of answer \(j\) by
\[
 s_j(b,\vec f)=\#\{u:\mathsf{ans}_j(b,\vec f)
                       \ne\mathsf{ans}_j(b,\vec f^{\oplus u})\}.
\]

\begin{lemma}[Downward variation of the clipped score]
\label{lem:clipped-downward-variation}
Let \(r(b,\vec f)\) be the clipped score of the fixed verifier above,
whose explicit description has length at most \(\ell_{\rm in}(\kappa)\).
Fix the boundary tables \(b\) and a core-bit array \(\vec f\).
For \(u_0\ge0\), suppose
\begin{equation}
 \max_{j\le M_{\rm qry}}s_j(b,\vec f)\le u_0.
 \label{eq:answer-sensitivity}
\end{equation}
Then
\begin{equation}
 \operatorname{Down}_{\rm core}(r;b,\vec f)
 \le L_{\rm score}(u_0):=3\pi^2\ell_{\rm in}(\kappa)^2u_0.
 \label{eq:downward-variation}
\end{equation}
\end{lemma}
\begin{proof}
Put \(S=(b,\vec f)\) and \(O=O_{R^{k\leftarrow S}}\).
Corollary~\ref{cor:clipped-certificate} gives a nonnegative certificate
\(w_O\) of mass at most \(3\pi^2\ell_{\rm in}(\kappa)^2\), supported
on \(j\le M_{\rm qry}\). By \eqref{eq:clipped-certificate}, these
same weights bound the positive score decrease for every core-bit
flip. Summing over flips, each weight \(w_O(j)\) occurs exactly
\(s_j(b,\vec f)\) times in the resulting upper bound. Therefore
\[
 \operatorname{Down}_{\rm core}(r;b,\vec f)
 \le\sum_{j\le M_{\rm qry}}w_O(j)s_j(b,\vec f)
 \le u_0\sum_jw_O(j)
 \le3\pi^2\ell_{\rm in}(\kappa)^2u_0.
\]
\end{proof}

It remains to control the event
\[
 \mathcal B_{u_0}
 =\{(b,\vec f):\max_{j\le M_{\rm qry}}s_j(b,\vec f)>u_0\}.
\]
For fixed \(b\), restrict the row-bit circuit of
Lemma~\ref{lem:base-completion-locality} by fixing the boundary bits.
This leaves a constant-depth circuit on the \(m_{\rm core}\) core bits,
with the same size bound \(\exp(\poly(\kappa))\), uniformly in \(b\)
and the fixed outside oracle. By \cite[Lemma~43]{AIK},
\[
 \Pr_{\vec f}[s_j(b,\vec f)>u_0]
 \le2m_{\rm core}2^{-u_0/\Gamma_{\rm core}},
\]
where \(\Gamma_{\rm core}=\poly(\kappa)\) depends only on the fixed
resource bound. A union bound over \(j\), followed by averaging over
\(b\), gives
\begin{equation}
 \Pr_{b,\vec f}[(b,\vec f)\in\mathcal B_{u_0}]
 \le\varepsilon(u_0):=2m_{\rm core}M_{\rm qry}
                         2^{-u_0/\Gamma_{\rm core}}.
 \label{eq:completion-sensitivity}
\end{equation}
The lemma applies outside \(\mathcal B_{u_0}\), so
\begin{equation}
 \begin{aligned}
 &\Pr_{b,\vec f}[\operatorname{Down}_{\rm core}(r;b,\vec f)>L_{\rm score}(u_0)]\\
 &\quad\le\Pr_{b,\vec f}[(b,\vec f)\in\mathcal B_{u_0}]
       \le\varepsilon(u_0).
 \end{aligned}
 \label{eq:downward-variation-tail}
\end{equation}
This is the score estimate needed in
\eqref{eq:recording-covariance-bound}.

\subsection{Combining the bounds}\label{sec:covariance}

\begin{proof}[Proof of Lemma~\ref{lem:weighted-stability}]
Use \eqref{eq:conditional-mean-bound} in
\eqref{eq:weighted-decomposition}, and bound the covariance by its
absolute value. Equations~\eqref{eq:recording-covariance-bound}
and~\eqref{eq:downward-variation-tail} give
\begin{equation}
 \begin{aligned}
 &\mathbb E_{b,\vec f}[|r(b,\vec f)-c|W_v(b,\vec f)]\\
 &\quad\le\left(\mathbb E_{b,\vec f}[|r(b,\vec f)-c|]
       +3\pi^2\ell_{\rm in}(\kappa)^2u_0\phi_{\rm core}
       +\frac{m_{\rm core}B_t\varepsilon(u_0)}2\right)M_v.
 \end{aligned}
 \label{eq:weighted-before-cutoff}
\end{equation}
Choose the cutoff so that the exceptional probability absorbs the
pointwise factor \(m_{\rm core}B_t/2\):
\begin{equation}
 u_0=\left\lceil\Gamma_{\rm core}
       \left(\kappa^2+2n+
       \log_2(M_{\rm qry}m_{\rm core}^2B_t)\right)\right\rceil.
 \label{eq:sensitivity-cutoff}
\end{equation}
Then
\[
 \frac{m_{\rm core}B_t\varepsilon(u_0)}2
 =m_{\rm core}^2B_tM_{\rm qry}2^{-u_0/\Gamma_{\rm core}}
 \le2^{-\kappa^2}N^{-2}.
\]
Set
\begin{equation}
 \lambda_t=3\pi^2\ell_{\rm in}(\kappa)^2u_0\phi_{\rm core}
                    +2^{-\kappa^2}N^{-2}.
 \label{eq:weighted-leakage}
\end{equation}
Substituting in \eqref{eq:weighted-before-cutoff} and using
\(M_v=\|\ket{v}\|^2\) proves
\[
 \mathbb E_S[|r(S)-c|W_v(S)]
 \le\bigl(\mathbb E_S[|r(S)-c|]+\lambda_t\bigr)\|\ket{v}\|^2.
\]
For each fixed polynomial resource bound, \(\Gamma_{\rm core}\),
\(\log M_{\rm qry}\), \(\log B_t\), and \(n\) are polynomially
bounded in \(\kappa\), so \(u_0\) is polynomially bounded.
Together with \(\phi_{\rm core}\le\poly(t,n)N^{-1/2}\) and
\(\kappa\le n(\kappa)\le\poly(\kappa)\), this gives
\begin{equation}
 \lambda_t\le\operatorname{poly}(\kappa)N^{-1/2}.
 \label{eq:covariance-rate}
\end{equation}
All estimates are uniform in the fixed outside data, \(z\), and \(c\).
\end{proof}

Section~\ref{sec:weighted-query} combines this weighted score bound
with the separate state-comparison error \(e_t\) from
Proposition~\ref{prop:controlled-recorder-transfer} to bound a coherent
query to the persistent Boolean completion.

\section{Proof of the independent-row comparison}\label{sec:erasure}

We combine Sections~\ref{sec:independent-phase}, \ref{sec:classical-stability},
and~\ref{sec:weights} to compare three worlds, then fix the public
oracle by Borel--Cantelli. The query comparison first bounds a signed
average. A two-run reduction gives the second moment needed to rule
out cancellation between positive and negative gaps.

\subsection{Three worlds}\label{sec:comparison-experiments}
For fixed \((R,Q)\), consider the following acceptance probabilities:
\begin{center}
\begin{tabular}{@{}llll@{}}
\toprule World & Public oracle & Challenge & Acceptance\\\midrule
1 & \((O_R,Q)\) & \(U_k\), uniform key \(k\) & \(p_A(R,Q)\)\\
2 & \((O_R,Q)\) & \(U(S)\), independent uniform row \(S\) & \(q_A(R,Q)\)\\
3 & \((O_R,Q)\) & independent Haar unitary \(V\) & \(h_A(R,Q)\)\\
\bottomrule
\end{tabular}
\end{center}
Each world uses one persistent completion \(Q\sim\nu_R\).
All challenge queries use the coherent interfaces of
Section~\ref{sec:construction}. Put
\(d_A(R,Q)=p_A(R,Q)-q_A(R,Q)\) and
\(d_A^{\mathrm{Haar}}(R,Q)=p_A(R,Q)-h_A(R,Q)\).
For every fixed \((R,Q)\), Theorem~\ref{thm:independent-phase} gives
\begin{equation}
 |q_A(R,Q)-h_A(R,Q)|\le\Delta_t=O(t^2N^{-1/8}),
 \label{eq:worlds-two-three}
\end{equation}
provided \(3(2t)^{16}\le N\): the public gates are independent of the
fresh challenge. It remains to compare worlds 1 and 2.

For that comparison, sample \(R,k,S\) independently and set
\(R'=R^{k\leftarrow S}\). The experiment with public oracle
\((O_{R'},Q')\), \(Q'\sim\nu_{R'}\), and challenge \(U(S)\) has the
averaged law of world 1. Indeed, \(R'\) is uniform, \(k\) is independent
of \(R'\), and its \(k\)-th row is \(S\). World 2 uses the same
challenge with \((O_R,Q)\). Thus only the public oracle changes.

\begin{lemma}[Public-query normal form]\label{lem:public-query-normal-form}
There are polynomial-time, oracle-independent encodings satisfying
\begin{equation}
 Q(e_R(x))=R(x),\qquad Q(e_B(x))=B_R(x).
 \label{eq:public-query-encodings}
\end{equation}
They give exact coherent simulations of base queries through \(Q\),
including controlled and complement queries, with polynomial overhead.
\end{lemma}
\begin{proof}
Encode a witness-free verifier that reads the specified bit and accepts
exactly when it is one. Its clipped score is that bit, so every completion
agrees. Compute the encoding, query, and uncompute. On an inactive control
branch route the target to a spare \(\ket{+}\); complement access appends
\(X\). Both preserve the same sampled table.
\end{proof}

\subsection{Comparing worlds 1 and 2}
\label{sec:weighted-query}\label{sec:public-query-hybrid}
Count all public queries after this simulation. Fix polynomial bounds
on challenge queries, public queries, address length, workspace, and
verifier resources, covering both verifier query count and query length.
Let \(\alpha,\varepsilon_{\rm stab}\) be the simultaneous scalar bound
and its exceptional probability from Section~\ref{sec:classical-stability}, and
let \(\lambda_t\) be the weighted bound from Section~\ref{sec:weights}.

\paragraph{One persistent table.}
\label{sec:boolean-table}\label{sec:persistent-table}
We need a query bound that holds even after earlier queries have become
correlated with the completion. Let \(\mathcal Z\) contain every tagged
address within the fixed bounds and introduce a private qubit \(D_z\)
for each \(z\). Put
\[
 P_u=\begin{pmatrix}\sqrt{1-u}&\sqrt u\\\sqrt u&-\sqrt{1-u}\end{pmatrix},
 \qquad G(u)=(P_u\otimes I)\operatorname{CNOT}_{D_z\to Y}(P_u\otimes I).
\]
Initialize the private table to \(\ket{0^{\mathcal Z}}\) and use
\begin{equation}
 \widetilde O_R=\sum_z\ket{z}\bra{z}\otimes G_z(r_z(O_R)).
 \label{eq:table-query}
\end{equation}
This represents one fixed \(Q\), throughout all queries and stages.
To see this, conjugate by \(\mathcal P_R=\bigotimes_zP_{r_z(O_R)}\).
The initial table becomes
\(\sum_b\sqrt{\nu_R(b)}\ket{b}\), and every query reads \(b_z\)
without changing \(b\). Tracing out the table averages over exactly
one completion. Adjacent conjugations cancel, and the final private
unitary does not affect the output. Other initial data may depend on
\(R\), but are independent of these completion coins conditional on \(R\).
This finite table is a marginal of the single infinite completion;
it is never reset or resampled.

On the target \(\ket{+}\) subspace, \(G(u)-G(v)\) vanishes; on
\(\ket{-}\), direct calculation gives
\[
 \|G(u)-G(v)\|^2
 =4\bigl(\sqrt{u(1-v)}-\sqrt{v(1-u)}\bigr)^2\le4|u-v|.
\]
The last inequality follows from \((a-b)^2\le|a^2-b^2|\) for \(a,b\ge0\).
Hence, for \(D_{k,S}=\widetilde O_{R'}-\widetilde O_R\),
\begin{equation}
 D_{k,S}^\dagger D_{k,S}
 \preceq4\sum_z|r_z(O_{R'})-r_z(O_R)|\ket{z}\bra{z}\otimes I,
 \qquad \|D_{k,S}\|\le2.
 \label{eq:table-query-weighted}
\end{equation}
This holds with arbitrary auxiliary entanglement and for controlled or
complement access; the private table need not be in its initial state.

\paragraph{One public query.}
Fix \(R\) in the simultaneous scalar-stability event. A purified
world-2 prefix \(\ket{\Phi_{R,S}}\) has operations independent of \(S,k\)
except for at most \(t\) challenge queries. Proposition~\ref{prop:controlled-recorder-transfer}
gives an unrotated, forest-supported reference vector
\[
 \ket{v}=\sum_z\ket{z}\ket{v_z},\qquad
 \sum_z\|\ket{v_z}\|^2=\|\ket{v}\|^2\le1,\qquad
 \mathbb E_S[\|\ket{\Phi_{R,S}}-\ket{q_v(S)}\|^2]\le e_t^2.
\]
Retain all auxiliary registers, including the private table.
The address projections commute with \(E_t\) and \(P_{\mathcal F}\), so
\(\ket{v_z}\) remains forest-supported and independent of \(S,k\).
At each fixed \(k,z\), apply Lemma~\ref{lem:weighted-stability} with
\(c=r_z(O_R)\). Average over \(k\) and use scalar stability to obtain
\begin{equation}
 \begin{aligned}
 &\mathbb E_{k,S}\Bigl[\sum_z|r_z(O_{R'})-r_z(O_R)|W_{v_z}(S)\Bigr]\\
 &\quad\le\sum_z\bigl(\mathbb E_{k,S}[|r_z(O_{R'})-r_z(O_R)|]
                     +\lambda_t\bigr)\|\ket{v_z}\|^2
 \le\alpha+\lambda_t.
 \end{aligned}
 \label{eq:weighted-score-change}
\end{equation}
No address component is normalized. By
\eqref{eq:table-query-weighted}, linearity of evaluation, and the
triangle inequality in \(L^2(k,S)\),
\begin{equation}
 \begin{aligned}
 \bigl(\mathbb E_{k,S}[\|D_{k,S}\ket{\Phi_{R,S}}\|^2]\bigr)^{1/2}
 &\le\bigl(\mathbb E_{k,S}[\|D_{k,S}\ket{q_v(S)}\|^2]\bigr)^{1/2}+2e_t\\
 &\le2\sqrt{\alpha+\lambda_t}+2e_t.
 \end{aligned}
 \label{eq:rounded-query-loss}
\end{equation}
Thus summing addresses costs only \(\sum_z\|\ket{v_z}\|^2\le1\),
not the number of addresses.

\begin{lemma}[Averaged signed comparison]\label{lem:finite-query-bound}
Under these resource bounds and the hypotheses of Sections~\ref{sec:independent-phase}
and~\ref{sec:weights}, an algorithm with \(Q_{\rm pub}\) public queries satisfies
\begin{equation}
 |\mathbb E_{R,Q}[d_A(R,Q)]|
 \le\epsilon_A^{\rm round}:=
 2Q_{\rm pub}(\sqrt{\alpha+\lambda_t}+e_t)+\varepsilon_{\rm stab}.
 \label{eq:signed-average}
\end{equation}
The bound is uniform over arbitrary fixed initial states and interquery
operations independent of the sampled public oracle and challenge;
they need not have efficient descriptions. Completion verifiers keep
their original gate set and resource bounds.
\end{lemma}
\begin{proof}
Let \(\ket{\Phi_j}\) and \(\ket{\Phi'_j}\) be the two prefixes before
public query \(j\), with \(\ket{\Phi_j}\) from world 2. Write
\[
 \widetilde O_{R'}\ket{\Phi'_j}-\widetilde O_R\ket{\Phi_j}
 =\widetilde O_{R'}(\ket{\Phi'_j}-\ket{\Phi_j})
       +D_{k,S}\ket{\Phi_j}.
\]
The common gates between public queries, including the same challenge
\(U(S)\), preserve vector distance. Starting from identical initial
states and the same zero table, induction bounds the final distance
by \(\sum_j\|D_{k,S}\ket{\Phi_j}\|\).
For good \(R\), average \eqref{eq:rounded-query-loss} and use that
acceptance distance is at most vector distance. Bad \(R\) contributes
at most \(\varepsilon_{\rm stab}\) once. Arbitrary auxiliary registers
and fixed interquery operations are allowed throughout.
\end{proof}

\subsection{From the signed mean to the second moment}
\label{sec:second-moment}\label{sec:fixing}\label{sec:training}
We use the two-run argument of \cite[Corollary~32]{KQST}.
Both runs must share the same public oracle, so their conditional
acceptance probabilities multiply before we average over \((R,Q)\).

\begin{lemma}[Exact second-moment reduction]\label{lem:exact-second-moment}
Let \(A\) have at most \(t\ge1\) challenge queries, \(q\) encoded public
queries, address length \(\ell\), and accessible register size \(w\).
Let \(\ell_{\rm enc}(\kappa,n)\) bound the encoded row-bit address length.
Apply Lemma~\ref{lem:finite-query-bound} with enlarged resources
\[
 t,\qquad 2q+10nt,\qquad
 \max\{\ell,\ell_{\rm enc}(\kappa,n)\},\qquad
 2w+\poly(\kappa,n,\ell,t,q),
\]
respectively, and verifier bounds covering every resulting address.
Writing its bound as \(\epsilon_B^{\rm round}\), we have
\begin{equation}
 \mathbb E_{R,Q}[d_A(R,Q)^2]\le\epsilon_B^{\rm round}.
 \label{eq:second-moment}
\end{equation}
\end{lemma}
\begin{proof}
For a fixed private unitary \(W\), write \(q_{A,W}(R,Q)\) for \(A\)'s
acceptance probability with challenge \(W\). Define \(B_W\) as follows.
Sample a fair bit \(b\). Run \(A\) with a fresh keyed challenge \(U_J\)
computed from \(R\) if \(b=0\), and with \(W\) if \(b=1\), obtaining
\(X\). Run a fresh copy of \(A\) on the external challenge, obtaining
\(Y\), and accept if \(X\oplus Y=b\). The two runs use the same
\((R,Q)\), independent local randomness and fresh work registers;
\(J\) is independent of the external key. Conditional on \((R,Q)\),
if the second run accepts with probability \(s\), then
\[
 \Pr[B_W=1\mid R,Q]=\tfrac12+
       (p_A(R,Q)-q_{A,W}(R,Q))(s-\tfrac12).
\]
For an independent uniform row \(T\), \(\mathbb E_T[q_{A,U(T)}(R,Q)]=q_A(R,Q)\).
Subtracting the two external-challenge cases therefore gives
\begin{equation}
 \begin{aligned}
 d_{B_W}(R,Q)&=(p_A(R,Q)-q_{A,W}(R,Q))d_A(R,Q),\\
 \mathbb E_T[d_{B_{U(T)}}(R,Q)]&=d_A(R,Q)^2.
 \end{aligned}
 \label{eq:predictor-identity}
\end{equation}
Only the second run uses the external challenge. Each \(U_J\) query
costs at most \(10n\) exact row-bit queries by
Appendix~\ref{sec:phase-implementation}; encode them through \(Q\).
This gives the displayed resources. Treat \(W\) as a fixed interquery
operation and apply the uniform bound of Lemma~\ref{lem:finite-query-bound}.
Average over \(T\) in \eqref{eq:predictor-identity}. The same private
table represents both runs; no completion bits are resampled.
\end{proof}

\begin{remark}[The Haar gap]\label{rem:haar-second-moment}
Equation~\eqref{eq:worlds-two-three} adds \(\Delta_t\) to the signed
bound. In the two-run proof sample the private \(W\) from Haar measure
instead. The same identity and the bound uniform in \(W\) give
\(\mathbb E_{R,Q}[(d_A^{\mathrm{Haar}}(R,Q))^2]\le\epsilon_B^{\rm round}+\Delta_t\).
The private-unitary algorithms need not be efficient or uniform;
only the original adversaries enter the countable intersection below.
\end{remark}

\subsection{Fixing the oracles}\label{sec:fixed-oracle-conclusion}
For \(\kappa\le n(\kappa)\le\poly(\kappa)\) and \(N=2^n\ge K=2^\kappa\),
the previous sections give
\begin{equation}
 \begin{aligned}
 \alpha&=6\cdot2^{-\lfloor\kappa/16\rfloor}=O(K^{-1/16}),&
 \varepsilon_{\rm stab}&\le2^{-\kappa^2},\\
 \phi_{\rm core},\lambda_t&\le\poly(\kappa)N^{-1/2},&
 e_t,\Delta_t&=O(t^2N^{-1/8}).
 \end{aligned}
 \label{eq:coarse-rates}
\end{equation}
All hypotheses hold eventually for fixed polynomial resources, with
restriction depth chosen as in \eqref{eq:scalar-depth}. Consequently
\(\sqrt{\alpha+\lambda_t}+e_t\le\poly(\kappa)2^{-\kappa/32}\);
\(\sqrt\alpha\) determines the exponent. Apply this also with the
\emph{enlarged two-run resources}, recomputing \(\lambda_t\) and the
simultaneous scalar event for those bounds. Let \(d(R,Q)\) denote either
\(d_A(R,Q)\) or \(d_A^{\mathrm{Haar}}(R,Q)\). One polynomial \(P_A\ge1\) gives
\begin{equation}
 |\mathbb E_{R,Q}[d(R,Q)]|,\quad\mathbb E_{R,Q}[d(R,Q)^2]
 \le P_A(\kappa)2^{-\kappa/32}.
 \label{eq:mean-and-moment}
\end{equation}
Here \(P_A\) may depend on the adversary, dimension rule, finite verifier
family, and enlarged resources. The separate absolute-mean bound is
\(\mathbb E_{R,Q}[|d(R,Q)|]\le\sqrt{P_A(\kappa)}2^{-\kappa/64}\), by Cauchy--Schwarz.

\begin{proposition}[A fixed public oracle for all uniform adversaries]
\label{prop:fixed-oracle-comparison}
There is a probability-one set of infinite pairs \((R,Q)\) on which
\(|d_A(R,Q)|,|d_A^{\mathrm{Haar}}(R,Q)|=O(2^{-\kappa/128})\) eventually for every
uniform polynomial-time adversary and permitted fixed finite verifier
family, each with its own polynomial resources. The cutoff may depend
on the pair, adversary, family, and resources. For almost every \(R\),
this set has conditional \(\nu_R\)-measure one.
\end{proposition}
\begin{proof}
For either gap, Markov's inequality gives
\(\Pr[|d(R,Q)|>\epsilon]\le\min\{1,\mathbb E[d(R,Q)^2]/\epsilon^2\}\).
In particular,
\begin{equation}
 \Pr_{R,Q}[|d(R,Q)|>2^{-\kappa/128}]
 \le P_A(\kappa)2^{-\kappa/64}.
 \label{eq:summable-tail}
\end{equation}
This is summable in \(\kappa\). Borel--Cantelli gives the eventual
bound for each fixed adversary, finite verifier family, and polynomial
resource bound. These choices are countable; intersect
their probability-one events, for both gaps, to obtain \(\mathcal G\).
All parameters use the same infinite pair \((R,Q)\); no independence
between parameters is required. Finally,
\begin{equation}
 1=\Pr[(R,Q)\in\mathcal G]
   =\mathbb E_R[\nu_R(\{Q:(R,Q)\in\mathcal G\})]
 \label{eq:conditional-fixing}
\end{equation}
implies the conditional assertion. The phase-synthesis error
\(\poly(\kappa)2^{-\kappa^2}\) from
Appendix~\ref{sec:phase-implementation} is absorbed in the final bound.
The adversaries are uniform, without oracle-dependent advice, and
verifier inputs remain classical.
\end{proof}

This proves Lemma~\ref{lem:row-decoupling} and, using the world-2/world-3
comparison, Theorem~\ref{thm:main}.

\begin{remark}[Sharper eventual rates]\label{rem:sharper-fixed-rates}
Every fixed \(0<c<1/64\) also works: the failure bound at threshold
\(2^{-c\kappa}\) is \(P_A(\kappa)2^{-(1/32-2c)\kappa}\).
Intersect along \(c_j=(1-1/j)/64\), \(j\ge2\), to obtain a common
probability-one event for all such \(c\). At the endpoint, threshold
\(\kappa\sqrt{P_A(\kappa)}2^{-\kappa/64}\) gives failure probability
at most \(1/\kappa^2\), and hence an eventual bound with a polynomial
prefactor.
\end{remark}

\subsection{\texorpdfstring{Other verifier tables and \(\mathsf{PH}\) queries}{Other verifier tables and PH queries}}
\label{sec:other-verifiers}\label{sec:exact}
For a fixed ordinary \(\mathsf{QMA}\) verifier \(\mathcal W\) over the same gate set
and classical input \(x\), encode
its explicit circuit in the universal verifier, with register sizes in
unary and an injective identifying string \((j,\mathcal W,x)\).
Exact circuit execution gives the same maximum acceptance score over
all witnesses, for every base oracle; unused witness qubits are ignored.
Thus clipping gives the same answer probabilities. Distinct tagged
inputs use independent bits conditional on \(R\); repeated inputs reuse
the same bit, even across stages. Reversible encoding therefore preserves
the entire completion law of any fixed finite verifier family, including
off-promise entries. If a verifier has no gap inputs, its scores are
\(0\) or \(1\), so its language gate is deterministic. Since
\(\PH^R\subseteq\mathsf P^{O_R}\), compute, XOR, and uncompute also
implements every fixed \(\mathsf{PH}\) query coherently, including within verifiers.
For Appendix~\ref{app:qma2}, retain all table tags and use
\(\nu_R^{\rm mix}\) from \eqref{eq:qma2-joint-law}; encode base queries
through the ordinary table \(Q^{(1)}\).

\clearpage
\phantomsection
\addcontentsline{toc}{section}{References}
\begingroup\raggedright
\bibliographystyle{alpha}
\bibliography{references}
\endgroup

\clearpage

\appendix
\section{Collision-free encoding and core compression}\label{app:core}

We prove the encoding estimates on
\(\mathcal G_{\mathrm{good}}=\bigoplus_{s=0}^t\mathcal G_{\mathrm{good},s}\).
Throughout, \(N=2^n\), \(t\ge1\), and \(q>8t\); the construction uses
\(q=N\). Every map acts as the identity on untouched Hilbert references.

\subsection{Normalized memories and collision-free packing}

Use the normalized multiset convention of~\eqref{pubcore:record-relation-state}.
Put \(\mathcal T_q=\{f:\mathbb F_2^n\to\mathbb Z_q\}\),
\(\vec f=(f_2,f_3,f_4)\), and \(\xi_j(z)=e^{2\pi i f_j(z)/q}\).
For disjoint sets \(P,M\), define
\begin{equation}
 J_j\ket{P,M}
 =\frac1{q^{N/2}}\sum_{f_j\in\mathcal T_q}
       \left(\prod_{u\in P}\xi_j(u)
             \prod_{u\in M}\overline{\xi_j(u)}\right)\ket{f_j}.
 \label{pubcore:signed-set-packing}
\end{equation}
Character orthogonality gives
\[
\begin{aligned}
 \bra{P',M'}J_j^\dagger J_j\ket{P,M}
 &=\E_{f_j}\Bigl[\prod_u\xi_j(u)^{\mathbf1_P(u)-\mathbf1_M(u)
                         -\mathbf1_{P'}(u)+\mathbf1_{M'}(u)}\Bigr]\\
 &=\mathbf1[(P,M)=(P',M')].
\end{aligned}
\]
The endpoint registers contain all disjoint set pairs, with no size
cutoff or constraint linking their sizes to the middle memory. The latter has space
\[
 \mathcal F_{\le t}
 =\bigoplus_{\ell+r\le t}
       \operatorname{Sym}^{\ell}(\mathbb C^N)
       \otimes\operatorname{Sym}^{r}(\mathbb C^N).
\]
For a good record \((L,R)\), define
\begin{equation}
 J_D\ket{L,R}
 =\ket{X_+,X_-}\otimes\ket{D_+,D_-}
       \otimes\ket{Y_+,Y_-}.
 \label{pubcore:decoding-isometry}
\end{equation}
The six sets recover \((L,R)\) by~\eqref{eq:good-record-recovery}.
The memory Hadamards preserve each symmetric length sector, so
\[
 J_D^\dagger J_D=I,\qquad
 J_D^\dagger\mathsf H_{\mathrm{mem}}^\dagger
                \mathsf H_{\mathrm{mem}}J_D=I.
\]

Let \(\Pi_{\mathrm{cf}}\) retain middle multisets satisfying
\(m_{Z_+}(z)+m_{Z_-}(z)\le1\) for every \(z\). Define
\(J_{\mathrm{pack}}=J_2\otimes J_3\otimes J_4\) there and zero
on its orthogonal complement. Thus~\eqref{eq:all-table-packing}, with
phase order \(q\), gives
\[
 J_{\mathrm{pack}}^\dagger J_{\mathrm{pack}}=\Pi_{\mathrm{cf}}
\]
across all bidegrees. Its output contains only the three table registers.
Consequently,
\begin{equation}
\begin{aligned}
 T_{\mathrm{good}}
 &=J_{\mathrm{pack}}\Pi_{\mathrm{cf}}\mathsf H_{\mathrm{mem}}J_D,\\
 T_{\mathrm{good}}^\dagger T_{\mathrm{good}}
 &=J_D^\dagger\mathsf H_{\mathrm{mem}}^\dagger\Pi_{\mathrm{cf}}
                    \mathsf H_{\mathrm{mem}}J_D\preceq I.
\end{aligned}
 \label{pubcore:cf-factorization}
\end{equation}
In particular, \(T_{\mathrm{good}}\ket{\varnothing,\varnothing}
=\ket{\mathbf1}=\ket{\Omega_{234}}\), the uniform table purification.

For the column expansion, order
\(L=((x_i^+,y_i^+))_{i=1}^{\ell}\) and
\(R=((x_j^-,y_j^-))_{j=1}^{r}\).
Input symmetrization and output packing each contribute
\((\ell!r!)^{-1/2}\). Hence
\begin{equation}
\begin{aligned}
 &q^{3N/2}\bra{\vec f}T_{\mathrm{good}}\ket{L,R}\\
 &=\frac{\chi_{L,R}(\vec f)}{\ell!r!}
   \sum_{\pi\in S_\ell,\,\tau\in S_r}
   \sum_{\substack{z_1,\ldots,z_\ell,w_1,\ldots,w_r\\
                   \text{all distinct}}}\\
 &\qquad\cdot
     \prod_{i=1}^{\ell}H_{x_{\pi(i)}^+\oplus y_{\pi(i)}^+,z_i}\xi_3(z_i)
     \prod_{j=1}^{r}H_{x_{\tau(j)}^-\oplus y_{\tau(j)}^-,w_j}
                                      \overline{\xi_3(w_j)}\\
 &=\chi_{L,R}(\vec f)\mu_{L,R}^{\mathrm{cf}}(\vec f)
   =a_{L,R}^{\mathrm{cf}}(\vec f).
\end{aligned}
 \label{pubcore:packed-column}
\end{equation}
The \(\ell!r!\) permutation terms coincide after relabeling the
distinct locations. Thus \(T_{\mathrm{good}}\ket{L,R}
=\ket{a_{L,R}^{\mathrm{cf}}}\).
These are the same coefficient functions used in
Appendices~\ref{app:endpoint-differences} and~\ref{app:derivatives}.

\subsection{A two-address collision estimate}

\begin{lemma}[Collision-free encoding loss]
\label{pubcore:good-gram}
On the good-record space,
\begin{equation}
\begin{aligned}
 \|(I-\Pi_{\mathrm{cf}})\mathsf H_{\mathrm{mem}}J_D\|^2
 &\le\varepsilon_{\mathrm{cf},t}=\frac{2t(t-1)}N,\\
 0\preceq I-T_{\mathrm{good}}^\dagger T_{\mathrm{good}}
 &\preceq\varepsilon_{\mathrm{cf},t}I.
\end{aligned}
 \label{pubcore:good-gram-bound}
\end{equation}
\end{lemma}

\begin{proof}
Fix an endpoint block \(\mathcal E=(X_+,X_-,Y_+,Y_-)\), with
\(s=\ell+r\) and \(m=\ell!r!\). Suppress its endpoint factors in
\(J_D\) to obtain \(J_{\mathcal E}\). If
\(\mathcal O(L,R)\) is the set of separate orderings of \(D_+,D_-\), then
\[
 J_{\mathcal E}\ket{L,R}
   =\frac1{\sqrt m}\sum_{\mathbf d\in\mathcal O(L,R)}\ket{\mathbf d},
 \qquad |\mathcal O(L,R)|=m.
\]
Define 
\begin{equation}
 \Gamma_{\mathcal E}^{\mathrm{cf}}
 =J_{\mathcal E}^\dagger(H^{\otimes s})^\dagger
          \Pi_{\mathrm{cf}}H^{\otimes s}J_{\mathcal E}.
 \label{pubcore:cf-block-gram}
\end{equation}
In ordered coordinates, let \(P_{ij}\) impose equality of positions
\(i,j\), and put
\[
 B_{ij}=J_{\mathcal E}^\dagger(H^{\otimes s})^\dagger
                  P_{ij}H^{\otimes s}J_{\mathcal E}\succeq0.
\]
Since \(I-\Pi_{\mathrm{cf}}\preceq\sum_{i<j}P_{ij}\), it suffices
to bound these compressed pair operators. Next for two orderings $\ket{\mathbf d}$ and $\ket{\mathbf e}$, we bound
\begin{equation}
\begin{aligned}
 &\bra{\mathbf d}(H^{\otimes s})^\dagger
                  P_{ij}H^{\otimes s}\ket{\mathbf e}\\
 &=\left(\sum_z\overline{H_{d_i,z}H_{d_j,z}}
                         H_{e_i,z}H_{e_j,z}\right)
                         \prod_{k\ne i,j}\mathbf1[d_k=e_k]\\
 &=\frac{\mathbf1[d_i\oplus d_j=e_i\oplus e_j]}N
                         \prod_{k\ne i,j}\mathbf1[d_k=e_k].
\end{aligned}
 \label{pubcore:cf-pair-bound}
\end{equation}
Each matrix entry of \(B_{ij}\) therefore has normalization \(1/(Nm)\)
times the number of compatible pairs of orderings.

\emph{Opposite signs.}
The untouched labels recover all but one edge of each sign by
Cartesian XOR injectivity. The remaining domain and image of each
sign force the same record. Each input ordering has one compatible
output ordering, so
\[
 \bra{L',R'}B_{ij}\ket{L,R}
 =\frac{m}{Nm}\mathbf1[(L',R')=(L,R)],
 \qquad B_{ij}=\frac1N I.
\]

\emph{Two positive positions.}
The untouched labels recover all but two positive edges. The remaining
domains and images admit the original matching and its two-edge swap.
The ordering counts are
\[
\begin{aligned}
 \text{same record:}\qquad &\frac{2m}{Nm}=\frac2N,\\
 \text{specified swapped record:}\qquad
 &\frac{(2!)(2!)(\ell-2)!r!}{N\ell!r!}
       =\frac{2}{N\binom\ell2}.
\end{aligned}
\]
In the second line, the two selected input edges and two selected
output edges each have \(2!\) orders; the untouched input positions
have \((\ell-2)!r!\) orders and determine the output positions.
There are at most \(\binom\ell2\) swapped records; nongood targets
are absent. Thus every absolute row sum is at most
\[
 \frac2N+\binom\ell2\frac{2}{N\binom\ell2}=\frac4N.
\]
Symmetry gives the same column bound. Hence, 
\(\|B_{ij}\|\le4/N\). Two negative positions give the same bound
with \(\ell,r\) interchanged.

Summing over the \(\binom\ell2\) positive, \(\binom r2\) negative,
and \(\ell r\) opposite-sign pairs gives
\begin{equation}
\begin{aligned}
 0\preceq I-\Gamma_{\mathcal E}^{\mathrm{cf}}
 &\preceq\sum_{i<j}B_{ij}\\
 &\preceq\frac{4\binom\ell2+4\binom r2+\ell r}{N}I
 \preceq\frac{2s(s-1)}N I.
\end{aligned}
 \label{pubcore:good-row-bound}
\end{equation}
Take the orthogonal sum over endpoint blocks and use \(s\le t\).
Since \(\mathsf H_{\mathrm{mem}}J_D\) is an isometry,
\[
 \|(I-\Pi_{\mathrm{cf}})\mathsf H_{\mathrm{mem}}J_D\|^2
 =\|I-T_{\mathrm{good}}^\dagger T_{\mathrm{good}}\|
 \le\frac{2t(t-1)}N.
\]
\end{proof}

\subsection{\texorpdfstring{\(HFH\)}{HFH} records XORs}
\label{pubcore:xor-recording}

Let \(f:\mathbb F_2^n\to\mathbb Z_q\) be uniform, put
\(\omega=e^{2\pi i/q}\), and write
\[
 \mathsf F\ket{z}\ket{f}=\omega^{f(z)}\ket{z}\ket{f},
 \qquad
 \ket{\Omega}=q^{-N/2}\sum_f\ket{f},
 \qquad
 \ket{\zeta_z}=q^{-N/2}\sum_f\omega^{f(z)}\ket{f}.
\]
Since
\(\langle\zeta_{z'}\mid\zeta_z\rangle
=\E_f\bigl[\omega^{f(z)-f(z')}\bigr]=\mathbf1[z=z']\), the states
\(\ket{d}_{\mathrm{xor}}:=\sum_zH_{z,d}\ket{\zeta_z}\)
are orthonormal. Notice that
\begin{equation}
 \sum_zH_{y,z}H_{z,x}H_{z,d}
 =\frac1{N\sqrt N}\sum_z(-1)^{(x\oplus y\oplus d)\cdot z}
 =\frac1{\sqrt N}\mathbf1[d=x\oplus y]
 \label{pubcore:xor-hadamard-identity}
\end{equation}
and hence
\begin{equation}
 (H_A\mathsf F H_A)\ket{x}\ket{\Omega}
 =\frac1{\sqrt N}\sum_y\ket{y}\ket{x\oplus y}_{\mathrm{xor}}.
 \label{pubcore:empty-xor-recording}
\end{equation}
Thus \(HFH\) records the XOR of its input and output in a basis of
its private table register. This basis change is only for the analysis.
The next lemma gives the corresponding update on a populated memory.

Use the signed symmetric memory \(\mathcal F_{\le t}\).
For a multiset \(M_\sigma\), let \(m_{M_\sigma}(z)\) count occurrences
of \(z\). Normalized insertion and deletion are
\begin{equation}
 a_{\sigma,z}^\dagger\ket{M_\sigma}
 =\sqrt{m_{M_\sigma}(z)+1}\,\ket{M_\sigma\cup\{z\}},
 \qquad
 a_{\sigma,z}\ket{M_\sigma}
 =\sqrt{m_{M_\sigma}(z)}\,\ket{M_\sigma\setminus\{z\}}.
 \label{pubcore:normalized-insertion}
\end{equation}
Deletion is zero when \(z\) is absent. Define insertion through total
length \(t+1\) and use \(\Pi_{\le t}\) on both sides of every update;
the cutoff deletion is the exact adjoint of the cutoff insertion.
Let \(J_{\mathrm{cf}}\) be the single-table packing map
\eqref{pubcore:signed-set-packing} on collision-free signed sets,
extended by zero to the rest of \(\mathcal F_{\le t}\). Put
\[
 T_{\mathrm{mid}}=J_{\mathrm{cf}}\mathsf H_{\mathrm{mem}},
 \qquad
 \Pi_{\mathrm{rot}}=\mathsf H_{\mathrm{mem}}^\dagger
                      \Pi_{\mathrm{cf}}\mathsf H_{\mathrm{mem}}.
\]
Here \(\Pi_{\mathrm{rot}}\) tests distinct table locations after the
memory Hadamards; it is not the projector onto distinct XOR labels.

\begin{lemma}[Hadamard-conjugated phase recording]
\label{pubcore:xor-recorder}
Define the bounded update
\begin{equation}
 \mathsf X_t=\Pi_{\le t}
   \left[\frac1{\sqrt N}\sum_{x,y}\ket{y}\bra{x}_A\otimes
        (a_{+,x\oplus y}^\dagger+a_{-,x\oplus y})\right]\Pi_{\le t}.
 \label{pubcore:xor-update}
\end{equation}
Then
\begin{equation}
\begin{aligned}
 T_{\mathrm{mid}}^\dagger(H_A\mathsf F H_A)T_{\mathrm{mid}}
    &=\Pi_{\mathrm{rot}}\mathsf X_t\Pi_{\mathrm{rot}},\\
 T_{\mathrm{mid}}^\dagger T_{\mathrm{mid}}&=\Pi_{\mathrm{rot}},
 \qquad
 \|\mathsf X_t\|\le\sqrt{t+1}+\sqrt t\le2\sqrt{t+1}.
\end{aligned}
 \label{pubcore:middle-compression}
\end{equation}
Replacing \(\mathsf F\) by \(\mathsf F^\dagger\) replaces
\(\mathsf X_t\) by \(\mathsf X_t^\dagger\), with the same encoding.
\end{lemma}

\begin{proof}
In the packed basis, a positive phase increment retains exactly
\(0\mapsto+1\) and \(-1\mapsto0\), each with coefficient one.
The first inserts a positive entry; the second deletes a negative
entry. The transition \(+1\mapsto+2\) has no collision-free target.
Since \(q>8t\), these counters have no modular aliases, and a counter
vector determines its signed sets and their sizes. Thus, on the full
sum of length and sign sectors,
\begin{equation}
 J_{\mathrm{cf}}^\dagger\mathsf FJ_{\mathrm{cf}}
 =\Pi_{\mathrm{cf}}\Pi_{\le t}
   \left[\sum_z\ket{z}\bra{z}_A\otimes
                  (a_{+,z}^\dagger+a_{-,z})\right]
       \Pi_{\le t}\Pi_{\mathrm{cf}},
 \qquad J_{\mathrm{cf}}^\dagger J_{\mathrm{cf}}=\Pi_{\mathrm{cf}}.
 \label{pubcore:local-packed-update}
\end{equation}

The change-of-basis argument is similar to the proofs of Claims~20
and~21 in~\cite[Section~10.1]{MH}, with locations in place of edges.
For a symmetric \(k\)-tensor \(\ket{v}\), normalized insertion is
\[
 a^\dagger(\ket{u})\ket{v}
 =\sqrt{k+1}\,\operatorname{Sym}_{k+1}(\ket{u}\otimes\ket{v}),
\]
where \(\operatorname{Sym}_{k+1}\) is the orthogonal symmetrizer.
It commutes with \(H^{\otimes(k+1)}\). Hence, between matching
length sectors,
\begin{equation}
 \mathsf H_{\mathrm{mem}}^\dagger a_{\sigma,z}^\dagger
           \mathsf H_{\mathrm{mem}}
       =\sum_d H_{z,d}a_{\sigma,d}^\dagger,
 \qquad
 \mathsf H_{\mathrm{mem}}^\dagger a_{\sigma,z}
           \mathsf H_{\mathrm{mem}}
       =\sum_d H_{z,d}a_{\sigma,d}.
 \label{pubcore:insertion-covariance}
\end{equation}
The second identity is the adjoint of the first, since \(H\) is real.
Conjugating the bracketed update in~\eqref{pubcore:local-packed-update},
with its length cutoffs, by the query and memory Hadamards gives \(\mathsf X_t\) by
\eqref{pubcore:xor-hadamard-identity}. The length cutoff commutes
with these Hadamards, and the collision projections become
\(\Pi_{\mathrm{rot}}\), proving the two identities in
\eqref{pubcore:middle-compression}.

The address controls in~\eqref{pubcore:local-packed-update} are
orthogonal. Since
\(\|a_{+,z}^\dagger\Pi_{\le t}\|\le\sqrt{t+1}\) and
\(\|a_{-,z}\Pi_{\le t}\|\le\sqrt t\), the update between its
collision projections has norm at most \(\sqrt{t+1}+\sqrt t\)
even before those projections are applied. Unitary change of basis
gives the same bound for \(\mathsf X_t\). Taking adjoints proves
the negative-query identity.
\end{proof}

The lemma is exact with its two collision projections. The small loss
from removing them will come from Lemma~\ref{pubcore:good-gram} on
\(J_D\), which also retains the endpoint records.

\subsection{Applying the XOR recorder to the core}
\label{pubcore:all-branch-proof}

There are two losses to bound. The Hadamards give each allowed insertion
amplitude \(1/\sqrt N\), whereas the recorder uses \(1/\sqrt{N-s}\)
on a source of size \(s\). The collision projection also removes
assignments with repeated middle locations. We bound the amplitude
change first, then the loss from this projection.

\paragraph{Changing the amplitude.}
Let \(\widetilde W\) use the same allowed insertions and adjoint
deletions as \(W_{\mathrm{good}}\), all with amplitude \(1/\sqrt N\):
\begin{equation}
\begin{aligned}
 \widetilde W\ket{x,L,R}
 =\frac1{\sqrt N}\biggl(&\sum_{y\in f_+(x;L,R)}
             \ket{y,L\cup\{(x,y)\},R}\\
       &+\sum_{y:\,(x,y)\in R}\ket{y,L,R\setminus\{(x,y)\}}\biggr).
\end{aligned}
 \label{pubcore:creation-maps}
\end{equation}
In terms of the creation
blocks in~\eqref{eq:allowed-creations}, where \(s\) is the source size
of the creation map,
\begin{align}
 W_{\mathrm{good}}
 &=\sum_{s=0}^{t-1}
       (\mathsf C_{+,\to,s}+\mathsf C_{-,\leftarrow,s}^\dagger),
 \label{pubcore:restricted-recorder}\\
 \widetilde W
 &=\sum_{s=0}^{t-1}\sqrt{1-s/N}\,
       (\mathsf C_{+,\to,s}+\mathsf C_{-,\leftarrow,s}^\dagger).
 \label{eq:raw-allowed-creations}
\end{align}
Thus deletion from size \(s+1\) uses the same factor as insertion
from size \(s\). The creation maps are contractions with orthogonal
size blocks, so
\begin{equation}
 \|\widetilde W-W_{\mathrm{good}}\|
 \le2\max_{s<t}(1-\sqrt{1-s/N})\le\frac{2t}{N}.
 \label{pubcore:amplitude-loss}
\end{equation}

\paragraph{Attaching the endpoint records.}
For \(j=2,4\), the endpoint phase adds one to the exponent at the
query address. Its action in signed-set coordinates is
\begin{equation}
\begin{aligned}
 \mathsf E_j&:=J_j^\dagger\mathsf F_jJ_j,\\
 \mathsf E_j\ket{x}\ket{P,M}
 &=\ket{x}\otimes
 \begin{cases}
  \ket{P\cup\{x\},M},&x\notin P\cup M,\\
  \ket{P,M\setminus\{x\}},&x\in M,\\
  0,&x\in P.
 \end{cases}
\end{aligned}
 \label{pubcore:endpoint-set-update}
\end{equation}
Distinct retained inputs have distinct outputs, so
\(\|\mathsf E_j\|\le1\). Define the formal core in XOR coordinates by
\begin{equation}
 \widehat C:=\mathsf E_4\mathsf X_t\mathsf E_2.
 \label{pubcore:bounded-insertion-core}
\end{equation}
Lemma~\ref{pubcore:xor-recorder}, applied to the middle table, gives
\begin{equation}
 \|\widehat C\|\le2\sqrt{t+1}
 \label{pubcore:insertion-core-norm}
\end{equation}
and
\begin{equation}
 T_{\mathrm{good}}^\dagger\mathsf C T_{\mathrm{good}}
 =J_D^\dagger\Pi_{\mathrm{rot}}\widehat C\Pi_{\mathrm{rot}}J_D.
 \label{pubcore:packed-core}
\end{equation}
Indeed, each endpoint table occurs once in \(\mathsf C\), and its
packing maps commute with the other layers. Its compression therefore
gives \(\mathsf E_2\) or \(\mathsf E_4\); the middle compression is
\(\Pi_{\mathrm{rot}}\mathsf X_t\Pi_{\mathrm{rot}}\).
The endpoint updates commute with \(\Pi_{\mathrm{rot}}\).

On good records, the endpoint sets and middle labels have equal
cardinalities for each sign. A positive insertion by \(\mathsf X_t\)
therefore forces both endpoint updates to insert fresh labels
\(x,y\). Cartesian XOR injectivity in the target identifies the
new edge as \((x,y)\). A negative deletion instead forces both
endpoint updates to erase negative labels; injectivity in the source
\(X_-\times Y_-\) forces them to belong to that same edge.
The respective middle multiplicities are zero and one, so both
coefficients are one. The good source and target conditions and the
two-sided cutoffs retain exactly the allowed updates of
\(\widetilde W\), each with amplitude \(1/\sqrt N\). Thus
\begin{equation}
 J_D^\dagger\widehat C J_D=\widetilde W.
 \label{eq:raw-core-compression}
\end{equation}

\begin{proof}[Proof of Proposition~\ref{pubcore:compression}]
Lemma~\ref{pubcore:good-gram} gives
\[
 \|(I-\Pi_{\mathrm{rot}})J_D\|
 =\|(I-\Pi_{\mathrm{cf}})\mathsf H_{\mathrm{mem}}J_D\|
 \le\sqrt{\varepsilon_{\mathrm{cf},t}}.
\]
Removing the two collision projections gives
\[
 \Pi_{\mathrm{rot}}\widehat C\Pi_{\mathrm{rot}}-\widehat C
 =-(I-\Pi_{\mathrm{rot}})\widehat C\Pi_{\mathrm{rot}}
       -\widehat C(I-\Pi_{\mathrm{rot}}).
\]
By \eqref{pubcore:packed-core}, \eqref{eq:raw-core-compression}, and
\(\|J_D\|=1\),
\begin{equation}
\begin{aligned}
 \|T_{\mathrm{good}}^\dagger\mathsf C T_{\mathrm{good}}-\widetilde W\|
 &\le2\|\widehat C\|
          \|(I-\Pi_{\mathrm{rot}})J_D\|\\
 &\le4\sqrt{t+1}\sqrt{\varepsilon_{\mathrm{cf},t}}.
\end{aligned}
 \label{pubcore:compression-unnormalized}
\end{equation}
The last inequality uses Lemma~\ref{pubcore:good-gram} and
\eqref{pubcore:insertion-core-norm}. Adding the amplitude loss
\eqref{pubcore:amplitude-loss} gives
\begin{equation}
 \|T_{\mathrm{good}}^\dagger\mathsf C T_{\mathrm{good}}-W_{\mathrm{good}}\|
 \le4\sqrt{t+1}\sqrt{\varepsilon_{\mathrm{cf},t}}+\frac{2t}{N}
 \le\eta_t=\frac{8t^2}{\sqrt N},
 \label{pubcore:detailed-compression}
\end{equation}
since, for \(t,N\ge1\),
\[
 \frac{4\sqrt{2(t^3-t)}}{\sqrt N}+\frac{2t}{N}
 \le(4\sqrt2+2)\frac{t^2}{\sqrt N}
 <\frac{8t^2}{\sqrt N}.
\]

For the conjugate sampled core, the same lemma gives
\(\mathsf E_4^\dagger\mathsf X_t^\dagger\mathsf E_2^\dagger\).
The same calculation
gives the bound for the directly defined
\(\overline W_{\mathrm{good}}
=\mathsf C_{-,\to}+\mathsf C_{+,\leftarrow}^\dagger\).
Adjoints give the inverse and transpose bounds, while
\eqref{pubcore:contractive-encoding} bounds idle by
\(\varepsilon_{\mathrm{cf},t}\le\eta_t\).
Orthogonality of the five control blocks makes the controlled error
the maximum branch error, at most \(\eta_t\). All branches use the
same encoding and global phase, and act as identity on
untouched reference registers.
\end{proof}

\section{Boundary twirls and defect calculations}\label{app:prefix}

The boundary estimates in Section~\ref{sec:pub-boundary-details}
follow from the two-copy twirls below. We use the normalized symmetric
embedding~\eqref{pubcore:record-relation-state}.

\subsection{Phase--Hadamard twirls}
The phases below have order \(q\ge3\), or are independent uniform
unit-circle phases.

For \(a\in\mathbb F_2^n\), define
\[
 Q_a^\oplus=\sum_u\ket{u,u\oplus a}\bra{u,u\oplus a},\qquad
 \mathsf Z_a=\sum_p(-1)^{a\cdot p}\ket{p}\bra{p},\qquad
 \ket{\Phi_a}=N^{-1/2}\sum_p(-1)^{a\cdot p}\ket{p,p}.
\]
Let \(\mathsf{Swap}\) exchange the two coordinate registers.

\begin{lemma}[Exact phase--Hadamard twirls]
\label{lem:pub-boundary-twirls}
For \(B=HF\), where \(F\) is diagonal with independent uniform phases,
\begin{align}
 \E_F\Bigl[(B\otimes B)^\dagger Q_a^\oplus(B\otimes B)\Bigr]
   &=\frac{I+(\mathsf Z_a\otimes\mathsf Z_a)\mathsf{Swap}-Q_0^\oplus}{N}
      \preceq\frac{2I}{N},\label{eq:pub-same-twirl}\\
 \E_F\Bigl[(B\otimes\overline B)^\dagger Q_a^\oplus
                   (B\otimes\overline B)\Bigr]
   &=\ket{\Phi_a}\bra{\Phi_a}+
                   \frac{I-Q_0^\oplus}{N},\label{eq:pub-mixed-twirl}\\
 \E_F\Bigl[(B\otimes\overline B)^\dagger
   (Q_0^\oplus-\ket{\Phi_0}\bra{\Phi_0})(B\otimes\overline B)\Bigr]
   &=\frac{I-Q_0^\oplus}{N}.\label{eq:pub-deletion-twirl}
\end{align}
The same bounds hold after conjugating both tensor factors or
interchanging the two mixed factors.
\end{lemma}

\begin{proof}
Write the diagonal phase operator as
\[
 F=\sum_p\xi_p\ket{p}\bra{p},
\]
where \(\xi_p\) is the random phase at address \(p\).
For \(k=0,\pm1,\pm2\), phase orthogonality gives
\(\E_F\bigl[\xi_p^k\bigr]=\mathbf1[k=0]\). By independence, each
product below survives averaging exactly when every address occurs
equally often with each sign. The two possible pairings overlap when all four indices agree,
so we subtract that case once:
\[
\begin{aligned}
 \E_F\Bigl[\overline{\xi_p\xi_r}\,\xi_{p'}\xi_{r'}\Bigr]
 &=\mathbf1[p=p',r=r']+\mathbf1[p=r',r=p']
       -\mathbf1[p=r=p'=r'],\\
 \E_F\Bigl[\overline{\xi_p}\xi_r\xi_{p'}\overline{\xi_{r'}}\Bigr]
 &=\mathbf1[p=p',r=r']+\mathbf1[p=r,p'=r']
       -\mathbf1[p=r=p'=r'].
\end{aligned}
\]
Expanding the Hadamards before applying the projector gives
\[
\begin{aligned}
 &\bra{p,r}(H\otimes H)Q_a^\oplus(H\otimes H)\ket{p',r'}\\
 &=\frac1{N^2}\sum_{u,v,u',v'}
    (-1)^{p\cdot u+r\cdot v+p'\cdot u'+r'\cdot v'}
       \bra{u,v}Q_a^\oplus\ket{u',v'}\\
 &=\frac1{N^2}\sum_u
    (-1)^{u\cdot(p\oplus p')+(u\oplus a)\cdot(r\oplus r')}\\
 &=\frac{(-1)^{a\cdot(r\oplus r')}}N
          \mathbf1[p\oplus r=p'\oplus r'].
\end{aligned}
\]
The projector retains exactly \(u'=u\) and \(v'=v=u\oplus a\);
the final sum vanishes unless \(p\oplus r=p'\oplus r'\).
Multiplying this entry by the two phase moments gives
\eqref{eq:pub-same-twirl} and~\eqref{eq:pub-mixed-twirl}, respectively.
The crossed pairing gives the signed swap in the first formula and
\(\ket{\Phi_a}\bra{\Phi_a}\) in the second. In each moment, the
overlap of the two pairings gives the subtracted diagonal projector.
For \(S_a=(\mathsf Z_a\otimes\mathsf Z_a)\mathsf{Swap}\),
\[
 S_a^\dagger=S_a,\qquad S_a^2=I,
 \qquad I+S_a-Q_0^\oplus\preceq I+S_a\preceq2I.
\]
For the deletion identity, the maximally entangled vector is fixed by
\(B\otimes\overline B\). Subtracting its projector from the mixed twirl gives
\[
 (B\otimes\overline B)\ket{\Phi_0}=\ket{\Phi_0},
 \qquad
 \E_F\Bigl[(B\otimes\overline B)^\dagger Q_0^\oplus
                (B\otimes\overline B)\Bigr]
       -\ket{\Phi_0}\bra{\Phi_0}
 =\frac{I-Q_0^\oplus}{N}.
\]
Complex conjugation and exchange of the two copies give the remaining
versions.
\end{proof}

\subsection{Domain and image defects}
For \(U\in\{W,\overline W\}\) and
\(\mathcal J\in\{\mathcal D,\mathcal I\}\), put
\[
\begin{aligned}
 \Pi_{\mathcal D(U)}&=U^\dagger U,
 &\Pi_{\mathcal I(U)}&=UU^\dagger,\\
 \Delta^{\mathcal J}_{\ell,r}(U)
 &=\Pi_{\mathrm{bij},\ell,r}-\Pi_{\mathcal J(U)}\Pi_{\ell,r},
 &\Delta^{\mathcal J}_{\ell,r}&:=\Delta^{\mathcal J}_{\ell,r}(W).
\end{aligned}
\]
The domain and image projectors commute with each sign-length
projector \(\Pi_{\ell,r}\) and lie in the bijective sector. Hence
\[
 \bigl(\Delta^{\mathcal J}_{\ell,r}(U)\bigr)^2
   =\Delta^{\mathcal J}_{\ell,r}(U),\qquad
 0\preceq\Delta^{\mathcal J}_{\ell,r}(U)
       \preceq\Pi_{\mathrm{bij},\ell,r}.
\]
We use the ordered-register inequalities
\eqref{eq:pub-domain-input}--\eqref{eq:pub-image-input}, stated in
Appendix~\ref{app:controlled-recording} and obtained from
\cite[Corollary~11.1 and Claim~26]{MH}.

\begin{proof}[Proof of Lemma~\ref{lem:pub-boundary-defects}]
Fix \(m=\ell+r\), with \(1\le m\le N/2\).
For the domain inequality, the query and positive domain slots use
\(G=HF_1\), while the negative domain slots use \(\overline G\).
The image rotations cancel on the right side of
\eqref{eq:pub-domain-input}. Each positive slot contributes at most
\(2I/N\) by~\eqref{eq:pub-same-twirl}; each negative slot contributes
at most \(I/N\) by~\eqref{eq:pub-deletion-twirl}. Including the remainder
in~\eqref{eq:pub-domain-input} gives
\[
 \E_b\Bigl[(G\otimes Q_b)^\dagger\Delta^{\mathcal D}_{\ell,r}
                   (G\otimes Q_b)\Bigr]
 \preceq\frac{N}{N-m+1}
       \left[\frac{2\ell+r}{N}+2r\sqrt{m/N}\right]I.
\]
For the image inequality, the query and negative image slots use
\(D^\dagger=\overline{HF_5}\), while the positive image slots use
\(D^T=HF_5\). The domain rotations cancel on the right side of
\eqref{eq:pub-image-input}, giving
\[
 \E_b\Bigl[(D^\dagger\otimes Q_b)^\dagger\Delta^{\mathcal I}_{\ell,r}
                   (D^\dagger\otimes Q_b)\Bigr]
 \preceq\frac{N}{N-m+1}
       \left[\frac{\ell+2r}{N}+2\ell\sqrt{m/N}\right]I.
\]
Using \(\ell,r\le m\), \(N/(N-m+1)\le2\), and
\(1/N\le\sqrt{m/N}\), both coefficients are bounded by
\[
 \frac{N}{N-m+1}
       \left[\frac{2m}{N}+2m\sqrt{m/N}\right]
 \le\frac{4m}{N}+4m\sqrt{m/N}
 \le8m\sqrt{m/N}.
\]
The barred-recorder and deletion bounds follow by symmetry: relabel the
coefficient indices \(L,R\) and conjugate the boundary tables for the barred recorder;
exchange domain and image for the adjoint. The boundary distribution
is invariant under conjugation, so the same bound holds for all four
nonidle branches in Table~\ref{tab:controlled-boundaries}.
The zero-length defect is zero. Since every term preserves \(\ell,r\),
\[
 \bigoplus_{\ell+r\le i}8(\ell+r)\sqrt{(\ell+r)/N}\,I_{\ell,r}
 \preceq8i\sqrt{i/N}\,I=\gamma_i I.
\]
\end{proof}

\subsection{Additional XOR collisions}
\begin{proof}[Proof of Lemma~\ref{lem:pub-boundary-good-mass}]
Fix an ordered \((\ell,r)\) block and a boundary-independent vector
\(\ket{v}=P_{\mathrm{nc}}\ket{v}\) in this block. For a positive edge
\(L_j\) and a negative edge \(R_k\), write
\(X=(X_{L_j},X_{R_k})\), \(Y=(Y_{L_j},Y_{R_k})\), and define
\[
 B_{jk}=\sum_a(Q_a^\oplus)_X\otimes(Q_a^\oplus)_Y,
 \qquad Q_X=(Q_0^\oplus)_X,\quad Q_Y=(Q_0^\oplus)_Y,
 \quad \Theta_a=\ket{\Phi_a}\bra{\Phi_a}.
\]
The domain and image tables are independent, so
\eqref{eq:pub-mixed-twirl} gives
\begin{align}
 \E_b\Bigl[Q_b^\dagger B_{jk}Q_b\Bigr]
 &=\sum_a\left[(\Theta_a)_X+\frac{I-Q_X}{N}\right]
       \otimes\left[(\Theta_a)_Y+\frac{I-Q_Y}{N}\right]\nonumber\\
 &=\sum_a(\Theta_a)_X\otimes(\Theta_a)_Y
                  +\frac{I-Q_X\otimes Q_Y}{N}.
 \label{eq:pub-double-mixed-twirl}
\end{align}
Here we used \(\sum_a\Theta_a=Q_0^\oplus\). Moreover,
\[
 0\preceq\sum_a(\Theta_a)_X\otimes(\Theta_a)_Y
       \preceq Q_X\otimes Q_Y,
 \qquad (Q_X\otimes Q_Y)P_{\mathrm{nc}}=0.
\]
The projector \(Q_X\otimes Q_Y\) requires the two edges to have
both endpoints equal. Such a pair is excluded by \(P_{\mathrm{nc}}\),
so the rank-one terms contribute zero on this subspace. Therefore
\[
 P_{\mathrm{nc}}\E_b\Bigl[Q_b^\dagger B_{jk}Q_b\Bigr]P_{\mathrm{nc}}
 \preceq\frac1N P_{\mathrm{nc}}.
\]

For \(j<k\) and \(u<v\), define the positive rectangle projector
\[
 R^+_{jk;uv}
 =\sum_a(Q_a^\oplus)_{X_{L_j},X_{L_k}}
                 \otimes(Q_a^\oplus)_{Y_{L_u},Y_{L_v}}.
\]
For each fixed \(F_5\), equation~\eqref{eq:pub-same-twirl} gives
\[
\begin{aligned}
 \E_{F_1}\Bigl[Q_b^\dagger R^+_{jk;uv}Q_b\Bigr]
 &\preceq\frac2N\sum_a
   \Bigl[(D^T\otimes D^T)^\dagger Q_a^\oplus
                              (D^T\otimes D^T)\Bigr]_{Y_{L_u},Y_{L_v}}\\
 &=\frac2N I,
\end{aligned}
\]
since \(\sum_aQ_a^\oplus=I\). The \(X\) and \(Y\) registers are
distinct even when the two edge pairs overlap. Define \(R^-_{jk;uv}\)
by replacing \(L\) with \(R\); complex conjugation gives its bound.

A record fails the additional-goodness conditions only if at least one
of these rectangle or opposite-sign events occurs. Their projectors
are diagonal in the ordered record basis, so the union bound gives
\[
 I-P_{\mathrm{add}}
 \preceq\sum_{\substack{j<k\\u<v}}R^+_{jk;uv}
       +\sum_{\substack{j<k\\u<v}}R^-_{jk;uv}
       +\sum_{j=1}^{\ell}\sum_{k=1}^r B_{jk}.
\]
There are \(\binom\ell2^2\), \(\binom r2^2\), and \(\ell r\) terms,
respectively. Thus
\[
 \E_b\Bigl[\|(I-P_{\mathrm{add}})Q_b\ket{v}\|^2\Bigr]
 \le\frac{2[\binom\ell2^2+\binom r2^2]+\ell r}{N}\|\ket{v}\|^2.
\]
The sums commute with separate same-sign permutations. Compressing
the summed inequality to the normalized symmetric subspace therefore
preserves the bound. For \(\ell+r\le i\),
\[
 \binom\ell2^2+\binom r2^2\le\binom i2^2,
 \qquad \ell r\le\frac{i^2}{4},
 \qquad
 \frac{2\binom i2^2+i^2/4}{N}\le\frac{9i^4}{N}\quad(i\ge1).
\]
Finally, for \(\ket{v}=\sum_{\ell+r\le i}\ket{v_{\ell,r}}\),
\[
 \E_b\Bigl[\|(I-P_{\mathrm{add}})Q_b\ket{v}\|^2\Bigr]
 =\sum_{\ell+r\le i}
      \E_b\Bigl[\|(I-P_{\mathrm{add}})Q_b\ket{v_{\ell,r}}\|^2\Bigr]
 \le\frac{9i^4}{N}\|\ket{v}\|^2.
\]
The loss on the empty record is zero. Tensoring with an identity retains arbitrary
reference registers.
\end{proof}

\section{Local tools for core-table differences}\label{app:core-tools}

We prove the randomization and moment bounds used to control the
largest change under a core-bit flip. We then count the labels that
prevent a good record from extending by one edge or an opposite-sign
pair. All Hilbert spaces below may contain arbitrary auxiliary
registers. No bound depends on their dimension.
For the encoding and extension lemmas, use the construction's
independent uniform \(N\)-th root phase tables, with \(N=2^n>8t\).

\subsection{Selected characters and Hilbert moments}

The first estimate preserves a maximum by treating it as the norm on
a finite direct product. Characters are indexed by exponents modulo
the phase order; their support consists of the nonzero exponents.

\begin{lemma}[Randomizing a selected phase character]
\label{lem:selected-phase}
Let \((\xi_a)_a\) be independent uniform \(N\)-th roots of unity.
Let \(\ket{F_a(\xi)}\) be polynomials with values in a complex normed
space \(X\). Suppose every character in \(\ket{F_a(\xi)}\) has exactly
\(d\ge1\) coordinates in its support and has the same exponent
\(k\not\equiv0\pmod N\) at \(a\). For independent uniform
\(N\)-th roots \((\delta_a)_a\),
\begin{equation}
 \left\|\sum_a\ket{F_a(\xi)}\right\|_{L^2(\xi;X)}
 \le\mathrm e d
       \left\|\sum_a\delta_a^k\ket{F_a(\xi)}\right\|_{L^2(\xi,\delta;X)}.
 \label{eq:selected-phase-randomization}
\end{equation}
\end{lemma}
\begin{proof}
For \(d\ge2\), include each coordinate in a random set \(A\)
independently with probability \(1/d\), and put
\[
 \ket{T_A(\xi)}=\sum_{a\in A}
   \mathbb E_{\xi_{A\setminus\{a\}}}[\ket{F_a(\xi)}].
\]
A character survives exactly when its support meets \(A\) only at
\(a\). Hence
\[
 \mathbb E_A[\ket{T_A(\xi)}]
 =d^{-1}(1-d^{-1})^{d-1}\sum_a\ket{F_a(\xi)}.
\]
For fixed \(A\), write
\(\ket{T_A(\xi)}=\sum_{a\in A}\xi_a^k
\ket{g_{a,A}(\xi_{A^c})}\).
Replacing \(\xi_a^k\) by \(\delta_a^k\) preserves its distribution.
Also,
\[
 \sum_{a\in A}\delta_a^k\ket{g_{a,A}(\xi_{A^c})}
 =\mathbb E_{\xi_A}\left[
    \sum_{a\in A}\delta_a^k\xi_a^{-k}\ket{F_a(\xi)}\right].
\]
Jensen and the measure-preserving substitution
\(\delta_a\mapsto\delta_a\xi_a\) bound this norm by
\(\|\sum_{a\in A}\delta_a^k\ket{F_a(\xi)}\|_{L^2}\).
Since \(\mathbb E[\delta_a^k]=0\), averaging the fresh phases outside
\(A\) bounds the latter by the norm of the full randomized sum.
Minkowski and
\(d(1-d^{-1})^{-(d-1)}\le\mathrm e d\) prove the claim.
For \(d=1\), take \(A\) to contain every coordinate; the same
argument gives constant one.
\end{proof}

We will split terms by their nonzero selected exponent and support
size before applying this lemma. A zero selected exponent does not
satisfy its hypothesis.

\begin{lemma}[Randomizing a selected sign]\label{lem:selected-sign}
Let \(\ket{F_x(\varepsilon)}\) be homogeneous multilinear polynomials
of degree \(s\ge1\) in independent fair signs, with values in a
normed space \(X\). Suppose every monomial of \(\ket{F_x(\varepsilon)}\) contains
\(\varepsilon_x\). For independent fair signs \(\delta_x\),
\begin{equation}
 \left\|\sum_x\ket{F_x(\varepsilon)}\right\|_{L^2(\varepsilon;X)}
 \le\mathrm e s
       \left\|\sum_x\delta_x\ket{F_x(\varepsilon)}\right\|_{L^2(\varepsilon,\delta;X)}.
 \label{eq:rademacher-contraction}
\end{equation}
\end{lemma}
\begin{proof}
Use the preceding proof with \(N=2\), \(d=s\), and \(k=1\).
All scalar multipliers are real, so this also applies to real normed
spaces.
\end{proof}

\begin{lemma}[Hilbert moment bound]
\label{lem:hilbert-chaos-moment}
For \(1\le j\le\ell\), let \((\tau_a^{(j)})_a\) be independent
unit-modulus random variables whose distributions are invariant under
multiplication by \(-1\). The families are mutually independent; let \(\tau\) denote their full array.
For fixed complex Hilbert vectors \(\ket{g_{a_1,\ldots,a_\ell}}\), put
\[
 \ket{G(\tau)}=\sum_{a_1,\ldots,a_\ell}
     \tau_{a_1}^{(1)}\cdots\tau_{a_\ell}^{(\ell)}
     \ket{g_{a_1,\ldots,a_\ell}}.
\]
For every even integer \(p\ge2\),
\begin{equation}
 \left(\mathbb E[\|\ket{G(\tau)}\|^p]\right)^{2/p}
 \le(p-1)^\ell
       \sum_{a_1,\ldots,a_\ell}\|\ket{g_{a_1,\ldots,a_\ell}}\|^2.
 \label{eq:hilbert-chaos-moment}
\end{equation}
For \(m\) such sums \(\ket{G_u(\tau)}\) using the same fresh variables,
\begin{equation}
 \mathbb E[\max_{u\le m}\|\ket{G_u(\tau)}\|^2]
 \le m^{2/p}(p-1)^\ell
       \max_{u\le m}\sum_{a_1,\ldots,a_\ell}
                \|\ket{g^{(u)}_{a_1,\ldots,a_\ell}}\|^2.
 \label{eq:chaos-maximum}
\end{equation}
\end{lemma}
\begin{proof}
First take one family of fair signs and write \(p=2k\).
In the expansion of
\(\mathbb E[\|\sum_a\varepsilon_a\ket{g_a}\|^{2k}]\), a term
survives only if every sign index occurs an even number of times.
Pair its \(2k\) occurrences into pairs with equal indices. Bound
coefficient inner products by products of norms. There are
\((2k-1)!!\) pairings, and each pairing contributes at most
\((\sum_a\|\ket{g_a}\|^2)^k\). Thus
\[
 \mathbb E\left[\left\|\sum_a\varepsilon_a\ket{g_a}\right\|^{2k}\right]
 \le(2k-1)!!\left(\sum_a\|\ket{g_a}\|^2\right)^k
 \le(p-1)^k\left(\sum_a\|\ket{g_a}\|^2\right)^k.
\]
For the stated unit-modulus variables, condition on representatives
modulo sign. The remaining factors are independent fair signs; the
representatives multiply coefficients by scalars of modulus one.
This proves the one-family estimate. Apply it successively to the
\(\ell\) families, using the Hilbert direct sum of coefficient
vectors at each step. This gives~\eqref{eq:hilbert-chaos-moment}.
Finally,
\[
 \mathbb E[\max_u\|\ket{G_u(\tau)}\|^2]
 \le\left(\sum_u\mathbb E[\|\ket{G_u(\tau)}\|^p]\right)^{2/p},
\]
which proves~\eqref{eq:chaos-maximum}.
\end{proof}

The finite phase alphabet satisfies the symmetry needed here.
If \(N=2^n\), \(0<|k|<N\), and \(\delta\) is a uniform
\(N\)-th root, then \(\delta^k\) is uniform on the roots of order
\(N/\gcd(N,k)\). This order is even, so conditioning modulo sign
leaves an independent fair sign. This includes even \(k\); the power
map need not be a permutation of all \(N\)-th roots.

\subsection{Sparse masks and contractions}

\begin{lemma}[Sparse masks before a contraction]
\label{lem:sparse-mask-contraction}
Let \(\mathcal A:\mathcal K\to L^2(\Omega;\mathcal H)\) be a
linear contraction between Hilbert spaces for a fixed probability
measure on \(\Omega\).
Let \((P_x)_x\) be a finite family of orthogonal projections with
\(\sum_x(I-P_x)\preceq K I\). For a fixed \(\ket{c}\in\mathcal K\),
\begin{equation}
 \mathbb E_\omega[\max_x\|(\mathcal AP_x\ket{c})(\omega)\|^2]
 \le2(1+K)\|\ket{c}\|^2.
 \label{eq:sparse-mask-contraction}
\end{equation}
The map is the same for every \(x\). The vector and projections are
independent of \(\omega\); they may depend on fixed outside data.
The projections need not be mutually orthogonal.
\end{lemma}
\begin{proof}
Linearity and the squared triangle inequality give
\[
 \begin{aligned}
 &\mathbb E_\omega[\max_x\|(\mathcal AP_x\ket{c})(\omega)\|^2]\\
 &\quad\le2\mathbb E_\omega[\|(\mathcal A\ket{c})(\omega)\|^2]
       +2\sum_x\mathbb E_\omega[\|(\mathcal A(I-P_x)\ket{c})(\omega)\|^2]\\
 &\quad\le2\|\ket{c}\|^2+2\bra{c}\sum_x(I-P_x)\ket{c}
 \le2(1+K)\|\ket{c}\|^2.
 \end{aligned}
\]
\end{proof}
A preliminary projection \(P^0\), independent of \(x,\omega\), may
be absorbed into \(\ket{c}\). In the diagonal applications below,
set \(P_x=I\) off its range; the exclusion bound need only hold on
that range. Marked labels remain as direct-sum controls.

\begin{corollary}[Averaging one phase coordinate]
\label{cor:whole-coordinate-mean}
Let \(\ket{F(\xi)}\) be a Hilbert-valued polynomial in independent
uniform phase coordinates, with at most \(r\) occupied coordinates
in each product character. Let \(\mathbb E_z[\cdot]\) average the
whole phase value at \(z\). Then
\begin{equation}
 \mathbb E_\xi[\max_z\|\mathbb E_z[\ket{F(\xi)}]\|^2]
 \le2(r+1)\mathbb E_\xi[\|\ket{F(\xi)}\|^2].
 \label{eq:whole-coordinate-mean}
\end{equation}
\end{corollary}
\begin{proof}
In the orthonormal product-character coefficient space, let \(P_z\)
retain characters unoccupied at \(z\). Each character is rejected
at most \(r\) times, so \(\sum_z(I-P_z)\preceq rI\).
Evaluation is an isometry into \(L^2(\xi)\), and \(P_z\) evaluates
to \(\mathbb E_z\). Apply Lemma~\ref{lem:sparse-mask-contraction}.
This includes Hilbert direct sums. For additional independent
tables, condition on them first and then average.
\end{proof}

Let \(\vec f=(f_2,f_3,f_4)\), with phases \((\xi_2,\xi_3,\xi_4)\).
Write \(\mathcal T_r\) for the rescaled coefficient map of the common
core encoding on good records of size \(r\le t\), extended by zero
outside the good subspace. Equation~\eqref{pubcore:contractive-encoding}
gives
\begin{equation}
 \mathbb E_{\vec f}[\|(\mathcal T_r\ket{w})(\vec f)\|^2]
 \le\|\ket{w}\|^2,
 \label{eq:residual-encoding-contraction}
\end{equation}
including arbitrary auxiliary registers and marked-label direct sums.
Each encoded character uses at most \(r\) middle locations.

\paragraph{Contraction on marked image labels.}
For square-integrable functions \(\ket{F_y(\vec f)}\) and
\(\sigma\in\{+1,-1\}\), conditional-expectation contractivity and
pointwise unitarity of \(H\operatorname{diag}(\xi_4^\sigma)\) give
\begin{equation}
 \begin{aligned}
 &\mathbb E_{\vec f}\left[\sum_z\left\|\mathbb E_z\left[
      \sum_yH_{y,z}\xi_4(y)^\sigma\ket{F_y(\vec f)}
                  \right]\right\|^2\right]\\
 &\quad\le\mathbb E_{\vec f}\left[\sum_z\left\|
      \sum_yH_{y,z}\xi_4(y)^\sigma\ket{F_y(\vec f)}
                  \right\|^2\right]
 =\mathbb E_{\vec f}\left[\sum_y\|\ket{F_y(\vec f)}\|^2\right].
 \end{aligned}
 \label{eq:marked-image-contraction}
\end{equation}
Here \(\mathbb E_z\) averages the whole middle word \(f_3(z)\).
For two image marks, use
\(H\operatorname{diag}(\xi_4^\sigma)\otimes
H\operatorname{diag}(\xi_4^{-\sigma})\), followed componentwise by
\(\mathbb E_{\{z,z'\}}\) and the projection onto \(z\ne z'\).
The same bound holds with input energy
\(\mathbb E_{\vec f}[\sum_{y,y'}\|\ket{F_{y,y'}(\vec f)}\|^2]\).
The input functions may depend on \(\xi_4\): unitarity holds for
each phase sample, even though the two factors share that table.
Composing these maps with \(\mathcal T_r\) gives contractions with
all spectator marks and auxiliary registers retained.

\subsection{Good extensions}

For a good residual record \((L,R)\) of size \(r\), use the sets
\(X_\sigma(L,R),Y_\sigma(L,R),D_\sigma(L,R)\) from Section~\ref{pubcore:central},
with \(\sigma\in\{+1,-1\}\), and put \(X_{\rm all}(L,R)=X_+(L,R)\cup X_-(L,R)\).
The projections below
act diagonally on \(\ket{L,R}\) and as the identity on any auxiliary
space.

\begin{lemma}[Good-extension masks]\label{lem:good-extension-masks}
Fix either \((y,\sigma)\) or \((d,y,y',\sigma)\), independently of
the core tables. Add \((x,y,\sigma)\) in the first case, and add
\((x,y,\sigma),(x\oplus d,y',-\sigma)\) in the second.
On good residual records, the condition that the extension is good
can be written as \(P_xP^0\),
where \(P^0\) is independent of the new domain label \(x\).
If the resulting size is at most \(s\), the remaining masks satisfy
\begin{equation}
 \sum_x(I-P_x)\preceq K_s I,
 \qquad K_s=16(s+1)^3.
 \label{eq:good-extension-mask-count}
\end{equation}
The same statement holds on the Hilbert direct sum over the fixed
marks.
\end{lemma}
\begin{proof}
\emph{One edge.} Fix \(y,\sigma\) and add \((x,y,\sigma)\).
Let \(P^0_{y,\sigma}\) require that \(y\) is absent from all old
image labels and that XOR is injective on
\(X_\sigma(L,R)\times(Y_\sigma(L,R)\cup\{y\})\). These conditions are
independent of \(x\). On their support, the remaining forbidden
labels are exactly
\begin{equation}
 \begin{aligned}
 \mathcal B_{y,\sigma}(L,R)
 ={}&X_{\rm all}(L,R)\\
 &\cup\{a\oplus b\oplus c:
      a\in X_\sigma(L,R),\ b\ne c\in Y_\sigma(L,R)\cup\{y\}\}\\
 &\cup\{y\oplus d:d\in D_{-\sigma}(L,R)\}.
 \end{aligned}
 \label{eq:one-edge-forbidden-labels}
\end{equation}
The first set excludes repeated domains. The second excludes new
same-sign Cartesian XOR collisions involving \(x\); collisions
between two old domains were excluded by \(P^0_{y,\sigma}\).
The third excludes an actual XOR label of the opposite sign.
These are precisely the three goodness conditions. Their union has
size at most
\[
 r+r\binom{r+1}{2}+r\le4(r+2)^3.
\]
Let \(P_x\) keep the basis record when
\(x\notin\mathcal B_{y,\sigma}(L,R)\), and set \(P_x=I\) off
the range of \(P^0_{y,\sigma}\). The full extension projection on
good residuals is \(P_xP^0_{y,\sigma}\).

\emph{An opposite-sign pair.} Fix \(d,y,y',\sigma\) and add
\[
 e=(x,y,\sigma),\qquad e'=(x\oplus d,y',-\sigma).
\]
Let \(P^0_{d,y,y',\sigma}\) impose both one-image preliminary
conditions and
\[
 d\ne0,\qquad y\ne y',\qquad d\ne y\oplus y'.
\]
These exclude repeated new domains, repeated new images, and equal
new actual XOR labels. Each condition is independent of \(x\).
On this range the remaining failures are exactly
\begin{equation}
 x\in\mathcal B_{y,\sigma}(L,R)
 \quad\text{or}\quad
 x\oplus d\in\mathcal B_{y',-\sigma}(L,R).
 \label{eq:pair-forbidden-labels}
\end{equation}
There is no additional same-sign condition between the two new
edges, since their signs differ. Hence at most \(8(r+2)^3\) labels
are forbidden. Define \(P_x\) by their complement and set it to
identity off the preliminary range, as above.

For one edge, \(r+1\le s\); for a pair, \(r+2\le s\).
Thus \(K_s\) bounds both counts. Since the masks are diagonal,
counting forbidden \(x\)'s proves~\eqref{eq:good-extension-mask-count}
on each basis record. The same bound holds in every marked-label
block, proving the direct-sum statement.
\end{proof}

\section{Changes under an endpoint-table bit flip}\label{app:endpoint-differences}

We bound changes in the second and fourth phase tables. Removing the
affected edge exposes one source coordinate of the boundary rotation.
If its signed occupation is nonzero, we randomize that source phase.
If the occupation is zero, we remove an opposite-sign partner at the
same source coordinate. The two boundary phases then cancel, leaving
a Hadamard coefficient that supplies the same inverse power of \(N\).

Write \(S=(b,\vec{f})\), where \(b=(f_1,f_5)\) and
\(\vec{f}=(f_2,f_3,f_4)\), with phase tables
\((\xi_2,\xi_3,\xi_4)\). All phase values are independent uniform
\(N\)-th roots of unity, with \(N=2^n\). For a bit \(j\) at location
\(x\) in table \(i\), let
\[
 \ket{q_v^{\Delta(i,x,j)}(S)}
 =\frac{\ket{q_v(S)}-\ket{q_v(S^{\oplus(i,x,j)})}}2.
\]
The other bits and tables remain fixed. We use the rescaled state
\(\ket{q_v(S)}\) from \eqref{eq:record-feature-evaluator}.

\begin{proposition}[Maximum change under an endpoint-table bit flip]
\label{prop:maximal-endpoint-bit-change}
Let \(n,t\ge1\), assume \(3(2t)^{16}\le N\), and let
\(\ket{v}\) be independent of all phase tables with record size at
most \(t\). Its coefficients may lie in any auxiliary Hilbert space.
Put
\[
 p_\star=2\lceil\ln(24nN)\rceil,\qquad K_t=16(t+1)^3.
\]
For each \(i\in\{2,4\}\),
\begin{equation}
 \mathbb E_S\left[\max_{x,j}
       \|\ket{q_v^{\Delta(i,x,j)}(S)}\|^2\right]
 \le
 \frac{24\mathrm e^7(t+1)^{10}(1+K_t)(p_\star-1)^3}{N}
       \|\ket{v}\|^2.
 \label{eq:endpoint-amplitude-energy}
\end{equation}
The maximum ranges over the \(nN\) bits of table \(i\). The expectation
averages all five tables. No forest-support assumption or bound on
the auxiliary dimension is needed.
\end{proposition}

We first treat table \(2\) at a fixed total record size \(s\).
The selected source coordinate has either nonzero or zero signed
occupation. We bound these parts separately, then sum over their
labels and over \(s\). Throughout, \(\mathbb E_z[\cdot]\) averages the
whole middle phase \(\xi_3(z)\), and
\(\mathbb E_{\{z,z'\}}[\cdot]\) averages both indicated middle phases.

\subsection{Removing the affected edge}

For a signed edge \((p,q,\sigma)\), with \(\sigma\in\{+1,-1\}\),
define the annihilation map by
\begin{equation}
 a_{p,q,\sigma}\ket{L,R}
 =\sqrt{m_{p,q,\sigma}}\,
       \ket{(L,R)\setminus(p,q,\sigma)}.
 \label{eq:endpoint-annihilation}
\end{equation}
Here \(m_{p,q,\sigma}\) is the multiplicity in the indicated sign
bank; the result is zero when this multiplicity is zero. These are algebraic maps between the normalized symmetric
sectors of \eqref{pubcore:record-relation-state}. They do not define
additional oracle operations.

Put
\[
 B_{i,\sigma}(b)=H\operatorname{diag}(\xi_i^\sigma),
 \qquad i\in\{1,5\}.
\]
In \eqref{eq:common-boundary-rotation}, the positive-edge rotation is
\((HF_1)\otimes(HF_5)\), and the negative-edge rotation is its
entrywise conjugate. Write \(Q_b^{(s)}\) for its restriction to total
size \(s\), including every positive and negative edge-count sector.
The normalized symmetric action gives
\begin{equation}
 a_{x,y,\sigma}Q_b^{(s)}
 =\sum_{p,q}(B_{1,\sigma}(b))_{x,p}(B_{5,\sigma}(b))_{y,q}
          Q_b^{(s-1)}a_{p,q,\sigma}.
 \label{eq:endpoint-rotated-annihilation}
\end{equation}
Every map acts as the identity on the original auxiliary space.

Recall the residual encoding \(\mathcal T_r\) in
\eqref{eq:residual-encoding-contraction}. It maps a good vector of
size \(r\) to its probability-normalized core coefficient function,
and is zero outside the good subspace. By
\eqref{pubcore:contractive-encoding},
\[
 \mathbb E_{\vec{f}}[\|(\mathcal T_r\ket{w})(\vec{f})\|^2]\le\|\ket{w}\|^2.
\]
For an edge \(e\), let \(P_e\) project onto residual records that
become good after adding \(e\). Define \(P_{e,e'}\) in the same way
for two added edges. These projections include residual goodness.
They are the extension projections in
Lemma~\ref{lem:good-extension-masks}.

For a table-\(2\) bit \((x,j)\), put
\[
 \beta_{x,j}^\sigma(\vec{f})
 =\frac{\xi_2(x)^\sigma-
             \bigl(\xi_2^{\oplus(x,j)}(x)\bigr)^\sigma}{2},
 \qquad |\beta_{x,j}^\sigma(\vec{f})|\le1.
\]
In a good output record, at most one edge has domain \(x\).
Remove this edge and separate its middle assignment \(z\). For a
size-\(s\) input \(\ket{v_s}\), the packed-column formula
\eqref{pubcore:packed-column} gives the exact identity
\begin{equation}
 \begin{aligned}
 \ket{q_{v_s}^{\Delta(2,x,j)}(S)}
 =\sum_{\sigma,y,z}
 &\beta_{x,j}^\sigma(\vec{f})\xi_4(y)^\sigma
       H_{x\oplus y,z}\xi_3(z)^\sigma\\
 &{}\cdot\mathbb E_z\left[
       \bigl(\mathcal T_{s-1}P_{(x,y,\sigma)}
       a_{x,y,\sigma}Q_b^{(s)}\ket{v_s}\bigr)(\vec{f})\right].
 \end{aligned}
 \label{eq:endpoint-bit-expansion}
\end{equation}
The average keeps exactly the residual middle assignments that avoid
\(z\). Each occupied middle coordinate has exponent \(+1\) or \(-1\),
whose mean is zero. The selected edge has multiplicity one in a good
record, so its annihilation coefficient is one. The existing
symmetrization and packing normalizations introduce no further factor.

Define the commuting source operators
\[
 N_{p,\sigma}=\sum_q a_{p,q,\sigma}^\dagger a_{p,q,\sigma},
 \qquad C_p=N_{p,+}-N_{p,-}.
\]
A source basis record acquires the boundary character
\(\prod_p\xi_1(p)^{C_p}\). We split each selected source term in
\eqref{eq:endpoint-rotated-annihilation} according to whether
\(C_p\ne0\) or \(C_p=0\). These are occupations of the unrotated
source record.

\subsection{Nonzero source occupation}

Fix a sign \(\sigma\), a nonzero integer \(k\) with \(|k|\le s\),
and \(1\le d\le s\). Let \(\Pi_d\) project onto source records with
exactly \(d\) nonzero domain occupations \(C_p\), and put
\begin{equation}
 \ket{h_{p,y}^{\sigma,d,k}(b)}
 =\sum_q(B_{5,\sigma}(b))_{y,q}
       a_{p,q,\sigma}\Pi_d\mathbf1[C_p=k]\ket{v_s}.
 \label{eq:endpoint-active-state}
\end{equation}
For every boundary sample, unitarity of \(B_{5,\sigma}(b)\) gives
\begin{equation}
 \begin{aligned}
 A_{\sigma,d,k}
 &:=\sum_{p,y}\|\ket{h_{p,y}^{\sigma,d,k}(b)}\|^2\\
 &=\sum_{p,q}
       \|a_{p,q,\sigma}\Pi_d\mathbf1[C_p=k]\ket{v_s}\|^2.
 \end{aligned}
 \label{eq:endpoint-active-mass}
\end{equation}
In particular, this mass is independent of the boundary sample, and
\[
 \sum_{\sigma,d,k\ne0}A_{\sigma,d,k}\le s\|\ket{v_s}\|^2.
\]
Substituting \eqref{eq:endpoint-rotated-annihilation} into
\eqref{eq:endpoint-bit-expansion}, the corresponding part is
\begin{equation}
 \begin{aligned}
 \ket{A^{\sigma,d,k}_{x,j}(S)}
 =\sum_{p,y,z}
 &H_{x,p}H_{x\oplus y,z}\xi_1(p)^\sigma
       \beta_{x,j}^\sigma(\vec{f})\xi_4(y)^\sigma\xi_3(z)^\sigma\\
 &{}\cdot\mathbb E_z\left[
       \bigl(\mathcal T_{s-1}P_{(x,y,\sigma)}Q_b^{(s-1)}
       \ket{h_{p,y}^{\sigma,d,k}(b)}\bigr)(\vec{f})\right].
 \end{aligned}
 \label{eq:endpoint-active-expansion}
\end{equation}

\paragraph{Randomization.}
Group by \(p\) in \eqref{eq:endpoint-active-expansion}. The factor
\(\xi_1(p)^\sigma\) restores the annihilated occupation, so every
source character has support size \(d\) and exponent \(k\ne0\) at
\(p\). Lemma~\ref{lem:selected-phase}, in the maximum norm over
\((x,j)\), inserts a fresh multiplier \(\delta_p^k\) at \(L^2\)
cost \(\mathrm e d\). Next randomize only the selected middle sign
at \(z\). Collision freedom gives degree \(s\) in the original
middle signs, so Lemma~\ref{lem:selected-sign} costs \(\mathrm e s\).
The total cost is at most \(\mathrm e^2s^2\). Apply these inequalities
before fixing the original tables.

\paragraph{The image sum is a contraction.}
Fix \(b\). On the direct sum over \((p,y)\), put
\[
 \ket{c_{p,y}(b)}=P^0_{y,\sigma}P_{\mathrm{good},s-1}
 Q_b^{(s-1)}\ket{h_{p,y}^{\sigma,d,k}(b)},\qquad
 \ket{c(b)}=\bigoplus_{p,y}\ket{c_{p,y}(b)}.
\]
Here \(P_{\mathrm{good},r}\) is the good-record projection at size
\(r\). Lemma~\ref{lem:good-extension-masks} supplies the preliminary
projection \(P^0_{y,\sigma}\) and the remaining masks \(P_x\), with
\(\sum_x(I-P_x)\preceq K_sI\). Thus \(P_x\ket{c(b)}\) is the
full extension projection, block by block, and
\(\|\ket{c(b)}\|^2\le A_{\sigma,d,k}\). This vector is independent
of \(x,\vec f\); residual goodness is imposed after the boundary rotation.

Define the same map for every \(x,j\):
\begin{equation}
 (\mathcal A_{1,\sigma}\ket{w})_{p,z}(\vec f)
 =\mathbb E_z\left[\sum_y H_{y,z}\xi_4(y)^\sigma
                  (\mathcal T_{s-1}\ket{w_{p,y}})(\vec f)\right].
 \label{eq:endpoint-one-image-map}
\end{equation}
It is a contraction by \eqref{eq:marked-image-contraction} and
\eqref{eq:residual-encoding-contraction}. In particular, the residual
functions may depend on \(\xi_4\). The identity
\(H_{x\oplus y,z}=(-1)^{x\cdot z}H_{y,z}\) makes the randomized
sum \(\sum_{p,z}\delta_p^k\varepsilon_z\ket{g_{x,j;p,z}(S)}\), where
\[
 \ket{g_{x,j;p,z}(S)}
 =H_{x,p}\xi_1(p)^\sigma\beta_{x,j}^\sigma(\vec f)
   \xi_3(z)^\sigma(-1)^{x\cdot z}
   (\mathcal A_{1,\sigma}P_x\ket{c(b)})_{p,z}(\vec f).
\]
The selected middle phase stays outside \(\mathbb E_z\); the image
phases are inside. Hence
\begin{equation}
 \max_{x,j}\sum_{p,z}\|\ket{g_{x,j;p,z}(S)}\|^2
 \le\frac1N\max_x
       \|(\mathcal A_{1,\sigma}P_x\ket{c(b)})(\vec f)\|^2.
 \label{eq:endpoint-active-coefficients}
\end{equation}
Apply Lemma~\ref{lem:sparse-mask-contraction} at fixed \(b\), then
average \(b\), to obtain
\begin{equation}
 \mathbb E_S\left[\max_{x,j}\sum_{p,z}
             \|\ket{g_{x,j;p,z}(S)}\|^2\right]
 \le\frac{2(1+K_s)}N A_{\sigma,d,k}.
 \label{eq:endpoint-active-mask-bound}
\end{equation}
The masks precede \(\mathcal A_{1,\sigma}\); none is commuted through
its image transform.

Since \(N\) is a power of two and \(0<|k|\le s<N\), the fresh
phases \(\delta_p^k\) have even order and sign-symmetric distributions.
Conditional on the original tables, \eqref{eq:chaos-maximum} for
these phases and the fresh signs \(\varepsilon_z\) costs at most
\(\mathrm e(p_\star-1)^2\), since \((nN)^{2/p_\star}\le\mathrm e\).
Multiply by the squared randomization cost and use
\eqref{eq:endpoint-active-mask-bound}:
\begin{equation}
 \mathbb E_S[\max_{x,j}\|\ket{A^{\sigma,d,k}_{x,j}(S)}\|^2]
 \le\frac{2\mathrm e^5s^4(1+K_s)(p_\star-1)^2}{N}
          A_{\sigma,d,k}.
 \label{eq:endpoint-active-bound}
\end{equation}

\subsection{Zero source occupation}

For \(s=1\), the balanced part is zero. Assume \(s\ge2\).
When \(C_p=0\), the selected source phase cancels. We instead extract
an opposite-sign edge at \(p\). For a fixed \(\sigma\), define
\begin{equation}
 \ket{w_{p,q,q'}^\sigma}
 =a_{p,q',-\sigma}a_{p,q,\sigma}
       (N_{p,-\sigma})^{-1}\mathbf1[C_p=0]\ket{v_s},
 \label{eq:endpoint-balanced-extraction}
\end{equation}
where the inverse is zero on the zero eigenspace. The two sign banks
commute, so
\begin{equation}
 a_{p,q,\sigma}\mathbf1[C_p=0]\ket{v_s}
 =\sum_{q'}a_{p,q',-\sigma}^\dagger\ket{w_{p,q,q'}^\sigma}.
 \label{eq:endpoint-balanced-reconstruction}
\end{equation}
Indeed, the sum of the opposite-sign creation--annihilation products
is \(N_{p,-\sigma}\), which cancels its inverse. On its zero
eigenspace, \(C_p=0\) also makes the selected annihilator zero.
The inverse number operator accounts for every possible partner
without counting the source several times.

\paragraph{Pair mass and boundary cancellation.}
The normalized occupation basis gives
\begin{equation}
 \begin{aligned}
 \sum_{q,q'}\|\ket{w_{p,q,q'}^\sigma}\|^2
 &=\bra{v_s}\mathbf1[C_p=0]\,
       \frac{N_{p,\sigma}}{N_{p,-\sigma}}\,
       \mathbf1[C_p=0]\ket{v_s}\\
 &=\|\mathbf1[C_p=0,\ N_{p,-\sigma}>0]\ket{v_s}\|^2.
 \end{aligned}
 \label{eq:endpoint-balanced-source-mass}
\end{equation}
The ratio is zero when its denominator is zero. Each occupied
balanced source domain uses at least two edges. Thus
\[
 \sum_{p,q,q'}\|\ket{w_{p,q,q'}^\sigma}\|^2
 \le\frac{s}{2}\|\ket{v_s}\|^2.
\]
Rotate the two removed image coordinates, and then apply a Hadamard
transform to the common source domain:
\begin{equation}
 \begin{aligned}
 \ket{h_{p,y,y'}^\sigma(b)}
 &=\sum_{q,q'}(B_{5,\sigma}(b))_{y,q}(B_{5,-\sigma}(b))_{y',q'}
                         \ket{w_{p,q,q'}^\sigma},\\
 \ket{\widehat h_{d,y,y'}^\sigma(b)}
 &=\frac1{\sqrt N}\sum_p(-1)^{d\cdot p}\ket{h_{p,y,y'}^\sigma(b)}.
 \end{aligned}
 \label{eq:endpoint-balanced-transform}
\end{equation}
Both transformations are unitary on their marked-label direct sums.
Consequently,
\begin{equation}
 B_\sigma:=\sum_{d,y,y'}\|\ket{\widehat h_{d,y,y'}^\sigma(b)}\|^2
 =\sum_{p,q,q'}\|\ket{w_{p,q,q'}^\sigma}\|^2
 \le\frac{s}{2}\|\ket{v_s}\|^2.
 \label{eq:endpoint-balanced-mass}
\end{equation}

After rotating the reinserted partner in
\eqref{eq:endpoint-balanced-reconstruction}, the two removed domain
factors satisfy
\[
 (B_{1,\sigma}(b))_{x,p}(B_{1,-\sigma}(b))_{x',p}
 =\frac1N(-1)^{(x\oplus x')\cdot p}.
\]
The source phases cancel exactly. With \(d=x\oplus x'\), summing
over \(p\) gives \(N^{-1/2}\ket{\widehat h_{d,y,y'}^\sigma(b)}\).
The residual rotation is the same for every \(p\), so it commutes
with this transform of the marked index.

Put \(x'=x\oplus d\), \(e=(x,y,\sigma)\), and
\(e'=(x',y',-\sigma)\). Expanding the middle assignments of both
removed edges gives the balanced part
\begin{equation}
 \begin{aligned}
 \ket{B^\sigma_{x,j}(S)}
 =\frac1{\sqrt N}\sum_{d,y,y'}\sum_{z\ne z'}
 &\beta_{x,j}^\sigma(\vec{f})\xi_2(x\oplus d)^{-\sigma}
       \xi_4(y)^\sigma\xi_4(y')^{-\sigma}\\
 &{}\cdot H_{x\oplus y,z}H_{x\oplus d\oplus y',z'}
       \xi_3(z)^\sigma\xi_3(z')^{-\sigma}\\
 &{}\cdot\mathbb E_{\{z,z'\}}\left[
       \bigl(\mathcal T_{s-2}P_{e,e'}Q_b^{(s-2)}
                  \ket{\widehat h_{d,y,y'}^\sigma(b)}\bigr)(\vec{f})\right].
 \end{aligned}
 \label{eq:endpoint-balanced-expansion}
\end{equation}
This follows by substituting
\eqref{eq:endpoint-balanced-reconstruction} into
\eqref{eq:endpoint-rotated-annihilation} and
\eqref{eq:endpoint-bit-expansion}, and commuting creation through the
boundary rotation. Goodness makes both inserted edges fresh, so each
creation coefficient is one. Collision freedom requires \(z\ne z'\);
the average removes residual assignments using either location.
\paragraph{Three selected signs.}
Apply Lemma~\ref{lem:selected-sign} to the partner domain \(x'\)
in table \(2\), then successively to \(z,z'\) in table \(3\).
Use fresh families \(\eta_{x'},\varepsilon_z,\varepsilon'_{z'}\),
shared across all \((x,j)\) in the maximum. Goodness gives
\(x'\ne x\), and collision freedom gives \(z'\ne z\).

For every bit position \(j\), conditioning table-\(2\) words modulo
\(N/2\) writes each phase as a fair sign times its representative.
The factor \(\beta_{x,j}^\sigma(\vec f)\) has that same sign at
\(x\): a high-bit flip reverses it; a lower-bit flip changes only
the representative. Thus the full table-\(2\) sign polynomial has
degree \(s\), as do the middle-sign polynomials at both subsequent
steps. The total \(L^2\) cost is \((\mathrm e s)^3\). These
randomizations precede conditioning on the original tables. We keep
the fresh domain family indexed by \(x'\), not by \(d=x\oplus x'\).

\paragraph{The two image sums are a contraction.}
At fixed \(b\), put
\[
 \ket{c_{d,y,y'}(b)}=P^0_{d,y,y',\sigma}P_{\mathrm{good},s-2}
 Q_b^{(s-2)}\ket{\widehat h_{d,y,y'}^\sigma(b)},\qquad
 \ket{c(b)}=\bigoplus_{d,y,y'}\ket{c_{d,y,y'}(b)}.
\]
The preliminary projection includes \(d\ne0\), \(y\ne y'\), and
\(d\ne y\oplus y'\), as in Lemma~\ref{lem:good-extension-masks}.
Its remaining masks satisfy \(\sum_x(I-P_x)\preceq K_sI\) on the
marked direct sum. The vector is independent of \(x,\vec f\), and
\(\|\ket{c(b)}\|^2\le B_\sigma\). Define
\begin{equation}
 \begin{aligned}
 &(\mathcal A_{2,\sigma}\ket{w})_{d,z,z'}(\vec f)\\
 &\quad=\mathbf1[z\ne z']\,\mathbb E_{\{z,z'\}}\left[
 \sum_{y,y'}H_{y,z}H_{y',z'}\xi_4(y)^\sigma\xi_4(y')^{-\sigma}
              (\mathcal T_{s-2}\ket{w_{d,y,y'}})(\vec f)\right].
 \end{aligned}
 \label{eq:endpoint-two-image-map}
\end{equation}
This is the two-mark contraction from
\eqref{eq:marked-image-contraction}, following residual encoding and
retaining \(d\). It does not depend on \(x,j\), and the masks stay
before it.

The same Hadamard identity gives the randomized sum
\(\sum_{x',z,z'}\eta_{x'}\varepsilon_z\varepsilon'_{z'}
\ket{g_{x,j;x',z,z'}(S)}\), where
\[
 \begin{aligned}
 \ket{g_{x,j;x',z,z'}(S)}
 ={}&\frac1{\sqrt N}\beta_{x,j}^\sigma(\vec f)\xi_2(x')^{-\sigma}
       \xi_3(z)^\sigma\xi_3(z')^{-\sigma}(-1)^{x\cdot z+x'\cdot z'}\\
 &{}\cdot(\mathcal A_{2,\sigma}P_x\ket{c(b)})_{x\oplus x',z,z'}(\vec f).
 \end{aligned}
\]
Only now reindex \(x'\mapsto d=x\oplus x'\) in the coefficient
energy, separately at each fixed \(x\). This bijection gives
\begin{equation}
 \max_{x,j}\sum_{x',z,z'}\|\ket{g_{x,j;x',z,z'}(S)}\|^2
 \le\frac1N\max_x
       \|(\mathcal A_{2,\sigma}P_x\ket{c(b)})(\vec f)\|^2.
 \label{eq:endpoint-balanced-coefficients}
\end{equation}
Lemma~\ref{lem:sparse-mask-contraction}, first at fixed \(b\) and
then averaged over \(b\), gives
\begin{equation}
 \mathbb E_S\left[\max_{x,j}\sum_{x',z,z'}
             \|\ket{g_{x,j;x',z,z'}(S)}\|^2\right]
 \le\frac{2(1+K_s)}N B_\sigma.
 \label{eq:endpoint-balanced-mask-bound}
\end{equation}
Equation~\eqref{eq:chaos-maximum} for the three fresh families costs
\(\mathrm e(p_\star-1)^3\). Together with the squared
randomization cost, this proves
\begin{equation}
 \mathbb E_S[\max_{x,j}\|\ket{B^\sigma_{x,j}(S)}\|^2]
 \le\frac{2\mathrm e^7s^6(1+K_s)(p_\star-1)^3}{N}B_\sigma.
 \label{eq:endpoint-balanced-bound}
\end{equation}

\subsection{Summing the endpoint differences}

\begin{proof}[Proof of Proposition~\ref{prop:maximal-endpoint-bit-change}]
For a fixed size \(s\ge1\), the complete table-\(2\) difference is
the sum of the active and balanced parts above. There are at most
\[
 J_s=4s^2+2\le6s^2
\]
parts: two signs, at most \(s\) support sizes and \(2s\) nonzero
occupations, and one balanced part for each sign. Their masses satisfy
\[
 \sum_{\sigma,d,k\ne0}A_{\sigma,d,k}
       +\sum_\sigma B_\sigma\le2s\|\ket{v_s}\|^2.
\]
The coefficient in \eqref{eq:endpoint-balanced-bound} also bounds
the one in \eqref{eq:endpoint-active-bound}. Pointwise
Cauchy--Schwarz across the \(J_s\) parts therefore gives
\begin{equation}
 \mathbb E_S[\max_{x,j}\|\ket{q_{v_s}^{\Delta(2,x,j)}(S)}\|^2]
 \le\frac{24\mathrm e^7s^9(1+K_s)(p_\star-1)^3}{N}
              \|\ket{v_s}\|^2.
 \label{eq:endpoint-size-energy}
\end{equation}

Write \(\ket{v}=\sum_{s=0}^t\ket{v_s}\). The vacuum has zero
endpoint difference. Apply pointwise Cauchy--Schwarz across the at
most \(t+1\) size sectors, then use
\(\sum_s\|\ket{v_s}\|^2=\|\ket{v}\|^2\). Since \(K_s\le K_t\)
and \((t+1)t^9\le(t+1)^{10}\), this proves
\eqref{eq:endpoint-amplitude-energy} for table \(2\).

For table \(4\), exchange domain and image labels, the boundary
families \(\xi_1,\xi_5\), and the endpoint tables \(\xi_2,\xi_4\).
The positive boundary matrices \(H\operatorname{diag}(\xi)\) are
exchanged, as are their entrywise conjugates in the negative bank.
The goodness conditions and the middle factor \(H_{x\oplus y,z}\)
are invariant under this exchange. The annihilation identities,
mass bounds, and extension-mask counts therefore give the same
estimate. This proves the result for both endpoint tables.
\end{proof}

\section{Central-table and core-table differences}\label{app:derivatives}

We retain the middle-table estimate and then combine it with the two
endpoint estimates from Appendix~\ref{app:endpoint-differences}.
All expectations use the finite phase tables of the construction, and
all coefficient vectors may take values in an arbitrary Hilbert space.

Write \(b=(f_1,f_5)\) for the boundary sample,
\(a=(\xi_2,\xi_4)\) for the endpoint phases, and \(\xi=\xi_3\)
for the middle phases. As in \eqref{eq:record-feature-evaluator},
\(\ket{q_v(b,a,\xi)}\) denotes the coefficient after removing the
square-root sampling factor. The endpoint character \(\chi_{L,R}(a)\)
and collision-free middle polynomial \(\mu_{L,R}^{\rm cf}(\xi)\)
are those in \eqref{pubcore:packed-column}.
Let \(\eta_{z,r}\) be bit \(r\in[n]\) of the word \(h_3(z)\).
For \(u=(z,r)\), write \(\xi^{\oplus u}\) for the phase table
obtained by flipping this bit. For a scalar function \(F(b,a,\xi)\), and
analogously for a vector-valued function \(\ket{F(b,a,\xi)}\), define its bit difference by
\[
 F^{\Delta u}(b,a,\xi)
 :=\frac{F(b,a,\xi)-F(b,a,\xi^{\oplus u})}{2}.
\]
All other bits, including the other bits of \(h_3(z)\), stay fixed.

\begin{proposition}[Maximum change under a middle-bit flip]
\label{prop:maximal-middle-bit-change}
Let \(n,t\ge1\), put \(N=2^n\), and assume
\(3(2t)^{16}\le N\).
Let \(\ket{v}\) be independent of all phase tables and have record size
at most \(t\), with coefficients in any auxiliary Hilbert space.
Put \(\ket{P(b,a,\xi)}=\ket{q_v(b,a,\xi)}\), as in
\eqref{eq:record-feature-evaluator}.
Set \(m_{\rm mid}=nN\) and
\(p_{\rm mid}=2\lceil\ln(8m_{\rm mid})\rceil\).
Then
\begin{equation}
 \begin{gathered}
 \mathbb E_{b,a,\xi}[\max_u\|\ket{P^{\Delta u}(b,a,\xi)}\|^2]
 \le\mathfrak d_{\rm mid}^2\|\ket{v}\|^2,\\
 \mathfrak d_{\rm mid}^2:=\frac{4\mathrm e^5(t+1)t^6(p_{\rm mid}-1)^2}{N}.
 \end{gathered}
 \label{eq:amplitude-energy}
\end{equation}
Here \(\mathrm e\) is Euler's number. The maximum ranges over all
\(nN\) middle-table bits; the expectation averages all five phase
tables. The bound has no auxiliary-dimension factor.
\end{proposition}

A collision-free monomial affected at location \(z\)
has exactly one edge assigned to \(z\). We remove this edge's middle
factor and keep its full endpoint character. Averaging the remaining
polynomial over the whole word at \(z\) retains exactly the assignments
that avoid \(z\). This gives an exact identity for each single-bit
difference. We then randomize the selected edge's domain and image
signs. Pairing equal sign indices bounds their even moments, with a
logarithmic cost for the maximum over \(nN\) bits.

\subsection{An exact identity for a single-bit difference}

Fix the boundary sample \(b\) and write
\[
 \begin{aligned}
 \ket{c(b)}&=P_{\mathrm{good}}Q_b\ket{v},\\
 \ket{P_s(b,a,\xi)}&=\sum_{|L|+|R|=s}\ket{c_{L,R}(b)}\chi_{L,R}(a)
                         \mu_{L,R}^{\rm cf}(\xi),\\
 \ket{P(b,a,\xi)}&=\sum_{s=0}^t\ket{P_s(b,a,\xi)}.
 \end{aligned}
\]
The coefficients are independent of \(a,\xi\), are supported on
good records, and satisfy \(\|\ket{c(b)}\|\le\|\ket{v}\|\).
Let \(\ket{c_s(b)}\) be the component with record size \(s\), including
all positive and negative edge counts of that total size.

Fix \(1\le s\le t\). For a signed edge
\(e=(x_e,y_e,\sigma_e)\), where \(\sigma_e\in\{+1,-1\}\),
write \(e\in(L,R)\) when
\((x_e,y_e)\in L\) for \(\sigma_e=+1\), or
\((x_e,y_e)\in R\) for \(\sigma_e=-1\).
Let \((L,R)\setminus e\) delete this edge from its signed part. Set
\begin{equation}
 \ket{F_e(b,a,\xi)}=\sum_{\substack{|L|+|R|=s\\e\in(L,R)}}
 \ket{c_{L,R}(b)}\chi_{L,R}(a)\mu_{(L,R)\setminus e}^{\rm cf}(\xi).
 \label{eq:factor-removed-polynomial}
\end{equation}
Only the middle factor loses the selected edge. The endpoint
character \(\chi_{L,R}(a)\) still includes both endpoints of \(e\).
Each size-\(s\) record has only \(s\) choices of the edge to remove,
which gives the factor \(s\) in the following bound.

\begin{lemma}[Bound after removing one factor]\label{lem:marked-edge-energy}
For \(1\le s\le t\),
\begin{equation}
 \sum_e\mathbb E_{a,\xi}[\|\ket{F_e(b,a,\xi)}\|^2]\le s\|\ket{c_s(b)}\|^2.
 \label{eq:marked-energy}
\end{equation}
\end{lemma}
\begin{proof}
Deleting a fixed signed edge from a good record leaves a good record.
The remaining record together with that edge determines the original
record. Let \(\ket{c^{(e)}(b)}\) be the coefficient vector after this deletion.
Then \(\ket{F_e(b,a,\xi)}\) is the collision-free encoding of \(\ket{c^{(e)}(b)}\), multiplied
by the unit-modulus endpoint character of \(e\).
Contractivity in \eqref{pubcore:contractive-encoding} gives
\(\mathbb E_{a,\xi}[\|\ket{F_e(b,a,\xi)}\|^2]\le\|\ket{c^{(e)}(b)}\|^2\).
The sum over \(e\) counts each size-\(s\) coefficient exactly \(s\)
times, which proves the claim. The encoding acts as the identity on
the auxiliary Hilbert space.
\end{proof}

Each middle monomial of \(\ket{F_e(b,a,\xi)}\) uses \(s-1\) distinct locations,
with exponent \(+1\) or \(-1\) at each occupied location.
The average \(\mathbb E_z[\cdot]\) resamples all \(n\) bits of
\(h_3(z)\). Since \(N>8t\), either occupied character has zero
mean. Thus \(\mathbb E_z[\ket{F_e(b,a,\xi)}]\) retains exactly the assignments
whose remaining middle locations avoid \(z\).
Corollary~\ref{cor:whole-coordinate-mean}, applied to the Hilbert direct
sum \((\ket{F_e(b,a,\xi)})_e\), and Lemma~\ref{lem:marked-edge-energy} give
\begin{equation}
 \begin{aligned}
 &\mathbb E_{a,\xi}\left[\max_z\sum_e\|\mathbb E_z[\ket{F_e(b,a,\xi)}]\|^2\right]\\
 &\qquad\le2s\sum_e\mathbb E_{a,\xi}[\|\ket{F_e(b,a,\xi)}\|^2]
 \le2s^2\|\ket{c_s(b)}\|^2.
 \end{aligned}
 \label{eq:marked-coordinate-mean}
\end{equation}

Now fix a bit \(u=(z,r)\). In an affected collision-free monomial,
exactly one edge uses \(z\). Partition the affected assignments by
this edge, remove its middle factor, and use the preceding description
of \(\mathbb E_z[\ket{F_e(b,a,\xi)}]\). This proves the exact identity
\begin{equation}
 \begin{aligned}
 \ket{P_s^{\Delta u}(b,a,\xi)}&=\sum_e\alpha_{u,e}(\xi)\mathbb E_z[\ket{F_e(b,a,\xi)}],\\
 \alpha_{u,e}(\xi)
   &=H_{x_e\oplus y_e,z}
       \frac{\xi(z)^{\sigma_e}-(\xi^{\oplus u}(z))^{\sigma_e}}2,
 &|\alpha_{u,e}(\xi)|&\le N^{-1/2}.
 \end{aligned}
 \label{eq:middle-bit-expansion}
\end{equation}
The whole-word average appears only in the factor that avoids \(z\).
The coefficient \(\alpha_{u,e}(\xi)\) contains the actual single-bit
change. In particular, the identity and its bound hold for every
bit position \(r\), without treating a bit flip as a fresh phase sample.

\subsection{Proof of the maximum-change proposition}

\begin{proof}[Proof of Proposition~\ref{prop:maximal-middle-bit-change}]
It suffices to prove \eqref{eq:size-bit-energy} for each record size,
then sum over the sizes. Write \(m=m_{\rm mid}\) and
\(p=p_{\rm mid}\) within this proof.
Fix \(b\) and \(1\le s\le t\) as above.

\paragraph{Randomize the two selected endpoint signs.}
Condition on the endpoint phases modulo sign and on \(\xi\), as in
Appendix~\ref{app:core-tools}. Distinct endpoints make
\(\chi_{L,R}(a)\) homogeneous multilinear of degree \(s\) in each
endpoint-sign family. Since \(\ket{F_e(b,a,\xi)}\) retains this full
character, its terms contain the signs at both \(x_e,y_e\), for
either edge sign. Apply Lemma~\ref{lem:selected-sign} successively
to these two labels, with the maximum norm over \(u\). After
averaging the conditioned phases, the \(L^2\) cost is
\((\mathrm e s)^2\).

For fresh independent signs \(\delta_x,\gamma_y\), group the randomized sum by
unsigned edges:
\[
 \begin{aligned}
 \ket{a^{(u)}_{x,y}(b,a,\xi)}
 &=\sum_{e:\,(x_e,y_e)=(x,y)}
       \alpha_{u,e}(\xi)\mathbb E_z[\ket{F_e(b,a,\xi)}],\\
 \ket{G_u(b,a,\xi;\delta,\gamma)}
 &=\sum_{x,y}\delta_x\gamma_y\ket{a^{(u)}_{x,y}(b,a,\xi)}.
 \end{aligned}
\]
Each coefficient sums at most two signed edges. Cauchy--Schwarz and
\eqref{eq:middle-bit-expansion} therefore give
\begin{equation}
 \sum_{x,y}\|\ket{a^{(u)}_{x,y}(b,a,\xi)}\|^2
 \le\frac2N\sum_e\|\mathbb E_z[\ket{F_e(b,a,\xi)}]\|^2.
 \label{eq:marked-summary}
\end{equation}

\paragraph{Bound the maximum.}
Fix the original phases.
Apply \eqref{eq:chaos-maximum} to the two fresh sign families,
using \(m^{2/p}\le\mathrm e\), and then average the original phases.
Together with the squared randomization cost,
\eqref{eq:marked-summary}, and \eqref{eq:marked-coordinate-mean}, this gives
\begin{equation}
 \begin{aligned}
 &\mathbb E_{a,\xi}[\max_u\|\ket{P_s^{\Delta u}(b,a,\xi)}\|^2]\\
 &\quad\le(\mathrm e s)^4\mathbb E_{a,\xi,\delta,\gamma}
       [\max_u\|\ket{G_u(b,a,\xi;\delta,\gamma)}\|^2]\\
 &\quad\le(\mathrm e s)^4\frac{2\mathrm e(p-1)^2}{N}
       \mathbb E_{a,\xi}\left[
       \max_z\sum_e\|\mathbb E_z[\ket{F_e(b,a,\xi)}]\|^2\right]\\
 &\quad\le\frac{4\mathrm e^5s^6(p-1)^2}{N}\|\ket{c_s(b)}\|^2.
 \end{aligned}
 \label{eq:size-bit-energy}
\end{equation}
The single-coordinate bound counts only the \(s-1\) remaining middle
locations, and its Hilbert direct sum retains every auxiliary register.

\paragraph{Sum the sizes.}
The size-zero component has zero bit difference. Pointwise
Cauchy--Schwarz across the at most \(t+1\) sizes gives
\[
 \max_u\|\ket{P^{\Delta u}(b,a,\xi)}\|^2
 \le(t+1)\sum_{s=1}^t\max_u\|\ket{P_s^{\Delta u}(b,a,\xi)}\|^2.
\]
Use \eqref{eq:size-bit-energy},
\(\sum_s\|\ket{c_s(b)}\|^2=\|\ket{c(b)}\|^2\le\|\ket{v}\|^2\), and average \(b\).
This proves \eqref{eq:amplitude-energy}.

\end{proof}

\subsection{The three core tables}

The preceding proof averages a cube of \(nN\) middle bits.
The next proposition ranges over all \(3nN\) core bits and uses the
same sampling convention.

\begin{proposition}[Maximum change under a core-bit flip]
\label{prop:maximal-core-bit-change}
Let \(n,t\ge1\), \(N=2^n\), and assume
\(3(2t)^{16}\le N\). Let \(\ket{v}\) be independent of
all five tables, with record size at most \(t\) and coefficients in
any auxiliary Hilbert space. Put
\[
 M_v=\|\ket{v}\|^2,\qquad
 p_\star=2\lceil\ln(24nN)\rceil,\qquad K_t=16(t+1)^3,
\]
and set
\begin{equation}
 \mathfrak d_{\rm core}^2
 =\frac{64\mathrm e^7(t+1)^{10}(1+K_t)(p_\star-1)^3}{N}.
 \label{eq:core-difference-coefficient}
\end{equation}
For \(u=(j,x,i)\in\{2,3,4\}\times\mathbb F_2^n\times[n]\),
write \(S^{\oplus u}\) for the sample obtained by flipping bit \(i\)
of \(h_j(x)\), and use the half-difference
\[
 \ket{q_v^{\Delta u}(S)}
 =\frac{\ket{q_v(S)}-\ket{q_v(S^{\oplus u})}}2.
\]
Then
\begin{equation}
 \mathbb E_S\bigl[\max_u\|\ket{q_v^{\Delta u}(S)}\|^2\bigr]
 \le\mathfrak d_{\rm core}^2M_v,
 \label{eq:core-amplitude-energy}
\end{equation}
and, for \(W_v(S)=\|\ket{q_v(S)}\|^2\) and
\(W_v^{\Delta u}(S)=(W_v(S)-W_v(S^{\oplus u}))/2\),
\begin{equation}
 \mathbb E_S\bigl[\max_u|W_v^{\Delta u}(S)|\bigr]
 \le\phi_{\rm core}M_v,\qquad
 \phi_{\rm core}=2\mathfrak d_{\rm core}+2\mathfrak d_{\rm core}^2
 \le\poly(t,n)N^{-1/2}.
 \label{eq:core-maximal-flux}
\end{equation}
The expectations average all five tables, with the maxima inside.
The bounds include zero and subnormalized inputs and require no forest
support or bound on the auxiliary dimension.
\end{proposition}
\begin{proof}
The maximum over the three tables is at most the sum of their maxima.
Proposition~\ref{prop:maximal-endpoint-bit-change} bounds each endpoint
by
\[
 \frac{24\mathrm e^7(t+1)^{10}(1+K_t)(p_\star-1)^3}{N}M_v.
\]
For the middle table use \eqref{eq:amplitude-energy}. Since
\(p_{\rm mid}\le p_\star\) and \((t+1)t^6\le(t+1)^{10}\),
its coefficient is at most
\[
 \frac{4\mathrm e^7(t+1)^{10}(1+K_t)(p_\star-1)^3}{N}.
\]
The sum of these three bounds is dominated by
\eqref{eq:core-difference-coefficient}, since \(2\cdot24+4=52\le64\).
This proves \eqref{eq:core-amplitude-energy}.

For the weight, put \(A(S)=\max_u\|\ket{q_v^{\Delta u}(S)}\|\).
Direct expansion gives
\[
 W_v^{\Delta u}(S)
 =2\operatorname{Re}\bigl[
       (\ket{q_v(S)}-\ket{q_v^{\Delta u}(S)})^\dagger \ket{q_v^{\Delta u}(S)}\bigr].
\]
Thus Cauchy--Schwarz gives
\[
 \mathbb E_S\bigl[\max_u|W_v^{\Delta u}(S)|\bigr]
 \le2\bigl(\mathbb E_S[\|\ket{q_v(S)}\|^2]\bigr)^{1/2}
       \bigl(\mathbb E_S[A(S)^2]\bigr)^{1/2}
       +2\mathbb E_S[A(S)^2].
\]
Contractivity bounds the first expectation by \(M_v\), and
\eqref{eq:core-amplitude-energy} bounds the second by
\(\mathfrak d_{\rm core}^2M_v\). Their substitution proves
\eqref{eq:core-maximal-flux}. No division by \(M_v\) is used.
\end{proof}

\section{Parameters and error bounds}
\label{app:ledger}

We collect the construction bounds and the bounds for the correlated
public oracle. These bounds imply the security rate in Theorem~\ref{thm:main}.

\subsection{Construction parameters}
\label{app:construction-ledger}
\label{sec:instantiations}\label{app:phase-sampling}

Let \(t\ge1\) bound the number of queries to the unitary oracle.
We retain the conservative construction hypothesis
\(3(2t)^{16}\le N\).
It implies \(N>8t\) and \(t\le N/2\), as used for phase-character
orthogonality, the boundary estimates, and the recorder amplitudes.
The collision-free encoding is contractive by its normalized packing
factorization. Its loss bound follows from two-address collision
counting on the good-record space, and bounded insertion/deletion gives
the compression bound below.

\begin{center}\small
\begin{tabular}{@{}p{.34\textwidth}p{.58\textwidth}@{}}
\toprule Construction quantity & Bound or exact expression\\\midrule
Encoding loss \(\varepsilon_{\mathrm{cf},t}\) & \(2t(t-1)/N\)\\
Core-compression bound \(\eta_t\) & \(8t^2/\sqrt N\)\\
Encoded-vacuum discrepancy & \(0\)\\
Private-marginal allowance & \(\varepsilon_{\mathrm{cf},t}/2\)\\
Ordinary boundary defect \(\gamma_i\) & \(8i\sqrt{i/N}\)\\
Additional-goodness loss \(b_i\) & \(9i^4/N\), with \(b_0=0\)\\
Recorder covariance \(c_i\) & \(16\sqrt{2i(i+1)/N}\)\\
Recorder overlap loss \(\mathcal B_t\) & \(O(t^2N^{-1/4}+t^3N^{-1/2})\)\\
Prefix and marginal error \(e_t\) & \(O(t^2N^{-1/8})\)\\
Direct controlled Haar comparison & \(9t(t+1)N^{-1/8}\)\\
Independent-phase error \(\Delta_t\) & \(O(t^2N^{-1/8})\)\\
\bottomrule
\end{tabular}
\end{center}
Appendix~\ref{app:five-phase-transfer} proves
\(e_t,\Delta_t=O(t^2N^{-1/8})\).
The theorem retains \(3(2t)^{16}\le N\). For every fixed polynomial
\(t(\kappa)\) and \(n(\kappa)\ge\kappa\), this hypothesis holds
eventually and the errors are negligible.
PRF replacement adds negligible computational and synthesis errors.
Theorem~\ref{thm:controlled-haar} and the bounds used in its proof appear in
Appendix~\ref{app:controlled-recording}.

\begin{lemma}[Bounded-query phase moment matching]
\label{lem:phase-moment-matching}
Fix a bound \(t\ge1\) on queries to the unitary oracle.
Replace each uniform phase table by a \(2t\)-wise independent function
with uniform values in \(\mathbb Z_N\). Keep the five functions independent.
This replacement preserves the averaged output density operator.
The initial state and intervening operations must
be table-independent, and there is no separate access to the phase tables.
The same conclusion holds for replacing any subset of the layers,
conditional on all other data independent of those layers.
\end{lemma}
\begin{proof}
Every matrix entry of the averaged output state is a sum of products
of a ket query history and a bra query history. Each history uses each
phase layer at most once per unitary-oracle query, with either sign. Thus
each product depends on at most \(2t\) locations per layer, whose
joint distribution is unchanged. This includes repeated locations
and interference between control branches. The same argument applies
to a chosen subset after fixing the other layers and any independent
hybrid data.
\end{proof}

For a unitary-oracle query bound \(t\ge1\), put
\(r=\max\{n,\lceil\log_2(2t)\rceil,2\}\).
Construct \(\mathbb F_{2^r}\) deterministically in polynomial time
using \cite[Theorem~3.2]{Shoup90}. Inject the phase-table domain into
this field. For each layer, sample the \(2t\) coefficients
of a polynomial uniformly and independently, and retain \(n\) fixed
coordinates of its value. Vandermonde invertibility gives \(2t\)-wise
independent uniform phase words. The five coefficient lists use
\(10tr\) random bits and permit reversible polynomial-time evaluation.
In each PRF hybrid, these samplers efficiently implement the other
random layers. Condition on the remaining hybrid data and apply
Lemma~\ref{lem:phase-moment-matching} to those layers.

\paragraph{Coherent phase implementation.}\label{sec:phase-implementation}
Sample each phase table once and reuse it throughout the computation.
Evaluate each word into zero scratch, apply its known signed phase,
and uncompute, controlling only the phase rotation. Use five slots:
forward and conjugate visit layers \(1\) to \(5\).
Inverse and transpose visit them in reverse order.
Forward and transpose use positive phases. Inverse and conjugate use
negative phases. The idle branch uses zero phases. Apply \(H\) between
slots only on nonidle branches. Scratch is clean and all branches
implement one reused matrix with one consistent global phase.
Each unitary-oracle query uses ten word-oracle queries, or at most \(10n\) Boolean
queries when words are stored as bits. The moment lemma counts the
locations in the resulting phase circuit after the scratch is uncomputed.
The ten word queries per unitary-oracle query do not change its \(2t\)-wise
independence requirement.

Synthesize the controlled rotations with their specified relative
phases. Total operator error \(\eta\) bounds the vector and acceptance
errors, with cost polynomial in \(n,t,\log(1/\eta)\) over the fixed
gate set. The field construction and reversible evaluation are exact.
Only phase synthesis contributes implementation error. Taking the
synthesis accuracy to be \(2^{-\kappa^2}\) per unitary-oracle query gives total
error at most \(\poly(\kappa)2^{-\kappa^2}\) for a polynomial-query
experiment. The analytical reduction in Lemma~\ref{lem:exact-second-moment}
uses exact rotations and needs no synthesis approximation.

\subsection{Additional bounds for the correlated public oracle}
\label{app:correlated-ledger}

Fix a polynomial \(\ell_{\rm in}\) before sampling the oracle.
Its value at \(\kappa\) bounds the lengths of the explicit circuit
descriptions for all addressed verifier inputs. It therefore bounds
their total number of base queries and each query length.
Let \(M_{\rm qry}\) count the relevant oracle coordinates as in
\eqref{eq:oracle-coordinate-bound}. The sensitivity scale \(\Gamma_{\rm core}\)
and cutoff \(u_0\) depend on this bound on description length. Thus
\(\lambda_t\) is shorthand for a quantity depending on
\(\kappa,t,\ell_{\rm in}(\kappa)\), not on \(t\) alone. The number
and lengths of all public addresses also enter the simultaneous scalar event.

The correlated-oracle proof uses the core-bit difference estimate of
Proposition~\ref{prop:maximal-core-bit-change}, proved using
Appendices~\ref{app:core-tools}--\ref{app:derivatives}.
The three core tables have \(m_{\rm core}=3nN\) bits. The maximal
estimate averages all five tables and retains arbitrary auxiliary
spaces. The phase order is \(N=2^n>8t\), as implied by the
construction condition \(3(2t)^{16}\le N\).
Here \(d_A(R,Q)\) is the acceptance probability with a uniformly keyed
challenge from \(R\), minus the acceptance probability with a challenge
defined by an independent uniform row, as in Section~\ref{sec:erasure}.
The same rate bounds hold for the Haar gap
\(d_A^{\mathrm{Haar}}(R,Q)\) by Remark~\ref{rem:haar-second-moment}; its
absolute value is \(\operatorname{Adv}_A^{R,Q}(\kappa)\).

First, the exceptional-set argument needs a pointwise upper bound on
\(W_v(S)\). There are at most \((t+1)(2N^2)^t\) records of length at
most \(t\). A size-\(s\) feature in \eqref{pubcore:feature} has at most
\((N)_s\le N^s\) assignments, each of modulus \(N^{-s/2}\). Thus
\(|\mu_{L,R}^{\mathrm{cf}}(\vec f)|\le N^{s/2}\) for every core sample
\(\vec f\). For each boundary sample,
\(\|P_{\mathrm{good}}Q_bv\|^2\le\|v\|^2\). Hilbert-valued Cauchy--Schwarz applied
to the coefficient sum, with a loose factor of two, gives
\begin{equation}
 0\le W_v(S)\le2(t+1)(2N^3)^t\|v\|^2.
 \label{eq:pointwise-weight}
\end{equation}
The logarithm of the coefficient on the right is
\(O(tn+\log(t+1))\), as required for the exceptional-set estimate.
The bound permits \(W_v(S)>\|v\|^2\).
Contraction holds after averaging over \(S\).

Set \(B_t=2(t+1)(2N^3)^t\), the coefficient in
\eqref{eq:pointwise-weight}. The coefficients of
\eqref{eq:core-amplitude-energy}--\eqref{eq:core-maximal-flux} are
\[
 \begin{gathered}
 p_\star=2\lceil\ln(24nN)\rceil,\qquad K_t=16(t+1)^3,\\
 \mathfrak d_{\rm core}^2
 =\frac{64\mathrm e^7(t+1)^{10}(1+K_t)(p_\star-1)^3}{N},\\
 \phi_{\rm core}=2\mathfrak d_{\rm core}+2\mathfrak d_{\rm core}^2.
 \end{gathered}
\]
Here \(\mathrm e\) is Euler's number. Explicitly,
\[
 \begin{aligned}
 \phi_{\rm core}
 &={}\frac{16\mathrm e^{7/2}(t+1)^5\sqrt{1+K_t}(p_\star-1)^{3/2}}{\sqrt N}\\
 &\quad+\frac{128\mathrm e^7(t+1)^{10}(1+K_t)(p_\star-1)^3}{N}.
 \end{aligned}
\]
Both terms are retained for finite parameters. Since
\(1+K_t\le17(t+1)^3\) and \(p_\star-1=O(n+1)\), this gives
\[
 \phi_{\rm core}
 =O\!\left(\frac{(t+1)^{13/2}(n+1)^{3/2}}{\sqrt N}
          +\frac{(t+1)^{13}(n+1)^3}{N}\right).
\]
The score clipping has
Lipschitz constant \(3\). The core-bit sensitivity cutoff in
\eqref{eq:sensitivity-cutoff} is
\[
 u_0=\left\lceil\Gamma_{\rm core}
 \left(\kappa^2+2n+\log_2(M_{\rm qry}m_{\rm core}^2B_t)\right)\right\rceil.
\]
For fixed polynomial resources and
\(\kappa\le n(\kappa)\le\poly(\kappa)\), both \(u_0\) and
\(\log B_t\) are polynomially bounded. The construction condition
holds for all sufficiently large \(\kappa\); the cutoff may depend
on the resource descriptions.

\begin{center}\small
\begin{tabular}{@{}p{.34\textwidth}p{.58\textwidth}@{}}
\toprule Additional quantity & Bound\\\midrule
Core-bit count & \(m_{\rm core}=3nN\)\\
Core-bit amplitude difference &
\(\mathfrak d_{\rm core}^2=O((t+1)^{13}(n+1)^3/N)\)\\
Core-bit weight difference &
\(\phi_{\rm core}=2\mathfrak d_{\rm core}+2\mathfrak d_{\rm core}^2\)\\
Pointwise weight coefficient & \(B_t=2(t+1)(2N^3)^t\)\\
Scalar stability \(\alpha\) & \(6\cdot2^{-\lfloor\kappa/16\rfloor}\)\\
Scalar exceptional probability \(\varepsilon_{\rm stab}\) & \(\le2^{-\kappa^2}\)\\
Weighted coefficient \(\lambda_t\) &
\(3\pi^2\ell_{\rm in}(\kappa)^2u_0\phi_{\rm core}+2^{-\kappa^2}N^{-2}\)\\
Weighted asymptotic rate & \(\lambda_t\le\poly(\kappa)N^{-1/2}\)\\
Boolean query-state change & \(2\sqrt{\alpha+\lambda_t}+2e_t\) per
normalized public query\\
Uniform asymptotic query cost & \(\poly(\kappa)2^{-\kappa/32}\)\\
Signed mean and second moment &
\(\begin{gathered}
 |\mathbb E_{R,Q}[d_A(R,Q)]|,\ \mathbb E_{R,Q}[d_A(R,Q)^2]\\
 \le P_A(\kappa)2^{-\kappa/32}
 \end{gathered}\)\\
Expected absolute advantage &
\(\mathbb E_{R,Q}[\operatorname{Adv}_A^{R,Q}(\kappa)]
 \le\sqrt{P_A(\kappa)}\,2^{-\kappa/64}\)\\
Joint tail at \(2^{-\kappa/128}\) &
\(\Pr_{R,Q}[|d_A(R,Q)|>2^{-\kappa/128}]
 \le P_A(\kappa)2^{-\kappa/64}\)\\
Eventual fixed-oracle advantage & \(O(2^{-\kappa/128})\) almost surely\\
\bottomrule
\end{tabular}
\end{center}
Here \(P_A\) is a polynomial with \(P_A(\kappa)\ge1\), depending on the fixed adversary,
dimension rule, finite verifier family, and enlarged two-run resources.
There is no common polynomial or eventual cutoff for all adversaries.
The cutoff may also depend on the sampled joint oracle.

The unchanged scalar bound gives \(\alpha\le12K^{-1/16}\).
Since \(N\ge K=2^\kappa\), the core-bit bound and
\eqref{eq:rounded-query-loss} give
\[
 \begin{aligned}
 2\sqrt{\alpha+\lambda_t}+2e_t
 &\le\poly(\kappa)\bigl(K^{-1/32}+N^{-1/4}+N^{-1/8}\bigr)\\
 &\le\poly(\kappa)2^{-\kappa/32}.
 \end{aligned}
\]
The term \(\sqrt\alpha\) determines this key-length rate.
The construction error \(e_t=O(t^2N^{-1/8})\) decays faster.
Lemma~\ref{lem:finite-query-bound} sums over polynomially many public
queries and gives the same rate for the signed mean.
The separate comparison with Haar adds
\(\Delta_t=O(t^2N^{-1/8})\), which also decays faster.

Use the full encoded address set and resource bounds of
Lemma~\ref{lem:exact-second-moment} when applying scalar and weighted
stability. Re-evaluate their cutoffs, including \(\ell_{\rm in}\),
for the \(2q+10nt\) public queries and enlarged verifier descriptions.
This includes queries to \(R\) and \(B_R\), and all selected
table tags.
Lemma~\ref{lem:scalar-row-stability} bounds the scalar exceptional
probability by \(2^{-\kappa^2}\). The query comparison charges it once.
This bound is negligible at the final \(\kappa\)-scale for every
allowed polynomial dimension rule.
The \(2n\) and \(\log B_t\) terms in the core-bit cutoff ensure
that its exceptional contribution is at most \(2^{-\kappa^2}N^{-2}\).

Lemma~\ref{lem:exact-second-moment} gives the second-moment rate using
an average over an independent private row and one shared Boolean table.
The exponent remains \(1/32\); Cauchy--Schwarz then gives the
expected absolute rate \(1/64\).
Equation~\eqref{eq:summable-tail} gives the summable tail.
Section~\ref{sec:fixed-oracle-conclusion} uses this bound to fix one joint oracle
for all uniform adversaries and proves Lemma~\ref{lem:row-decoupling}.
The sharper eventual rates are in Remark~\ref{rem:sharper-fixed-rates}.

\section{\texorpdfstring{Extension to \(\mathsf{QMA}(2)\) promise-oracle security}{Extension to QMA(2) promise-oracle security}}
\label{app:qma2}

This appendix applies the witness-set generality established in
Section~\ref{sec:witnesses}. We specify the product witnesses and the
joint completion law, then apply the scalar, weighted, and fixed-oracle
bounds with table tags included in the query addresses.

\paragraph{Product witnesses.}\label{app:qma2-witnesses}
For \(\mathcal V_j^{(2)}\) on classical input \(z\), fix polynomial-size
registers \(A_z,B_z\), with size rules and bipartition independent of the
oracle \(O\).
For each \(O\), define
\begin{equation}
\begin{aligned}
 \mathcal W_z^{(2)}
   &=\{\rho_A\otimes\rho_B:
       \rho_A\in\mathcal D(\mathcal H_{A,z}),\quad
       \rho_B\in\mathcal D(\mathcal H_{B,z})\},\\
 g_{j,z}^{(2)}(O)
   &=\max_{\rho\in\mathcal W_z^{(2)}}
       \Pr[(\mathcal V_j^{(2)})^O(z,\rho)\text{ accepts}],\\
 r_{j,z}^{(2)}(O)
   &=\min\{1,\max\{0,3g_{j,z}^{(2)}(O)-1\}\}.
\end{aligned}\label{eq:qma2-score}
\end{equation}
The allowed set is compact as a continuous image of two compact sets.
Each block may be internally entangled and the verifier may act jointly
on them; the optimization remains over product states.

For the certificate bound, fix a finite verifier family and polynomially
bounded input lengths. Expand each computation into an explicit circuit
description, retaining its table tag, input, and witness bipartition.
Fix a polynomial \(\ell_{\rm in}\) that bounds these description
lengths at parameter \(\kappa\). This encoding preserves the scores
for every oracle. For each such pair \((j,z)\) and oracle \(O\),
Corollary~\ref{cor:clipped-certificate} gives a finitely supported
certificate \(w_O\) for \(r_{j,z}^{(2)}\), with
\begin{equation}
 \sum_iw_O(i)\le\mathfrak c:=3\pi^2\ell_{\rm in}(\kappa)^2.
 \label{eq:qma2-clipped-certificate}
\end{equation}
No convexity of \(\mathcal W_z^{(2)}\) is needed. The maximizing
product witness at one oracle remains an allowed witness at every
other oracle. The certificate-composition and downward-variation
arguments therefore apply with exactly the same mass bound.

\paragraph{The mixed completion.}\label{app:qma2-completions}
Tag each verifier by its canonical circuit description, type,
running-time bound, witness-size rules, and bipartition, using the standard
\(1/3,2/3\) promise endpoints. Independently of \(R\) and the original coins,
sample independent coins at distinct tagged pairs \((j,z)\):
\begin{equation}
 \Theta_{j,z}\sim\operatorname{Unif}[0,1),\qquad
 Q_j^{(2)}(z)=\mathbf1[\Theta_{j,z}<r_{j,z}^{(2)}(O_R)].
 \label{eq:qma2-threshold}
\end{equation}
With \(R\) uniform, the joint law is
\begin{equation}
 Q=(Q^{(1)},Q^{(2)})\mid R\sim\nu_R^{\rm mix}
 :=\nu_R\otimes
   \bigotimes_{(j,z)\,\text{of type }\mathsf{QMA}(2)}
     \operatorname{Bernoulli}(r_{j,z}^{(2)}(O_R)).
 \label{eq:qma2-joint-law}
\end{equation}
The product includes all such descriptions and finite inputs, preserving
independence even when scores coincide. Repeated queries, both runs of
the second-moment argument, and different security parameters reuse the
same bits. The \(Q^{(1)}\) marginal is unchanged.

\begin{corollary}[Mixed \(\mathsf{QMA}\) and \(\mathsf{QMA}(2)\) promise-oracle security]
\label{app:qma2-corollary}\label{cor:qma2-security}
Use Theorem~\ref{thm:main}'s construction with polynomial-time computable
\(\kappa\le n(\kappa)\le\poly(\kappa)\).
Almost every infinite \((R,Q)\) under \eqref{eq:qma2-joint-law} works
simultaneously for every uniform polynomial-time adversary and every
fixed finite family of the specified product-witness verifiers, each
with its own fixed polynomial resources.
For \(\mathcal L_{R,Q}=(R,Q^{(1)},Q_1^{(2)},\ldots,Q_r^{(2)})\), put
\begin{equation}
\begin{aligned}
 d_A^{\mathrm{Haar}}(R,Q)={}&\Pr_{k\leftarrow\{0,1\}^\kappa}
 [A^{\mathcal L_{R,Q};U_k}(1^\kappa)=1]\\
 &-\Pr_{U\leftarrow\Haar(U(N))}
 [A^{\mathcal L_{R,Q};U}(1^\kappa)=1].
\end{aligned}\label{eq:qma2-gap}
\end{equation}
Then \(|d_A^{\mathrm{Haar}}(R,Q)|=O(2^{-\kappa/128})\) eventually, with an
oracle-, adversary-, family-, and resource-dependent cutoff.
Both experiments retain the same public oracle and allow coherent controlled
access to one reused matrix with consistent global phase and arbitrary
work-register entanglement. No totality assumption is imposed.
Equivalently, for almost every \(R\), almost every conditional completion
works. Section~\ref{sec:other-verifiers} includes ordinary \(\mathsf{QMA}\) verifier tables over
\(O_R\), with their specified completion law, and direct \(B_R\) queries.
\end{corollary}
\begin{proof}
The product-witness sets are nonempty, compact, and independent of the
oracle. Their register sizes and fixed bipartitions satisfy the polynomial
bounds in Section~\ref{sec:witnesses}. Equation~\eqref{eq:qma2-clipped-certificate}
therefore supplies the certificate used by both stability lemmas.
The completion law supplies independent persistent bits conditional on
\(R\). These are all the additional hypotheses needed by the comparison
proof. Include table tags in every address and use common resource
bounds for the selected family and the two-run algorithm.
Choose \(\ell_{\rm in}\) to cover all their explicit circuit
descriptions. Lemma~\ref{lem:scalar-row-stability} then applies
with this cutoff.
Lemma~\ref{lem:public-query-normal-form} encodes base queries through
\(Q^{(1)}\); product-witness scores remain in separate tables.
The gap \(d_A^{\mathrm{Haar}}(R,Q)\) in \eqref{eq:qma2-gap} compares the keyed challenge with
Haar. Lemma~\ref{lem:finite-query-bound} bounds the change to an
independent row, and Theorem~\ref{thm:independent-phase} adds
\(\Delta_t\) for the comparison with Haar.
Use the Haar version of the two-run argument in
Remark~\ref{rem:haar-second-moment}: average over a private Haar matrix and
apply this combined bound with the enlarged resources of
Lemma~\ref{lem:exact-second-moment}.
Together with \eqref{eq:coarse-rates}, these bounds give
\begin{align}
 |\E_{R,Q}[d_A^{\mathrm{Haar}}(R,Q)]|&\le\epsilon_A^{\rm round}+\Delta_t,
       \label{eq:qma2-finite-mean}\\
 |\E_{R,Q}[d_A^{\mathrm{Haar}}(R,Q)]|,\quad\E_{R,Q}[(d_A^{\mathrm{Haar}}(R,Q))^2]
 &\le\poly(\kappa)2^{-\kappa/32},\label{eq:qma2-rates}\\
 \Pr_{R,Q}[|d_A^{\mathrm{Haar}}(R,Q)|>2^{-\kappa/128}]
 &\le\poly(\kappa)2^{-\kappa/64}.\label{eq:qma2-tail}
\end{align}
The fixed-oracle argument in the proof of Lemma~\ref{lem:row-decoupling}
in Section~\ref{sec:fixed-oracle-conclusion} applies because the
verifier descriptions, finite families, uniform adversaries, and resource
bounds are countable. Equation~\eqref{eq:conditional-fixing} gives the
conditional form.
Remark~\ref{rem:sharper-fixed-rates} gives the sharper eventual rates
with these same tagged resource bounds.
\end{proof}

Only Boolean XOR access is supplied, including superpositions and
off-promise inputs. Score verifiers receive classical inputs and their
prescribed witnesses, but no sampled tables, external challenge, or
adversary-supplied quantum input. No oracle-dependent advice is supplied.
The completion need not be a total \(\mathsf{QMA}(2)\) language. Applying target \(X\)
after a query complements that same table, without asserting complement
closure for \(\mathsf{QMA}(2)\). Appendix~\ref{sec:oracle-search}'s consequences concern
verifiers over \(O_R\); they imply no complexity collapse or absence of
OWFs relative to the full mixed oracle.

\section{Controlled path recording and error propagation}
\label{app:controlled-recording}

We extend the four-interface recording arguments of~\cite[Appendix~B]{SMLBH}
to controlled oracles and additional record restrictions. We combine these
arguments with the restricted-recording approach of
\cite[Section~4.2, Lemma~4.3]{ABGL} to analyze the five-phase
construction in Section~\ref{sec:pub-norm-loss-details}.

The notation is that of~\eqref{eq:controlled-recorders}, with boundary
factors \(G_a,D_a\) at the input and output, respectively, from
Table~\ref{tab:controlled-boundaries} and private
rotation \(Q_b=Q[G,D]\) from~\eqref{eq:common-boundary-rotation}.

Following the register notation of~\cite[Appendix~B]{SMLBH},
\(\mathsf A\) is the query register and \(\mathsf B\) contains the remaining
accessible registers: the control register, workspace, and any
reference. The projectors \(P_a\) act on the control register inside \(\mathsf B\).
The private registers \(\mathsf L,\mathsf R\) hold the normalized relation
states. In the CPFC construction, \(\mathsf P,\mathsf F\) hold the
permutation and phase-function samples, and \(\mathsf G,\mathsf D\) hold
the input and output boundary samples. Register names use sans-serif letters, distinct
from the operators used in the proofs below. For a vector,
\(\|\ket{\psi}\|\) is its Hilbert-space norm; for an operator, \(\|M\|\) is
the operator norm. The subscript in \(\|M\|_1\) denotes the trace norm.

\subsection{Restriction identities and coherent controls}
\label{app:recording-restrictions-proof}

Lemma~\ref{pubcore:recording-restrictions} gives, in the branch notation
of~\eqref{eq:controlled-recorders},
\begin{equation}
 W_a^\dagger V_a=\Pi_{\mathcal D(W_a)}
       \le I\otimes\Pi_{\mathrm{bij}}\quad(1\le a\le4).
 \label{eq:controlled-restrictions}
\end{equation}
Each creation or deletion changes size by one, and each projector
preserves size. One oracle query increases size by at most one.

\subsection{Error propagation through contractions}

The next lemma bounds the full-state error by the compression error and
total norm loss. We use it for the five-phase construction and the CPFC comparison.
\begin{lemma}[Error propagation through contractions]
\label{lem:approximate-compression-transfer}
\label{lem:pub-norm-loss}
Let \(T:\mathcal K\to\mathcal H\) and \(M_1,\ldots,M_t\) be contractions,
let \(U_1,\ldots,U_t\) be unitaries, and suppose
\(\|T^\dagger U_iT-M_i\|\le\eta\). From a unit vector \(\ket{K_0}\) and \(\ket{\Phi_0}\), define
\(\ket{K_i}=M_i\ket{K_{i-1}}\) and \(\ket{\Phi_i}=U_i\ket{\Phi_{i-1}}\). If
\(\|\ket{\Phi_0}-T\ket{K_0}\|\le d_0\), then
\begin{equation}
 \|\ket{\Phi_t}-T\ket{K_t}\|
       \le d_0+\sqrt{t(1-\|\ket{K_t}\|^2)+2t^2\eta}.
 \label{eq:approximate-compression-transfer}
\end{equation}
The spaces may include arbitrary Hilbert-space references, with all
maps extended by identity on untouched registers. 
\end{lemma}
\begin{proof}
Set \(B_i=T^\dagger U_iT\). For
\(\ket{\psi}\in\mathcal K\), contractivity gives
\begin{align*}
 \|U_iT\ket{\psi}-TM_i\ket{\psi}\|^2
 &=\|T\ket{\psi}\|^2+\|TM_i\ket{\psi}\|^2
       -2\operatorname{Re}\bra{\psi}M_i^\dagger B_i\ket{\psi}\\
 &\le\|\ket{\psi}\|^2-\|M_i\ket{\psi}\|^2
       -2\operatorname{Re}\bra{\psi}M_i^\dagger(B_i-M_i)\ket{\psi}\\
 &\le\|\ket{\psi}\|^2-\|M_i\ket{\psi}\|^2+2\eta\|\ket{\psi}\|^2.
\end{align*}
Write \(d_i=\|\ket{K_{i-1}}\|^2-\|\ket{K_i}\|^2\ge0\). Notice that
\begin{align*}
 \ket{\Phi_t}-T\ket{K_t}
 ={}&U_t\cdots U_1(\ket{\Phi_0}-T\ket{K_0})\\
 &+\sum_{i=1}^tU_t\cdots U_{i+1}(U_iT-TM_i)\ket{K_{i-1}}.
\end{align*}
Every \(\ket{K_i}\) has norm at most one. The triangle inequality
and Cauchy--Schwarz therefore bound its norm by
\[
 d_0+\sum_{i=1}^t\sqrt{d_i+2\eta}
 \le d_0+\sqrt{t\sum_{i=1}^td_i+2t^2\eta}.
\]
Since \(\sum_i d_i=1-\|\ket{K_t}\|^2\), this proves the bound.
\end{proof}

For subnormalized vectors \(\ket{x},\ket{y}\), we will also use
\begin{equation}
 \frac12\|\ket{x}\bra{x}-\ket{y}\bra{y}\|_1
       \le\frac{\|\ket{x}\|+\|\ket{y}\|}{2}\|\ket{x}-\ket{y}\|.
 \label{eq:controlled-subnormalized-distance}
\end{equation}
This is because we can write the difference as
\((\ket{x}-\ket{y})\bra{x}+\ket{y}(\bra{x}-\bra{y})\) and use the trace norm of a
rank-one operator.

\subsection{An exact purification of restricted recording}

We prove the following comparison with one common recorder environment
for all five branches.

\begin{theorem}[Direct controlled Haar comparison]
\label{thm:controlled-haar}
Let \(0\le t<N\) and consider an algorithm making \(t\) unitary-oracle queries.
Let \(\ket{\mathcal A_t^V}_{\mathsf{ABLR}}:=\ket{v_t}\) be its global
state after \(t\) queries to the controlled path-recording oracle
\(\widehat V\), starting with empty relations as
in~\eqref{eq:selector-recorder-trajectories}. 
Then
\begin{equation}
 \frac12\left\|
   \E_{U\leftarrow\Haar}\!\Bigl[\rho_t^U\Bigr]
       -\Tr_{LR}\ket{v_t}\bra{v_t}\right\|_1
       \le9t(t+1)N^{-1/8}.
 \label{eq:controlled-haar-bound}
\end{equation}

\end{theorem}

We use the four-interface recording lemmas of
\cite[Appendix~B]{SMLBH}. 

Let \(\pi\) be a uniform permutation, \(f:[N]\to\mathbb Z_3\) an independent
uniform function, and \(\omega=e^{2\pi i/3}\). Set
\(O_{\pi,f}=P_\pi F_f\), where \(F_f\ket{x}=\omega^{f(x)}\ket{x}\).
The purified permutation--function oracle acts on \(\mathsf{APF}\):
\[
 \mathsf{pfO}
 =\sum_{\pi,f}(P_\pi F_f)_{\mathsf A}
                  \otimes\ket{\pi,f}\bra{\pi,f}_{\mathsf{PF}},
 \qquad
 \mathsf{pfO}\ket{x,\pi,f}
       =\omega^{f(x)}\ket{\pi(x),\pi,f}.
\]
The sample registers are unchanged by every query. In the remaining
identities, \(O_a\) denotes the corresponding purified operation:
\[
 (O_0,O_1,O_2,O_3,O_4)
 =(I,\mathsf{pfO},\mathsf{pfO}^\dagger,\mathsf{pfO}^T,
       \overline{\mathsf{pfO}}).
\]
The transpose and conjugate are taken in the computational bases. Thus each sample block implements
\(I,O_{\pi,f},O_{\pi,f}^\dagger,O_{\pi,f}^T,
\overline{O_{\pi,f}}\), respectively, with one consistent phase.
For a bijective pair \((L,R)\) of total size \(s\), use the purification
of~\cite[Definition~20]{SMLBH}:
\begin{equation}
 \ket{\mathrm{pf}_{L,R}}=
 \frac1{\sqrt{3^N(N-s)!}}
 \sum_{\substack{\pi\supseteq L\cup R\\f:[N]\to\mathbb Z_3}}
 \omega^{\sum_{(x,y)\in L}f(x)-\sum_{(x,y)\in R}f(x)}\ket{\pi,f}.
 \label{eq:controlled-pf-relations}
\end{equation}
Lemma~11 of that paper gives the common isometry
\begin{equation}
 J_{PF}\ket{L,R}=\ket{\mathrm{pf}_{L,R}},\qquad
 J_{PF}^\dagger J_{PF}=I_{\mathrm{bij}}.
 \label{eq:controlled-common-compressor}
\end{equation}
Thus \(J_{PF}=(\mathsf{Compress}^{\mathsf{PF}})^\dagger\) in its notation:
it maps the bijective subspace of \(\mathsf{LR}\) into the sample
registers \(\mathsf{PF}\). Moreover,
\(J_{PF}\ket{\varnothing,\varnothing}
=\ket{\Omega_{PF}}=(N!3^N)^{-1/2}\sum_{\pi,f}\ket{\pi,f}\).
This one sample is reused throughout. The CPFC permutation--phase core
uses ternary phases; the five-phase circuit keeps \(q=N\).

The fresh-input identities and domain decompositions in
\cite[Lemmas~12 and~19--20]{SMLBH} give
\begin{equation}
 O_aJ_{PF}\Pi_{\mathcal D(W_a)}=J_{PF}W_a\quad(1\le a\le4),\qquad
 \Pi_{\mathcal D(W_a)}=W_a^\dagger W_a.
 \label{eq:controlled-pf-restrictions}
\end{equation}
The idle identity is \(O_0J_{PF}=J_{PF}W_0=J_{PF}\) on the bijective
record space.
Put \((C,D)=(W^L,W^R)\) or
\((\overline W^L,\overline W^R)\), and \(O=O_1\) or \(O_4\), respectively.
The cited fresh identities give
\[
 \begin{aligned}
  OJ_{PF}C^\dagger C&=J_{PF}C,&
  O^\dagger J_{PF}D^\dagger D&=J_{PF}D,\\
  OJ_{PF}(C^\dagger C+DD^\dagger)&=J_{PF}(C+D^\dagger).
 \end{aligned}
\]
Thus, for every bijective input,
\begin{equation}
 \|O_aJ_{PF}\ket{\psi}-J_{PF}W_a\ket{\psi}\|^2
 =\|(I-\Pi_{\mathcal D(W_a)})\ket{\psi}\|^2
 =\|\ket{\psi}\|^2-\|W_a\ket{\psi}\|^2.
 \label{eq:controlled-pf-one-step}
\end{equation}

\subsection{Effect of Boundary Scramblers}

We extend the four-branch estimates of
\cite[Lemmas~16 and~24]{SMLBH} for coherent access. Since these are operator estimates on
the same record space, their controlled versions follow by
orthogonal direct sums.

\begin{lemma}[Common controlled covariance]
\label{lem:controlled-covariance}
For arbitrary unitaries \(G,D\) and \(0\le j<N\), set
\(\widehat V^{G,D}=\sum_aP_a\otimes D_aV_aG_a\). Then
\begin{equation}
 \bigl\|[\widehat V^{G,D}Q[G,D]-Q[G,D]\widehat V]
                   \Pi_{\le j}\bigr\|
 \le c_j,\qquad c_j=16\sqrt{2j(j+1)/N}.
 \label{eq:controlled-covariance}
\end{equation}
\end{lemma}
\begin{proof}
Apply~\cite[Lemma~16, equations~(B.43)--(B.46)]{SMLBH} to the four
nonidle blocks, with input cutoff \(\Pi_{\le j}\). The norm of their
orthogonal sum is their maximum.
\end{proof}

Fix an
\((\ell,r)\) sector with \(s=\ell+r<N\), let
\(E^{=}=\sum_x\ket{x,x}\bra{x,x}\),
\(E^\Phi=\ket{\Phi}\bra{\Phi}\), and
\(\ket{\Phi}=N^{-1/2}\sum_x\ket{x,x}\).
In the separately symmetric tensor embedding, \(X_{L_j},Y_{L_j}\)
are the input and output registers of the \(j\)-th left edge;
\(X_{R_k},Y_{R_k}\) denote those of the \(k\)-th right edge.
Writing \(A\) for the query register,
\cite[Corollary~11.1 and Claim~26, equations~(11.32)--(11.35)]{MH} gives
\begin{align}
 \Delta^{\mathcal D}_{\ell,r}
 &\preceq\frac{N}{N-s+1}\left[
  \sum_{j=1}^{\ell}E^{=}_{A,X_{L_j}}
  +\sum_{k=1}^{r}(E^{=}_{A,X_{R_k}}-E^\Phi_{A,X_{R_k}})
  +2r\sqrt{s/N}\,I\right],\label{eq:pub-domain-input}\\
 \Delta^{\mathcal I}_{\ell,r}
 &\preceq\frac{N}{N-s+1}\left[
  \sum_{k=1}^{r}E^{=}_{A,Y_{R_k}}
  +\sum_{j=1}^{\ell}(E^{=}_{A,Y_{L_j}}-E^\Phi_{A,Y_{L_j}})
  +2\ell\sqrt{s/N}\,I\right].\label{eq:pub-image-input}
\end{align}
Here \(\Delta^{\mathcal D}_{\ell,r}
=\Pi_{\mathrm{bij},\ell,r}-\Pi_{\mathcal D(W)}\Pi_{\ell,r}\),
and we define the image defect analogously for \(\mathcal I\). We interpret
the right sides after compression to the separately symmetric tensor
embedding. The subtraction of the EPR projector accounts for coherent
deletion.

We keep these inequalities because the
phase--Hadamard boundaries need them without an exact-design assumption.
For exact designs, the corresponding twirls and forward sector bound
are imported from~\cite[Claims~1--2 and Section~11.4]{MH}:
\begin{align}
 \E_U\!\Bigl[(U\otimes U)^\dagger E^{=}(U\otimes U)\Bigr]
   &=\frac{I+\mathsf{Swap}}{N+1},\label{eq:auxiliary-same-twirl}\\
 \E_U\!\Bigl[(U\otimes\overline U)^\dagger(E^{=}-E^\Phi)(U\otimes\overline U)\Bigr]
   &=\frac{I-E^\Phi}{N+1}.\label{eq:auxiliary-mixed-twirl}
\end{align}
Here \(\mathsf{Swap}\) exchanges the two registers. The upper
bound for the averaged forward defect in sector \((\ell,r)\) is
\begin{equation}
 \frac{N}{N-s+1}
       \left[\frac{2\ell+r}{N+1}+2r\sqrt{s/N}\right].
 \label{eq:auxiliary-sector-defect}
\end{equation}
We use the resulting four-branch estimate directly as follows.

\begin{lemma}[Controlled design defect for CPFC]
\label{lem:auxiliary-design-defect}
Let \(G,D\) be independent samples from an exact unitary \(2\)-design. Set
\begin{equation}
 \Delta_{\mathrm{ctl},j}
 =\sum_{a=1}^4P_a\otimes
       (\Pi_{\mathrm{bij},\le j}-\Pi_{\mathcal D(W_a)}\Pi_{\le j}),\qquad
 G_{\mathrm{ctl}}=\sum_{a=0}^4P_a\otimes G_a,
 \label{eq:controlled-ordinary-defect}
\end{equation}
where \(\Pi_{\mathcal D(W_a)}\Pi_{\le j}\) restricts the domain
projector to input length at most \(j\). For \(0\le j<N\),
\begin{equation}
 \left\|\E_{G,D}\!\Bigl[(G_{\mathrm{ctl}}Q[G,D])^\dagger
       \Delta_{\mathrm{ctl},j}(G_{\mathrm{ctl}}Q[G,D])\Bigr]\right\|
       \le6j\sqrt{j/N}.
 \label{eq:auxiliary-design-defect}
\end{equation}
\end{lemma}
\begin{proof}
Use the four estimates~\cite[Lemma~24, equations~(B.70)--(B.73)]{SMLBH}
with zero twirling error. Ordinary exact \(2\)-designs suffice here:
the proof of that lemma uses only the two displayed twirls. Their
matrix entries have balanced degree \((2,2)\) in \(U,\overline U\), so
their averages equal the Haar averages.
Taking the maximum of the four block norms gives the stated bound.
\end{proof}

\subsection{One-query loss from boundary operators}
\label{app:controlled-stability}

Several PRU constructions have the form \(U=DCG\), where
\(G,D\) spread the query amplitudes and the core \(C\) admits a
recording description on a restricted subspace.
Table~\ref{tab:recording-constructions} lists examples. We apply the lemma
below to CPFC and the five-phase construction; the CLRFC and Kac rows
describe the restrictions in their cited proofs.
We analyze the
boundary unitaries of such constructions by bounding three errors: covariance error, loss
from the recorder domain, and loss from the additional good-subspace
projection. The lemma below combines these bounds into a one-query
overlap estimate and a bound on the norm retained after \(t\) queries.

\begin{table}[htbp]
\centering
\small
\renewcommand{\arraystretch}{1.15}
\begin{tabular}{@{}>{\raggedright\arraybackslash}p{.14\linewidth}
 >{\raggedright\arraybackslash}p{.33\linewidth}
 >{\raggedright\arraybackslash}p{.47\linewidth}@{}}
\toprule
Construction & Boundaries and core & Recording restriction \\
\midrule
CPFC & \(G=C_1,\ D=C_2\), Clifford; \(C=P_\pi F_f\)
 & Bijective pairs \((L,R)\); no additional record projection
   \cite[Appendix~B]{SMLBH}. \\
CLRFC & \(G=C_1,\ D=C_2\), exact \(2\)-designs;
   \(C=R_{h_2}L_{h_1}F_f\)
 & Distinct input right halves and output left halves; also restrictions
   on the live query \cite[Definition~41, Lemmas~27--28]{SMLBH}. \\
Kac walk & \(G=K_1,\ D=K_2\), independent walk segments;
   \(C=H_fP_\sigma\)
 & Distinct two-element blocks in each endpoint list; a modified
   recorder for forward and inverse queries \cite[Section~5]{LQSYZ26}. \\
Five phases & \(G=HF_1,\ D=F_5H\); \(C=F_4HF_3HF_2\)
 & Bijectivity and the XOR restrictions in \(P_{\mathrm{good}}\), as defined in
   Section~\ref{pubcore:central}. \\
\bottomrule
\end{tabular}
\caption{Boundary factors, recording cores, and restrictions. The rows
retain the query models and recording operators of the cited proofs.}
\label{tab:recording-constructions}
\end{table}

For CLRFC, write \(x=(x_L,x_R)\) with two equal halves. The maps
\(L_h,R_h\) update \(x_L\mapsto x_L\oplus h(x_R)\) and
\(x_R\mapsto x_R\oplus h(x_L)\), respectively, leaving the other half
unchanged; \(h_1,h_2\) are independent uniform functions on a half-register
\cite[Definition~38]{SMLBH}. In the Kac row, \(\sigma\) is a uniform
permutation and \(H_f\) applies sampled two-dimensional rotations on
the blocks \(\{0z,1z\}\). The boundaries \(K_1,K_2\) are independent
sequences of such steps \(H_fP_\sigma\), with the discretization and
walk lengths specified in~\cite[Section~4.1, Definition~13, Theorem~33]{LQSYZ26}.

\paragraph{The boundary sample.}
Fix \(t<N\) and a probability distribution \(\mu\) on boundary pairs \(b=(G,D)\).
Sample the pair once and retain its purification in the private
registers \(\mathsf{GD}\). For finite samples this is the usual state
\(\sum_b\sqrt{\mu(b)}\ket{b}\). To include Haar samples without formal
continuous basis vectors, use
\(\mathcal H_{\mathsf{GD}}=L^2(\mu)\) and represent the same purification
by the constant unit vector
\(\ket{\mathrm{init}(\mu)}_{\mathsf{GD}}:=\ket{\Omega}\), where
\(\Omega(b)=1\). The full state space and its norm are
\[
 \mathcal H_{\mathsf A}\otimes\mathcal H_{\mathsf B}
       \otimes\mathcal H_{\mathsf{LR}}\otimes\mathcal H_{\mathsf{GD}},
 \qquad
 \|\ket{\zeta}\|^2=\E_b\!\Bigl[\|\ket{\zeta(b)}\|^2\Bigr].
\]
Here \(\ket{\zeta(b)}\) is the vector associated with a fixed sampled
pair; its norm need not equal one. Define the sample-controlled
operators by
\[
 \begin{aligned}
  (\mathsf{cG}\ket{\zeta})(G,D)&=G_{\mathsf A}\ket{\zeta(G,D)},\\
  (\mathsf{cD}\ket{\zeta})(G,D)&=D_{\mathsf A}\ket{\zeta(G,D)},\\
  (\mathsf{cQ}\ket{\zeta})(G,D)&=Q[G,D]_{\mathsf{LR}}\ket{\zeta(G,D)}.
 \end{aligned}
\]
They apply the indicated unitary for each fixed sample and leave the
sample labels unchanged. These private sample registers are distinct
from the adversary's control register. The notation
\(\widehat W^b\) below likewise denotes the operator applying the
corresponding twirled oracle for each sample. For example, its forward
block is \(\mathsf{cD}\,W\,\mathsf{cG}\) and its inverse block is
\(\mathsf{cG}^\dagger W^\dagger\mathsf{cD}^\dagger\).
In particular,
\(\mathsf{cQ}(\ket{v}\otimes\ket{\Omega})\) assigns the vector
\(Q[G,D]\ket{v}\) to sample \((G,D)\). This is the action denoted by
\(Q_b\) on the purified sample space in the main body.

\paragraph{The restricted recording computation.}
Let \(P_{\mathrm{add}}\) be a boundary-independent, length-preserving
orthogonal projector on the private record space such that
\[
 [P_{\mathrm{add}},\Pi_{\mathrm{bij}}]=0,\qquad
 P_{\mathrm{add}}\ket{\varnothing,\varnothing}
       =\ket{\varnothing,\varnothing}.
\]
It imposes any additional record restrictions needed by the core.
Together with bijectivity, these define
\[
 P_{\mathrm{good}}=\Pi_{\mathrm{bij}}P_{\mathrm{add}}.
\]
The choice \(P_{\mathrm{add}}=I\) imposes only bijectivity. Let
\(A_1,\ldots,A_t\) be some isometries, extended
by identity on the private records and sample. Choose any unit initial
vector \(\ket{v_0}\) with empty relations, and define
\begin{align}
 \ket{K_0}&=\ket{v_0}\otimes\ket{\Omega},&
 \ket{K_{j+1}}&=P_{\mathrm{good}}\widehat W^bP_{\mathrm{good}} A_{j+1}\ket{K_j},&
 \ket{v_{j+1}}&=\widehat V A_{j+1}\ket{v_j},
 \label{eq:controlled-overlap-prefixes}
\end{align}
where \(\widehat W^b=\sum_aP_a\otimes D_aW_aG_a\).
In the notation for global adversary states, these are
\[
 \ket{\mathcal A_j^V}_{\mathsf{ABLR}}:=\ket{v_j},\qquad
 \ket{\mathcal A_j^{W,\mu;\mathrm{good}}}_{\mathsf{ABLRGD}}:=\ket{K_j}.
\]
The superscript \(\mathrm{good}\) records the additional projection \(P_{\mathrm{good}}\). The record length for both states is at most \(j\) after \(j\) queries. Thus each next-query
estimate uses the input cutoff \(\Pi_{\le j}\).

\paragraph{The three error bounds.}
For \(0\le j<t\), choose nonnegative bounds \(c_j,\gamma_j,b_{j+1}\)
satisfying the following inequalities.
\begin{itemize}[leftmargin=*]
\item \(c_j\), the \emph{covariance error}, bounds the difference between
applying the boundary unitaries to the query and rotating the records:
\[
 \|[\widehat V^bQ_b-Q_b\widehat V]\Pi_{\le j}\|\le c_j,
 \qquad \widehat V^b=\sum_aP_a\otimes D_aV_aG_a.
\]
This bound holds uniformly over the sampled boundary pairs.
Unrestricted EPR insertion has exact covariance, since
\((U\otimes\overline U)\ket{\Phi}=\ket{\Phi}\); the fresh-label
restrictions cause an approximation error \cite[Section~10]{MH}.
Lemma~\ref{lem:controlled-covariance} gives
\(c_j=16\sqrt{2j(j+1)/N}\) for every pair.

\item \(\gamma_j\), the \emph{boundary-domain loss}, bounds the averaged
squared norm outside \(\mathcal D(W_a)\), within the bijective sector.
For each nonidle branch \(a=1,\ldots,4\), put
\[
 \Delta_{a,j}=\Pi_{\mathrm{bij},\le j}
                  -\Pi_{\mathcal D(W_a)}\Pi_{\le j}.
\]
The bound is
\begin{equation}
 \E_b\!\Bigl[(G_aQ_b)^\dagger
       \Delta_{a,j}(G_aQ_b)\Bigr]
       \preceq\gamma_j I.
 \label{eq:general-controlled-defect}
\end{equation}
Since \(\Delta_{a,j}\) is a projector, this is equivalent to
\[
 \E_b\!\Bigl[\|\Delta_{a,j}G_aQ_b\ket{\psi}\|^2\Bigr]
       \le\gamma_j\|\ket{\psi}\|^2
\]
for every boundary-independent vector \(\ket{\psi}\).
The EPR subtraction in~\eqref{eq:pub-domain-input} retains coherent deletion.

\item \(b_{j+1}\), the \emph{additional-record loss}, bounds the squared
norm discarded by \(P_{\mathrm{add}}\), the projection onto the subspace
satisfying the additional conditions required by the core:
\begin{equation}
 \|(I-P_{\mathrm{add}})\mathsf{cQ}
       (\ket{v_{j+1}}\otimes\ket{\Omega})\|^2
       \le b_{j+1}.
 \label{eq:general-additional-loss}
\end{equation}
The projection acts on the rotated reference after query \(j+1\).
Its choice depends on the recording core; for the five-phase construction,
it imposes the XOR restrictions. This estimate is required only on the
rotated reference, not on the boundary-dependent \(\ket{K_j}\).
Here \(b_{j+1}\) is a scalar bound, distinct from the boundary pair
\(b=(G,D)\).
\end{itemize}
Thus the overlap loss per query is
\(c_j+\sqrt{\gamma_j}+\sqrt{b_{j+1}}\). For the two constructions proved here,
the inputs are
\begin{equation}
\begin{array}{c|c|c|c}
 & c_j & \gamma_j & b_j\\ \hline
 \text{CPFC (exact designs)}
 & 16\sqrt{2j(j+1)/N} & 6j\sqrt{j/N} & 0\\[2pt]
 \text{Five phases}
 & 16\sqrt{2j(j+1)/N} & 8j\sqrt{j/N} & 9j^4/N
\end{array}
\label{eq:recording-error-budgets}
\end{equation}
The first row uses Lemmas~\ref{lem:controlled-covariance}
and~\ref{lem:auxiliary-design-defect}, with \(P_{\mathrm{add}}=I\).
The second uses Lemmas~\ref{lem:pub-boundary-defects}
and~\ref{lem:pub-boundary-good-mass}.

The next lemma converts these three estimates into overlap and norm
bounds. A comparison with the physical core still requires its
encoding identity or compression estimate, followed by an
error-propagation argument as in
Lemma~\ref{lem:approximate-compression-transfer}.

Set
\[
 \mathcal B_t=\sum_{j=0}^{t-1}
          (c_j+\sqrt{\gamma_j}+\sqrt{b_{j+1}}),\qquad
 s_j=\operatorname{Re}\bra{K_j}\mathsf{cQ}
       (\ket{v_j}\otimes\ket{\Omega}).
\]
\begin{lemma}[Stability of controlled restricted recording]
\label{lem:controlled-overlap}
Under the preceding setup and three error bounds, each query
\(0\le j<t\) satisfies
\begin{equation}
 s_{j+1}\ge s_j-c_j-\sqrt{\gamma_j}-\sqrt{b_{j+1}}.
 \label{eq:general-one-query-loss}
\end{equation}
Consequently, the overlap, vector distance, and discarded squared norm obey
\begin{align}
 \operatorname{Re}\bra{K_t}\mathsf{cQ}(\ket{v_t}\otimes\ket{\Omega})
      &\ge1-\mathcal B_t,\label{eq:general-controlled-overlap}\\
 \|\ket{K_t}-\mathsf{cQ}(\ket{v_t}\otimes\ket{\Omega})\|&\le\sqrt{2\mathcal B_t},&
 1-\|\ket{K_t}\|^2&\le2\mathcal B_t.
 \label{eq:general-controlled-overlap-consequences}
\end{align}
\end{lemma}
\begin{proof}
Both global states have norm at most one and record length at most \(j\),
and \(P_{\mathrm{good}}\ket{K_j}=\ket{K_j}\).
Since the private rotation fixes the empty record, \(s_0=1\).
Before the next query, write
\(\ket{k}=A_{j+1}\ket{K_j}\), \(\ket{\psi}=A_{j+1}\ket{v_j}\), \(\ket{k_a}=P_a\ket{k}\), \(\ket{\psi_a}=P_a\ket{\psi}\), and
\[
 \ket{z}=\sum_aD_aW_aG_a\ket{k_a},\qquad
 \ket{w}=\mathsf{cQ}(\ket{v_{j+1}}\otimes\ket{\Omega}).
\]
The vectors \(\ket{\psi_a}\) are independent of \(b\). We estimate the overlap
with \(\ket{z}\) by using covariance to replace \(\ket{w}\) with
\(\sum_aD_aV_aG_a\mathsf{cQ}(\ket{\psi_a}\otimes\ket{\Omega})\) with error at most
\(c_j\|\ket{k}\|\|\ket{\psi}\|\le c_j\). Cancelling the output factors and
using~\eqref{eq:controlled-restrictions} replaces each product
\(W_a^\dagger V_a\) with \(a\ne0\) by
\(\Pi_{\mathcal D(W_a)}\Pi_{\le j}\) on these inputs.

Since \(\bra{k_a}G_a^\dagger\) is
supported on \(\Pi_{\mathrm{bij},\le j}\) for every branch, including
\(a=0\), we can write the five products inside this matrix element as
\(I\otimes\Pi_{\mathrm{bij},\le j}-\Delta_{\mathrm{ctl},j}\).
The first term gives \(s_j\): the bra is bijective, and \(A_{j+1}\) is an
isometry commuting with \(\mathsf{cQ}\). This substitution uses the bra's support;
the idle operation remains identity and the rotated reference may be
nonbijective.

Write \(\Delta_{a,j}=\Pi_{\mathrm{bij},\le j}-\Pi_{\mathcal D(W_a)}\Pi_{\le j}\).
Since these are projectors, \eqref{eq:general-controlled-defect} gives
\[
 \|\Delta_{a,j}G_a\mathsf{cQ}(\ket{\psi_a}\otimes\ket{\Omega})\|^2
       \le\gamma_j\|\ket{\psi_a}\|^2.
\]
We sum the defect terms using Cauchy--Schwarz over the branches:
\begin{align*}
 \sum_{a=1}^4\|\ket{k_a}\|\,
       \|\Delta_{a,j}G_a\mathsf{cQ}(\ket{\psi_a}\otimes\ket{\Omega})\|
 &\le\left(\sum_a\|\ket{k_a}\|^2\right)^{1/2}
       \left(\gamma_j\sum_a\|\ket{\psi_a}\|^2\right)^{1/2}
 \le\sqrt{\gamma_j}.
\end{align*}
Only the \(\ket{\psi_a}\) need to be independent of the boundaries in this
calculation; the \(\ket{k_a}\) may depend on them. We have shown
\(\operatorname{Re}\bra{z}\ket{w}\ge s_j-c_j-\sqrt{\gamma_j}\).
Each \(\ket{k_a}\) is bijective, each \(W_a\) preserves that sector, and
the boundary factors leave the relations unchanged. Thus
\(P_{\mathrm{good}}\ket{z}=P_{\mathrm{add}}\ket{z}\) and
\[
 |\bra{z}P_{\mathrm{good}}\ket{w}-\bra{z}\ket{w}|
   =|\bra{z}(P_{\mathrm{add}}-I)\ket{w}|
   \le\sqrt{b_{j+1}}.
\]
Hence
\[
 s_{j+1}\ge s_j-c_j-\sqrt{\gamma_j}-\sqrt{b_{j+1}},\qquad
 s_t\ge1-\mathcal B_t.
\]
Since both states have norm at most one and \(s_t\le\|\ket{K_t}\|\),
\[
 \begin{aligned}
  \|\ket{K_t}-\mathsf{cQ}(\ket{v_t}\otimes\ket{\Omega})\|^2
  &\le2(1-s_t),\\
  1-\|\ket{K_t}\|^2&\le2(1-\|\ket{K_t}\|)\le2(1-s_t).
 \end{aligned}
\]
\end{proof}
We apply the lemma to the CPFC construction with
\(P_{\mathrm{add}}=I\), \(b_j=0\), and
\(\gamma_j=6j\sqrt{j/N}\) from Lemma~\ref{lem:auxiliary-design-defect}.
For \(j\ge1\), \(c_j\le32j/\sqrt N\) and
\(\sqrt{\gamma_j}\le3jN^{-1/4}\); both vanish at \(j=0\). Therefore
\begin{equation}
 \mathcal B_t\le35t^2N^{-1/4},\qquad
 \|\ket{K_t}-\mathsf{cQ}(\ket{v_t}\otimes\ket{\Omega})\|\le9tN^{-1/8},\qquad
 1-\|\ket{K_t}\|^2\le70t^2N^{-1/4}.
 \label{eq:auxiliary-overlap-bounds}
\end{equation}
Since \(P_{\mathrm{good}}=\Pi_{\mathrm{bij}}\), the projected recurrence agrees with
the twirled \(W\) recurrence without an additional projection. For the five-phase construction,
we instead use its phase--Hadamard defect bound and the additional XOR
projection. 

\subsection{From the twirled recording state to the controlled Haar comparison}

We compare the purified permutation--function oracle with twirled \(W\),
then apply Section~\ref{app:controlled-stability} to compare twirled \(W\)
with \(\widehat V\). Both comparisons retain the private environments.

Take \(P_{\mathrm{add}}=I\), and purify independent exact-design boundaries
in \(\ket{\mathrm{init}(\mu)}=\ket{\Omega}\). The core sample \((\pi,f)\)
is independent of these boundaries and is purified in
\(\ket{\Omega_{PF}}\). Define
\[
 \widehat O^b=\sum_{a=0}^4P_a\otimes D_aO_aG_a,
 \qquad
 \widehat W^b=\sum_{a=0}^4P_a\otimes D_aW_aG_a.
\]
The first purifies the physical oracle for the single sampled matrix
\(DP_\pi F_fG\). It includes idle, forward, inverse, transpose and
conjugate, using the same samples and global phase in every branch.
The second preserves bijectivity, so the projections in
\eqref{eq:controlled-overlap-prefixes} have no further effect here.
Consequently, the global adversary states satisfy
\[
 \begin{aligned}
 \ket{\Psi_0}&=J_{PF}\ket{K_0},&
 \ket{\Psi_i}&=\widehat O^bA_i\ket{\Psi_{i-1}},\\
 &&\ket{K_i}&=\widehat W^bA_i\ket{K_{i-1}}.
 \end{aligned}
\]
Here
\(\ket{\mathcal A_i^{\mathsf{pfO},\mu}}_{\mathsf{ABPFGD}}:=\ket{\Psi_i}\)
and \(\ket{\mathcal A_i^{W,\mu}}_{\mathsf{ABLRGD}}:=\ket{K_i}\);
the superscript \(\mathrm{good}\) can be omitted because there is no additional
restriction in this construction. The first equality follows from
\(J_{PF}\ket{\varnothing,\varnothing}=\ket{\Omega_{PF}}\).
We extend \(J_{PF}\) by identity on the boundaries and accessible
registers. It commutes with \(A_i,G_a,D_a\).

Apply the exact domain identity~\eqref{eq:controlled-pf-one-step}
separately to the orthogonal control blocks. Since \(G_a,D_a\) are
unitary, for every bijective input \(\ket{\psi}\) we have
\[
 \|\widehat O^bJ_{PF}\ket{\psi}-J_{PF}\widehat W^b\ket{\psi}\|^2
    =\|\ket{\psi}\|^2-\|\widehat W^b\ket{\psi}\|^2.
\]
This identity holds for each fixed sampled pair \(b\), so it remains
valid after averaging, even when the input vector depends on \(b\).
At recording call \(i\), the input length is at most \(i-1\) and the
output length is at most \(i\).
Thus the same encoding restricted to records of size at most \(t\)
contains every state of the recording computation to which \(J_{PF}\) is applied.

In particular, define the one-step discrepancy and loss of squared norm by
\[
 \ket{\varepsilon_i}
    =(\widehat O^bJ_{PF}-J_{PF}\widehat W^b)A_i\ket{K_{i-1}},
 \qquad
 d_i=\|\ket{K_{i-1}}\|^2-\|\ket{K_i}\|^2.
\]
Then \(\|\ket{\varepsilon_i}\|^2=d_i\ge0\). The telescoping argument
in Lemma~\ref{lem:approximate-compression-transfer}, with
\(U_i=\widehat O^bA_i\), now gives
\[
 \ket{\Psi_t}-J_{PF}\ket{K_t}
   =\sum_{i=1}^tU_t\cdots U_{i+1}\ket{\varepsilon_i},
 \qquad
 \sum_{i=1}^td_i=1-\|\ket{K_t}\|^2.
\]
We may take unitary extensions of the accessible isometries. The
triangle inequality and Cauchy--Schwarz therefore yield
\begin{equation}
 \|\ket{\Psi_t}-J_{PF}\ket{K_t}\|
 \le\sum_{i=1}^t\sqrt{d_i}
 \le\sqrt{t(1-\|\ket{K_t}\|^2)}
 \le\sqrt{2t\mathcal B_t}.
 \label{eq:controlled-pf-gentle}
\end{equation}

To compare accessible marginals, follow
\(\ket{\Psi_t}\leftrightarrow J_{PF}\ket{K_t}\leftrightarrow \ket{K_t}
\leftrightarrow \mathsf{cQ}(\ket{v_t}\otimes\ket{\Omega})\).
The private isometry \(J_{PF}\) preserves the marginal of \(\ket{K_t}\), and
\(\mathsf{cQ}\) preserves that of \(\ket{v_t}\otimes\ket{\Omega}\). We apply \(J_{PF}\) only to
bijective \(\ket{K_t}\), not to the possibly nonbijective rotated reference.
Equations~\eqref{eq:controlled-subnormalized-distance} and
\eqref{eq:general-controlled-overlap-consequences} give total error
at most
\[
 \sqrt{2t\mathcal B_t}+\sqrt{2\mathcal B_t}
 \le\sqrt{70}(t^{3/2}+t)N^{-1/8}
 \le9t(t+1)N^{-1/8}
 \qquad(t\ge1).
\]
Consequently,
\begin{equation}
 \frac12\left\|
   \E_{\pi,f,G,D}\!\Bigl[\rho_t^{DP_\pi F_fG}\Bigr]
       -\Tr_{LR}\ket{v_t}\bra{v_t}\right\|_1
       \le9t(t+1)N^{-1/8}.
 \label{eq:controlled-auxiliary-haar-bound}
\end{equation}
The average is obtained by tracing the one common sample purification,
so it retains every branch cross term.

To prove Theorem~\ref{thm:controlled-haar}, take independent Haar
\(G,D\), which are exact unitary \(2\)-designs. Conditional on
\(\pi,f,G\), Haar invariance gives \(DP_\pi F_fG\sim\Haar\).
This identifies the joint distribution of all five branches of the
same sampled matrix, so~\eqref{eq:controlled-auxiliary-haar-bound}
is the claimed Haar comparison. 

\subsection{Global and reduced states for the five-phase construction}
\label{app:five-phase-transfer}

We prove Proposition~\ref{prop:controlled-recorder-transfer} using
\(\ket{K_m}\) from~\eqref{eq:controlled-overlap-prefixes}. Write
\(\ket{u_m}=\mathsf{cQ}(\ket{v_m}\otimes\ket{\Omega})\).
Lemmas~\ref{lem:controlled-covariance}, \ref{lem:pub-boundary-defects},
and~\ref{lem:pub-boundary-good-mass} give
\[
 c_j=O(jN^{-1/2}),\qquad
 \gamma_j=O(j^{3/2}N^{-1/2}),\qquad b_j=O(j^4/N).
\]
The last estimate applies because Lemma~\ref{lem:reference-forest}
makes \(\ket{v_j}\) forest-supported and independent of the boundary tables.
For \(m\le t\), their sum in Lemma~\ref{lem:controlled-overlap} satisfies
\[
 \mathcal B_m=\sum_{j=0}^{m-1}(c_j+\sqrt{\gamma_j}+\sqrt{b_{j+1}})
       =O(t^2N^{-1/4}+t^3N^{-1/2}).
\]
That lemma gives
\[
 \|\ket{K_m}-\ket{u_m}\|=O(\sqrt{\mathcal B_m}),\qquad
 1-\|\ket{K_m}\|^2=O(\mathcal B_m).
\]

By Proposition~\ref{pubcore:compression}, the restricted steps have
compression error at most \(\eta_t\): the encoding acts only on the
private records and commutes with accessible operations and live
boundary factors. Equation~\eqref{pubcore:contractive-encoding} encodes
the empty record exactly, so
\begin{equation}
 \|\ket{\Phi_0}-T_{\mathrm{good}}\ket{K_0}\|=d_0=0.
 \label{eq:controlled-seed-correction}
\end{equation}
Using unitary extensions for accessible isometries,
Lemma~\ref{lem:approximate-compression-transfer} and the preceding
norm-loss bound give
\[
 \|\ket{\Phi_m}-T_{\mathrm{good}}\ket{K_m}\|=O(t^2N^{-1/8}).
\]
Since \(P_{\mathrm{good}}\ket{K_m}=\ket{K_m}\) and
\(E_t\ket{v_m}=T_{\mathrm{good}}P_{\mathrm{good}}\ket{u_m}\),
contractivity gives
\(\|T_{\mathrm{good}}\ket{K_m}-E_t\ket{v_m}\|=O(\sqrt{\mathcal B_m})\).
Thus
\begin{equation}
 \|\ket{\Phi_m}-E_t\ket{v_m}\|=O(t^2N^{-1/8}).
 \label{eq:controlled-full-environment-transfer}
\end{equation}
For the sample values \(\ket{q_{v_m}(S)}=(E_t\ket{v_m})(S)\),
the full-environment norm is exactly
\[
 \|\ket{\Phi_m}-E_t\ket{v_m}\|^2
 =\E_S\Bigl[\|\ket{\Phi_m(S)}-\ket{q_{v_m}(S)}\|^2\Bigr].
\]
The sample-conditioned vectors retain their original norms; the bound
averages over the same row \(S\) on both sides.

For the marginal bound, let \(\rho_K\) and
\(\rho_{T_{\mathrm{good}}K}\) denote accessible states after tracing all
private systems. Their difference is the positive semidefinite marginal
of \((I\otimes\sqrt{I-T_{\mathrm{good}}^\dagger T_{\mathrm{good}}})\ket{K}\), so
\begin{equation}
 \frac12\|\rho_K-\rho_{T_{\mathrm{good}}K}\|_1
 =\frac12\bra{K}(I-T_{\mathrm{good}}^\dagger T_{\mathrm{good}})\ket{K}
 \le\frac{\varepsilon_{\mathrm{cf},t}}2\|\ket{K}\|^2.
 \label{eq:controlled-marginal-correction}
\end{equation}
The private rotation \(\mathsf{cQ}\) preserves the accessible marginal of
\(\ket{v_m}\). Apply~\eqref{eq:controlled-subnormalized-distance} and
the triangle inequality along
\(\rho_{\Phi_m},\rho_{T_{\mathrm{good}}K_m},\rho_{K_m},\rho_{u_m}\).
The preceding bounds give
\[
 \frac12\|\rho_{\Phi_m}-\rho_{u_m}\|_1=O(t^2N^{-1/8}).
\]
All estimates retain arbitrary independent references. No intermediate
state or sample-conditioned vector is renormalized.

Choose a common bound for the vector and marginal errors and denote it by
\begin{equation}
 e_t=O(t^2N^{-1/8}).
 \label{eq:pub-prefix-error}
\end{equation}
Theorem~\ref{thm:controlled-haar} adds a Haar-comparison error of the
same order, giving
\begin{equation}
 \Delta_t=O(t^2N^{-1/8}).
 \label{eq:construction-rates}
\end{equation}
The implied constants are universal.

\section{Classical search and the black-box consequence}
\label{sec:oracle-search}

A completion of the universal \(\mathsf{QMA}\) verifier can find a classical
solution even when some of its answers lie outside the promise.
The search below allows a small loss in verification probability at
each selected bit. This proves a search lemma first. Function inversion,
puzzle solving, and the fully black-box separation then follow as
separate consequences.

Throughout this appendix, fix \(R\) and put \(O_R=(R,B_R)\).
Let \(Q\) be any Boolean completion of the universal verifier over
\(O_R\). Thus \(Q\) answers one on every YES input and zero on
every NO input, and it has an arbitrary fixed answer on every other
input. The search and inversion claims hold for every such completion;
only the PRU security theorem uses the distribution
\eqref{eq:boolean-completion-law}.

\subsection{Approximate search for a classical solution}

\begin{lemma}[Approximate classical search]
\label{lem:completed-oracle-search}
Let \(\mathcal V_{\mathrm{search}}^{O_R}(y,w)\) be a uniform
quantum verifier with a classical solution \(w\in\{0,1\}^m\),
where \(m\ge1\). Define
\[
 \operatorname{acc}_y(w)
 :=\Pr[\mathcal V_{\mathrm{search}}^{O_R}(y,w)\text{ accepts}].
\]
For rational \(0<\epsilon<c\le1\), a deterministic algorithm
with \(m\) queries to \(Q\) returns a string \(w\in\{0,1\}^m\).
If \(\max_w\operatorname{acc}_y(w)\ge c\), its output satisfies
\[
 \operatorname{acc}_y(w)\ge c-\epsilon.
\]
The algorithm terminates on every input, including inputs that fail
this hypothesis. Its running time and query lengths are polynomial in
the verifier's running time, \(|y|\), \(m\), \(1/\epsilon\), and
the bit lengths of \(c,\epsilon\).
The verifier queries only \(O_R\); the search algorithm queries
only \(Q\).
\end{lemma}

\begin{proof}
Put \(\delta=\epsilon/m\). We construct prefixes
\(p_0,p_1,\ldots,p_m\), with \(|p_j|=j\) and
\(p_0\) empty. When the hypothesis holds, we maintain
\begin{equation}
 \max_{w\text{ extending }p_j}\operatorname{acc}_y(w)
 \ge c-j\delta.
 \label{eq:classical-search-invariant}
\end{equation}
The final prefix is a full solution, so this invariant proves the
claimed guarantee. We establish it by specifying one query at each
step and checking both possible answers.

\paragraph{The query at step \(j\).}
For \(1\le j\le m\), construct a \(\mathsf{QMA}\) verifier that first measures
its \(m\)-qubit witness in the computational basis. Let \(w\) be
the resulting string. Reject unless \(w\) extends \(p_{j-1}0\).
Otherwise run \(\mathcal V_{\mathrm{search}}^{O_R}(y,w)\) on
fresh work registers
\(h=\lceil4\delta^{-2}\rceil\) times. Accept if the
fraction of accepting runs exceeds
\(c-(j-\tfrac12)\delta\).
All computations after the witness measurement use this same
classical string, and the repetitions are independent conditional
on it.

Hoeffding's inequality bounds either tail across a gap of
\(\delta/2\) by
\(\exp(-h\delta^2/2)\le e^{-2}<1/6\). Hence the query is a YES instance
if some extension of \(p_{j-1}0\) has value at least
\(c-(j-1)\delta\). It is a NO instance if every such extension
has value at most \(c-j\delta\). Indeed, in the first case the
corresponding classical witness is accepted with probability at least
\(5/6\). In the second case every classical witness is accepted with
probability at most \(1/6\). Measuring an arbitrary quantum witness
makes its acceptance probability a convex combination of these
classical acceptance probabilities, so it cannot violate soundness.
These bounds imply the universal verifier's \(2/3\) and \(1/3\)
promise thresholds.

\paragraph{Either answer preserves the invariant.}
If \(Q\) answers one, set \(p_j=p_{j-1}0\). This query cannot
be a NO instance. The preceding soundness implication therefore shows
that some extension has value greater than \(c-j\delta\).
If \(Q\) answers zero, set \(p_j=p_{j-1}1\). The query cannot
be a YES instance, so no zero-branch extension has value at least
\(c-(j-1)\delta\). The preceding invariant supplies an extension
at that value or above; it must lie in the one branch.
Thus \eqref{eq:classical-search-invariant} holds in both cases.
The argument also covers either answer on a query outside the promise.

There are exactly \(m\) steps. Their cumulative loss is at most
\(m\delta=\epsilon\). Each query describes a circuit of polynomial
size because \(h=O(m^2/\epsilon^2)\), and rational threshold
comparison is a classical polynomial-time computation. The search
algorithm constructs these descriptions; it never executes the
verification circuits. It always appends one bit and halts after
\(m\) queries, even when the initial invariant is false.
\end{proof}

\subsection{Function inversion and puzzle solving}

A quantum-computable classical-output function has a fixed classical
value at each input. A uniform quantum algorithm outputs that value
with probability at least \(2/3\) on every input, as in
\cite[Definition~9]{KQT25}. We use this convention relative to
\(O_R\), without giving the evaluator access to \(Q\).

\begin{corollary}[Inversion with a total deterministic algorithm]
\label{cor:base-function-inversion}
Every quantum-computable classical-output function over \(O_R\)
has a deterministic polynomial-time inverter with access to \(Q\).
Given \(1^m\) and an output in the range on \(m\)-bit inputs, it
returns an \(m\)-bit preimage. On all other inputs it still returns
an \(m\)-bit string in polynomial time.
Its coherent XOR oracle can be implemented exactly by a uniform
polynomial-size quantum circuit using the same Boolean oracle \(Q\).
\end{corollary}
\begin{proof}
Let \(f\) be the function and let its evaluator be fixed. On input
\((1^m,y)\), define a verifier that accepts a classical string \(w\)
exactly when a fresh evaluation on \(w\) outputs \(y\).
If \(f(w)=y\), its acceptance probability is at least \(2/3\).
If \(f(w)\ne y\), that probability is at most \(1/3\), since
the evaluator outputs \(f(w)\) with probability at least \(2/3\).
Apply Lemma~\ref{lem:completed-oracle-search} with \(c=2/3\)
and \(\epsilon=1/6\). On an output in the range, the returned
string has acceptance probability at least \(1/2\), so it must be
a preimage. For \(m=0\), return the empty string directly.

For a fixed table \(Q\), the entire search is deterministic and
has a polynomial running-time bound on every input. Denote its total
output function by \(\mathcal I^Q(1^m,y)\). Reversibly simulate
this computation, copy its output into the target register by XOR,
and reverse the computation to erase the work registers. Every
Boolean query gate is its own inverse and uses the same fixed \(Q\)
during the forward and reverse computations. This implements
\[
 |1^m,y,z\rangle
 \longmapsto
 |1^m,y,z\oplus\mathcal I^Q(1^m,y)\rangle
\]
with clean work registers and polynomially many \(Q\)-queries.
The circuit identity holds on every basis input, including inputs
outside the range, and therefore on their superpositions.
\end{proof}

\begin{corollary}[Solving puzzles with classical challenges and solutions]
\label{cor:base-puzzle-search}
Let \(\mathcal V_{\mathrm{search}}^{O_R}\) be a polynomial-time
quantum verifier for classical puzzles and classical solutions of a
known polynomially bounded length. Suppose a sampler produces pairs
\((y,w)\) with
\(\mathbb E_{y,w}\operatorname{acc}_y(w)\ge1-\eta\).
A polynomial-time algorithm with access to \(Q\), given the sampled
puzzle \(y\), returns a solution whose average verification probability
is at least \((1-3\eta)/2\).
\end{corollary}
\begin{proof}
For each puzzle set
\(a_*(y)=\max_w\operatorname{acc}_y(w)\).
The sampler's guarantee implies
\(\mathbb E_y[1-a_*(y)]\le\eta\).
Markov's inequality gives
\(\Pr_y[a_*(y)<2/3]\le3\eta\).
On every remaining puzzle, Lemma~\ref{lem:completed-oracle-search}
with \(c=2/3\) and \(\epsilon=1/6\) returns a solution accepted
with probability at least \(1/2\). Acceptance probabilities are
nonnegative on the other puzzles, which proves the stated average.
\end{proof}

\subsection{The fully black-box consequence}

We specify the oracle access in the reduction. A keyed unitary oracle
\(\mathcal P\) supports coherent keys and all five control branches
from Section~\ref{sec:construction}. In its PRU experiment, the
adversary also receives either a uniformly keyed unitary or an
independent Haar unitary as the challenge. The public oracle
\(\mathcal P\) stays the same in the two experiments.

A \emph{uniform fully black-box construction} in this model consists
of uniform polynomial-time oracle algorithms
\((\mathcal C,\mathcal R)\). For every \(\mathcal P\),
\(\mathcal C^{\mathcal P}\) evaluates a classical-output function
with error at most \(1/3\) on every input. Given any total inverter
oracle \(\mathcal I\) that inverts this function with non-negligible
probability on a uniform input, the reduction
\(\mathcal R^{\mathcal P,\mathcal I}\) has non-negligible advantage
in the PRU experiment. The inverter can be inefficient; the reduction
accesses it only as an oracle, possibly through its coherent XOR gate.
The construction has no access to \(Q\). Its input and output lengths,
and the resources of both algorithms, are polynomial in the security
parameter.

\begin{corollary}[Fully black-box separation]
\label{cor:fully-black-box-separation}
There is no uniform fully black-box construction in the preceding
model of quantum-computable classical-output one-way functions from
strong PRUs.
\end{corollary}
\begin{proof}
Suppose \((\mathcal C,\mathcal R)\) is such a construction.
Choose a fixed pair \((R,Q)\) for which Theorem~\ref{thm:main}
holds, and instantiate \(\mathcal P\) by the keyed five-phase
family. We obtain a contradiction in two steps: build a total inverter
using \(Q\), and implement the resulting reduction with the public
oracle access allowed by that theorem.

\paragraph{The evaluator and inverter.}
Appendix~\ref{sec:phase-implementation} implements coherent keyed
queries to \(\mathcal P\) through \(R\), including all five
branches. First amplify the success probability of
\(\mathcal C^{\mathcal P}\) above \(5/6\) by independent
repetition and majority of its classical outputs. Then choose the
phase-synthesis accuracy so that replacing its polynomially many
unitary queries changes the output distribution by less than \(1/12\).
This gives a uniform \(R\)-oracle evaluator for the same classical
function with success probability above \(2/3\) on every input.
The evaluator queries no completion bits.

Apply Corollary~\ref{cor:base-function-inversion} with the security
parameter included in the public input and the construction's known
input-length rule. It gives a total deterministic inverter
\(\mathcal I^Q\) with access only to \(Q\). It inverts every output in the range. In particular, its
success on a uniformly chosen input is one, so the assumed reduction
applies to this inverter. If the reduction uses coherent inversion
queries, the same corollary implements their XOR gates exactly by
polynomially many queries to the fixed table \(Q\).

\paragraph{Compose the reduction with the inverter.}
The composed adversary runs
\(\mathcal R^{\mathcal P,\mathcal I^Q}\). It implements public
\(\mathcal P\)-queries through \(R\), implements
\(\mathcal I^Q\)-queries through \(Q\), and forwards challenge
queries to its given challenge oracle. It therefore has precisely
public \((R,Q)\)-access and the challenge access allowed by
Theorem~\ref{thm:main}. In particular, the extra oracle used by the
inverter is present in the composed adversary.
Choose negligible synthesis error for the polynomially many public
unitary queries. The resulting uniform adversary retains the
reduction's non-negligible advantage. This contradicts
Theorem~\ref{thm:main}.
\end{proof}

The inversion and puzzle claims concern evaluation and verification
over \(O_R\), with classical challenges and solutions. They make
no claim about functions evaluated with \(Q\), puzzles verified with
\(Q\), or primitives whose challenges are quantum states. The
separation concerns the specified fully black-box constructions and
uniform reductions.

\section{A sensitivity bound from polynomial certificates}
\label{app:scalar-stability-tools}

We prove Lemma~\ref{lem:polynomial-max-sensitivity}. First, we state
the certificate bound for bounded polynomials. Next, we prove a
sensitivity bound for real-valued functions with one-sided
certificates. We then apply this bound to a maximum of polynomials.

\subsection{Polynomial certificates}

\begin{lemma}[Bounded polynomial certificates, {\cite[Theorem~1.7]{AB}}]
\label{lem:oracle-polynomial-certificate}
Let \(m\ge0\) and \(d\ge1\) be integers. Let
\(p:\{0,1\}^m\to[0,1]\) be a polynomial of degree at most
\(d\). At every input \(x\), there are nonnegative weights
\(c_i(x)\), \(i\in[m]\), such that
\begin{equation}
 \begin{aligned}
 |p(x)-p(x')|&\le\sum_{i:x_i\ne x'_i}c_i(x)
                   \quad\text{for every }x'\in\{0,1\}^m,\\
 \sum_{i=1}^m c_i(x)&\le\frac{\pi^2d^2}{4}.
 \end{aligned}
 \label{eq:oracle-polynomial-certificate}
\end{equation}
These weights form an \emph{absolute-difference certificate} at
\(x\). Their sum is its \emph{mass}.
\end{lemma}
\begin{proof}
Use \(x_i^2=x_i\) on Boolean inputs to make \(p\) multilinear
without increasing its degree. For \(m\ge1\), apply
\cite[Theorem~1.7]{AB}. For \(m=0\), the cube has one input and
the empty family of weights satisfies both bounds.
\end{proof}

\subsection{Sensitivity from one-sided certificates}

For \(f:\{0,1\}^m\to[0,1]\), a \emph{one-sided certificate}
at \(y\) consists of nonnegative weights \(w_y(i)\), \(i\in[m]\),
such that
\[
 f(y)-f(x)\le\sum_{i:y_i\ne x_i}w_y(i)
 \quad\text{for every }x\in\{0,1\}^m.
\]
Its mass is \(\sum_iw_y(i)\). Recall that
\[
 S_\delta(f,x)
 :=\#\{i\in[m]:|f(x)-f(x^{\oplus i})|\ge\delta\}.
\]

\begin{lemma}[A sensitivity bound from one-sided certificates]
\label{lem:one-sided-certificate-sensitivity}
Let \(m\ge0\) and \(C\ge1\). Suppose every input of
\(f:\{0,1\}^m\to[0,1]\) has a one-sided certificate of mass
at most \(C\). For uniform \(X\in\{0,1\}^m\), every
\(0<\delta\le1\), and every \(t\ge0\),
\begin{equation}
 \Pr[S_\delta(f,X)\ge t]
 \le\frac2\delta\,2^{-t\delta/(16C)}.
 \label{eq:one-sided-sensitivity-tail}
\end{equation}
\end{lemma}
\begin{proof}
Fix one certificate \(w_y\) at each input \(y\). Let
\[
 \begin{aligned}
 U_\delta(x)&:=\#\{i:f(x^{\oplus i})-f(x)\ge\delta\},\\
 D_\delta(x)&:=\#\{i:f(x)-f(x^{\oplus i})\ge\delta\}.
 \end{aligned}
\]
Then \(S_\delta(f,x)=U_\delta(x)+D_\delta(x)\). Each coordinate
counted by \(D_\delta(x)\) requires \(w_x(i)\ge\delta\), so
\(D_\delta(x)\le C/\delta\). We will prove that, for every
\(u>0\),
\begin{equation}
 \Pr[U_\delta(X)\ge u]
 \le\frac2\delta\,2^{-u\delta/(8C)}.
 \label{eq:upward-sensitivity-tail}
\end{equation}
This suffices: if \(t>2C/\delta\), then
\(S_\delta(f,x)\ge t\) implies \(U_\delta(x)\ge t/2\).
If \(t\le2C/\delta\), the right side
of~\eqref{eq:one-sided-sensitivity-tail} is at least one.

We prove~\eqref{eq:upward-sensitivity-tail} by the encoding method
of the satisfiability coding lemma~\cite{PPZ97}. For each level
\(a\in[0,1]\), we encode the sublevel set
\(E_a:=\{x:f(x)\le a\}\). Upward jumps permit shorter codes.
We then average the resulting count over \(a\).

\emph{Step 1: Encode each sublevel set.}
Fix \(a\) and a coordinate order \(\pi\). Inspect the coordinates
in order \(\pi\). At each step, consider the inputs in \(E_a\)
consistent with the preceding bits. If they all have the same next
bit, infer it. Otherwise, write that bit. The decoder maintains the
same set and reads a code bit exactly when both values remain possible.

The inputs in \(E_a\) are the leaves of a binary tree in this
coordinate order. Contract each node with one child. The remaining
paths contain exactly the written bits and are prefix-free. Thus
their lengths \(\ell_{\pi,a}(x)\) satisfy Kraft's inequality:
\[
 \sum_{x\in E_a}2^{-\ell_{\pi,a}(x)}\le1.
\]

\emph{Step 2: Bound the expected code length.}
Fix \(x\) and \(a\) with
\(f(x)\le a\le f(x)+\delta/2\). We claim that a uniform
coordinate order satisfies
\[
 \mathbb E_\pi\ell_{\pi,a}(x)
 \le m-\frac{\delta U_\delta(x)}{8C}.
\]
Fix a coordinate \(i\) counted by \(U_\delta(x)\), and let
\(y=x^{\oplus i}\). Generate the order by independent uniform
ranks \(V_j\in[0,1]\). Conditional on
\(V_i\ge1-\delta/(4C)\), the expected certificate mass at
coordinates later than \(i\) satisfies
\[
 \mathbb E\!\left[\sum_{j:V_j>V_i}w_y(j)\,\middle|\,V_i\right]
 =(1-V_i)\sum_{j\ne i}w_y(j)\le\frac\delta4.
\]
Markov's inequality shows that this later mass is less than
\(\delta/2\) with conditional probability at least \(1/2\).
The rank condition and the mass condition therefore both hold with
probability at least \(\delta/(8C)\).

On this event, the encoder infers bit \(i\). Indeed, an input
\(z\) consistent with the preceding bits and with \(z_i=y_i\)
can differ from \(y\) only at later coordinates. The certificate
at \(y\) gives
\[
 f(z)\ge f(y)-\sum_{j:y_j\ne z_j}w_y(j)
 >f(x)+\frac\delta2\ge a.
\]
Thus no such \(z\) belongs to \(E_a\), and every consistent
input in \(E_a\) has bit \(x_i\). Each coordinate counted by
\(U_\delta(x)\) is omitted with probability at least
\(\delta/(8C)\). Sum these probabilities to obtain the claimed
code-length bound.

\emph{Step 3: Count inputs and average over the level.}
For \(u>0\) and \(a\in[0,1]\), let
\[
 F_a:=\{x:U_\delta(x)\ge u,\ f(x)\le a\le f(x)+\delta/2\}.
\]
The code-length bound, Jensen's inequality, and Kraft's inequality
give
\[
 \begin{aligned}
 |F_a|\,2^{-m+u\delta/(8C)}
 &\le\sum_{x\in F_a}2^{-\mathbb E_\pi\ell_{\pi,a}(x)}\\
 &\le\mathbb E_\pi\sum_{x\in E_a}2^{-\ell_{\pi,a}(x)}\le1.
 \end{aligned}
\]
Hence \(\Pr[X\in F_a]\le2^{-u\delta/(8C)}\) for every
\(a\). Each input with \(U_\delta(x)\ge u>0\) has an upward
jump, so \(f(x)\le1-\delta\). It therefore belongs to \(F_a\)
for an interval of levels of length exactly \(\delta/2\).
Integrate over \(a\in[0,1]\) to obtain
\[
 \frac\delta2\Pr[U_\delta(X)\ge u]
 =\int_0^1\Pr[X\in F_a]\,da
 \le2^{-u\delta/(8C)}.
\]
This proves~\eqref{eq:upward-sensitivity-tail} and the lemma.
\end{proof}

\subsection{Application to a maximum of polynomials}

\begin{proof}[Proof of Lemma~\ref{lem:polynomial-max-sensitivity}]
At each input \(y\), choose an index \(a_y\) such that
\(p_{a_y}(y)=f(y)\). Apply
Lemma~\ref{lem:oracle-polynomial-certificate} to \(p_{a_y}\)
at \(y\), and write its weights as \(w_y(i)\). For every
input \(x\),
\[
 f(y)-f(x)
 \le p_{a_y}(y)-p_{a_y}(x)
 \le\sum_{i:y_i\ne x_i}w_y(i).
\]
These weights form a one-sided certificate for \(f\) at \(y\),
with mass at most \(C:=\pi^2D^2/4\ge1\).
Lemma~\ref{lem:one-sided-certificate-sensitivity} gives
\[
 \Pr[S_\delta(f,X)\ge t]
 \le\frac2\delta\,2^{-t\delta/(16C)}
 =\frac2\delta\,2^{-t\delta/(4\pi^2D^2)}.
\]
\end{proof}

\end{document}